\documentclass[11pt]{article}
\usepackage[margin=1in]{geometry}

 \usepackage{amsmath, amsthm, mathtools, dsfont, mathdots}
\usepackage{newtxtext}
\usepackage{newtxmath}

\renewcommand{\paragraph}[1]{\medskip\noindent\textbf{#1}}

\usepackage{enumitem}
\setlist{nosep,topsep=0pt,leftmargin=*}

\usepackage[hidelinks]{hyperref}

\usepackage{nicefrac}
\usepackage{array}
\usepackage{makecell}

\usepackage{graphicx}
\usepackage{subcaption}
\graphicspath{{Figures/}}

\usepackage{xfrac}
\usepackage{flushend}
\usepackage{multicol}
\usepackage{enumitem}
\usepackage{tikz}
\usepackage{wrapfig} 
\usepackage{worldflags}
\usepackage{multirow}
\usetikzlibrary{tikzmark}
\usepackage{mathtools}
\usepackage{tocloft}

\usepackage{etoc}
\usepackage{titletoc}
\usepackage{lipsum}

\usepackage[noend]{algpseudocode}
\usepackage[ruled]{algorithm2e} 

\usepackage{appendix}
\usepackage{lipsum}
\usepackage{geometry}

\usepackage{amsmath}
\usepackage{mathtools}
\usepackage{amsthm}
\usepackage{xspace}
\usepackage{nicefrac}

\usepackage[numbers]{natbib}

\newtheorem{definition}{\textbf{Definition}}
\newtheorem{assumption}{\textbf{Assumption}}
\newtheorem{remark}{\textbf{Remark}}

\newtheorem{lemma}{\textbf{Lemma}}
\newtheorem{theorem}{\textbf{Theorem}}
\newtheorem{proposition}{\textbf{Proposition}}
\newtheorem{corollary}{\textbf{Corollary}}

\DeclareMathOperator*{\argmax}{arg\,max}

\usepackage{accents}
\newcommand{\ubar}[1]{\underaccent{\bar}{#1}}
\DeclareRobustCommand{\ubar}[1]{\underaccent{\bar}{#1}}

\newcommand{\OACC}{\textsf{OACC}\xspace}
\newcommand{\OMCR}{\textsf{OMCR}\xspace}

\newcommand{\PUM}{\textsf{PUMax-}\xspace}

\newcommand{\OPT}{\texttt{OPT}\xspace}
\newcommand{\ALG}{\texttt{ALG}\xspace}

\newcommand{\CP}{\textsf{CP}\xspace}

\newcommand{\asigmas}{\alpha_{\star}^{(\sigma)}\xspace}
\newcommand{\ataus}{\alpha_{\star}^{(\tau)}\xspace}

\newcommand{\Dstar}{\Delta_{\star}^{(\sigma)}\xspace}
\newcommand{\chistar}{\chi_{\star}^{(\tau)}\xspace}

\def\Id{\,\mathrm{d}} 

\usepackage[export]{adjustbox}

\begin{document}

\title{Beyond a Single Optimal Design: A Dynamical Systems Characterization of Online Allocation with Convex Costs}

\author{
    Haoxin Sang\thanks{University of Alberta. Email: \texttt{hsang1@ualberta.ca}}\\
    \and
    Bo Sun\thanks{University of Ottawa. Email: \texttt{bo.sun@uottawa.ca}}\\
    \and 
    Xiaoqi Tan\thanks{University of Alberta. Email: \texttt{xiaoqi.tan@ualberta.ca}}
}

\maketitle

\begin{abstract}
  We study online allocation with convex costs (\OACC), a generalization of limited-supply models in which additional resources can be produced dynamically at convex cost. Prior work largely focuses on constructing a single algorithm that achieves strong competitive guarantees. In contrast, our main contribution is a structural characterization of the optimal design space of reserve functions (i.e., normalized pricing rules) for \OACC, yielding a family of optimal online algorithms.
  
  Using a principled dynamical-systems approach, we show that optimal online algorithms  arise as solutions of a nonlinear eigenvalue problem: the optimal competitive ratio corresponds to the dominant eigenvalue, and the associated eigenfunction determines the optimal reserve function. This perspective establishes the existence of infinitely many optimal designs attaining the best possible competitive ratio, thereby unifying and generalizing a broad class of prior algorithms. Beyond unification, our characterization uncovers new structural phenomena, including the ability to flexibly mix dynamic and static pricing without loss of optimality and the construction of universally competitive algorithms under unknown or stochastic costs. We further extend the analysis to more structured settings such as supply-oblivious arrivals and hard supply constraints, obtaining sharper guarantees that improve upon previous bounds.
\end{abstract}

\setcounter{tocdepth}{2} 
\tableofcontents

\newpage

\section{Introduction}

Over the past few decades, online resource allocation has emerged as a central research topic at the intersection of computer science, economics, and operations research. The fundamental goal is to allocate \textit{limited} resources to a sequence of arriving requests, each associated with a value or reward, in a manner that optimizes a global objective such as social welfare or revenue. These models have been extended to capture a range of additional considerations, including fairness \cite{Zargari2025,Banerjee2023}, risk sensitivity \cite{Christianson2024}, load balancing \cite{Lin2012,Moharir2013}, and application-specific features such as penalties for under- or over-allocation \cite{Balseiro2025}.

A notable generalization of the limited-supply model allows the system to dynamically produce additional resources, incurring costs modeled as a convex function of the total quantity produced. We refer to this framework as \textit{online allocation with convex costs} (\OACC). The \OACC model captures a broad range of practically relevant applications. For example, in online job scheduling for data centers, convex costs model energy costs~\cite{Zhang2015} or carbon costs~\cite{Alaviyar2025} as a function of computational load. In online network caching, the cost reflects penalties associated with cache misses~\cite{Dehghan2019,Zargari2025}. In online advertising, convex costs capture the increasing difficulty of acquiring additional budget under soft budget constraints \cite{Devanur2012}.

From a theoretical perspective, the design of online algorithms for different variants of \OACC has attracted significant attention over the past decade. Representative works include online convex packing and convering \cite{Azar2016}, online combinatorial auctions \cite{Tan2020a,Huang2019,Blum2011}, online knapsack \cite{Tan2020}, and, more recently, online (adversarial) selection with convex costs \cite{Nekouyan2025,Tan2025}. Despite these advances, the existing literature exhibits two notable limitations: 
\begin{itemize}
    \item \textit{Bounded value assumptions}: Most prior work assumes that request values are drawn from a bounded support (e.g., \cite{Nekouyan2025,Tan2025,Ma2020,Sun2022}), introducing additional parameters that are often difficult to justify in practice and typically lead to parameter-dependent performance guarantees. More critically, both the algorithmic designs and hardness results in these studies heavily rely on this bounded-support assumption and cannot be naturally extended to the more general setting with arbitrary arrivals. The only known exceptions are \cite{Blum2011,Azar2016,Huang2019,Tan2020a}. However, the algorithms in \cite{Blum2011, Azar2016} achieve only suboptimal or order-optimal guarantees, respectively, without precisely characterizing the optimal competitive ratio—particularly in settings involving unknown/stochastic costs and hard supply constraints. In contrast, the results of \cite{Huang2019} and \cite{Tan2020a} are optimal, but their optimality holds only for power cost functions of the form $f(y) = y^{\sigma}$, where $\sigma > 1$ is a constant.  
    
    \item \textit{Absence of a characterized design space}: 
    Nearly all existing work on \OACC focuses on constructing a single (asymptotically) optimal algorithm. However, even when worst-case optimal, a single design is often too rigid for practical applications. Online algorithms for \OACC can typically be interpreted as dynamic pricing schemes that adjust posted prices according to resource utilization. A design tailored to worst-case guarantees may therefore be overly conservative for a particular realization of the instance. This observation motivates the need for a richer design space of dynamic pricing mechanisms that can be adapted using instance-specific predictions. Moreover, in many applications the cost functions may not be known precisely in advance. Addressing such uncertainty requires a structural characterization of the admissible design space for an entire class of cost functions, so that robust algorithms can be derived from their intersection. At present, however, it remains unclear when such a design space exists and how to systematically construct one that contains multiple optimal designs for \OACC.
\end{itemize}

In this paper, we address the above gaps by characterizing the full space of admissible online designs through a principled dynamical-systems approach. Rather than producing a single optimal design, our framework identifies the structural conditions that define all competitive algorithms and characterizes those that are optimal. This structural perspective explains precisely when and why optimality arises. As a consequence, we obtain a family of optimal algorithms as natural solutions of the underlying dynamics,  thereby unifying a range of existing algorithms as special cases and extending the design space to accommodate richer constraints and model variations.

\subsection{Our Contributions and Techniques}
We study \OACC under a general setting, where (i) request values are arbitrary and not restricted by any bounded support, and (ii) supply costs are convex in the total allocation, satisfying mild regularity conditions on their elasticity. Our primary contribution is the design of a family of optimal online algorithms via a principled dynamical systems approach. We show that nearly all existing algorithms of different \OACC instances (e.g., \cite{Tan2020,Tan2020a,Azar2016,Huang2019,Blum2011}) can be unified within a common framework of \textit{pseudo-utility maximization}. The core idea is to design a \textit{pseudo-cost function} that dynamically adjusts the immediate value or reward obtained from each request. After normalization by the marginal supply cost, the pseudo-cost function reduces to the design of a \textit{reserve function} $ \phi $, which captures the extent to which the algorithm must ``look ahead'' to \textit{reserve} an appropriate portion of resources for potentially higher-value future arrivals. Within this unified framework, we establish the existence of infinitely-many optimal reserve functions that attain the same, best-possible competitive ratio. Our results strictly improve the best-known guarantees for online convex packing \cite{Azar2016} and recover the optimal guarantees established for fractional online combinatorial auctions \cite{Tan2020a,Huang2019} as special cases.\footnote{Our formulation focuses on additive (fractional) allocation models, but the underlying framework extends naturally to more general valuation settings, such as bundling in combinatorial auctions.} More precisely, we show that the family of optimal reserve functions fills a rich structure, which we term the \emph{optimal design space}, within which infinitely many optimal designs can be constructed in a flexible manner. The existing algorithms mentioned above correspond to particular solutions within this design space.

Notably, the characterized design space enables progress on several challenging variants of \OACC that were previously out of reach. First, we show that it admits a flexible combination of \textit{dynamic} and \textit{static} pricing schemes while preserving optimal competitive performance. This separation is particularly striking in light of prior results for online selection problems, where limited price adaptivity induces a monotone performance interpolation between fully static and fully dynamic solutions in both adversarial and stochastic settings \cite{sun2024static,Perez2025_prophet_limited,Nekouyan2025_risk_pricing}. In contrast, our results reveal a qualitative distinction between fractional and integral online allocation, suggesting a fundamentally different complexity landscape for online (integral) selection under limited price changes. Second, we extend the design space to \OACC with \textit{unknown} or \textit{stochastic} costs, obtaining a family of algorithms that are \textit{universally competitive} across all cost realizations—a property we term a \emph{universal design space}. Third, we provide a sharper analysis under supply-oblivious arrivals, where requests are indifferent to the serving supplier. This yields improved asymptotic competitive ratios that surpass prior bounds determined by the cost function of the worst supply node. Finally, we show that the optimal design space extends to settings that combine convex supply costs with hard supply constraints, resolving a question that persists since the seminal works of \citet{Blum2011} and \citet{Huang2019}.

From a technical perspective, we show that the design of competitive reserve functions reduces to solving a class of \textit{nonlinear eigenvalue problems} of the form $ \phi' = \alpha \cdot F(\phi, y) $, which, more broadly speaking, fall within the scope of \textit{nonlinear dynamical systems}. A key source of complexity arises from a \textit{singularity issue}, which emerges as a critical condition for ensuring the competitiveness of online algorithms. This introduces additional analytical challenges to an already difficult class of problems. Methodologically, our contributions are twofold. First, we provide a theoretical resolution of the nonlinear dynamical system by fully characterizing its solution space, including a rigorous analysis of the upper- and lower-bound solutions—also known as \textit{extremal solutions} in the theory of differential equations—which represent the most and least aggressive feasible designs that achieve the optimal competitive ratio. Second, we develop principled methods for constructing optimal designs (not necessarily direct solutions to the dynamical system) and propose numerically stable procedures for computing the extremal solutions, together with provable bounds on their approximation error. Our results apply to convex cost functions with bounded elasticity, thereby encompassing a broad class of functions that prior works were unable to handle (e.g., \cite{Tan2020a, Huang2019, Blum2011}).

\subsection{Further Related Work}
We review more work related to \OACC and the techniques we adopt to solve the underlying dynamical systems.

\textit{Online knapsack.} A closely related line of work studies the online (fractional) knapsack problem and its extensions. In this setting, items with values and weights arrive sequentially and must be packed into one or more knapsacks, possibly subject to eligibility constraints. The single-knapsack case with small items was studied by \citet{Chakrabarty2008}, where the small-item assumption effectively reduces the problem to a fractional model and enables the derivation of an optimal algorithm under bounded value-density assumptions. \citet{Sun2020} extended this to multiple knapsacks with full eligibility, while \citet{Tan2020} incorporated convex packing costs in the single-knapsack setting and \citet{Sun2022} further allowed departures. However, due to the strong impossibility results of \citet{Marchetti1995}, nearly all of these works rely on bounded-support or density assumptions on item values. In contrast, our framework allows arbitrary valuations and provides a structural characterization of optimal policies under convex costs, without requiring bounded value support.

\textit{Online matching and allocation.} Online (bipartite) matching is a foundational model capturing the core challenges of many online allocation and market design problems, most prominently in Internet advertising. In this framework, a set of offline agents is given upfront, while agents on the other side arrive sequentially and must be matched immediately and irrevocably. Since the seminal work of \citet{Karp1990}, numerous extensions have been developed, including weighted matching \cite{Aggarwal2011}, the AdWords problem \cite{Mehta2007}, and other budgeted allocation models; see the surveys by \citet{Mehta2013} and \citet{Huang2024} for comprehensive overviews. A fundamental distinction between traditional online matching and \OACC lies in the modeling of supply constraints. Standard matching assumes fixed budgets or capacities per offline node, resulting in hard budget exhaustion. In contrast, \OACC replaces fixed capacities with convex production costs, allowing supply to expand continuously at increasing marginal cost. This shift fundamentally alters the structure of optimal online decisions and competitive analysis. The closest related work is \citet{Devanur2012}, where concave return functions implicitly capture linear value minus convex cost, similar in spirit to our formulation. More recently, \citet{Kalen2026} extend online matching with concave returns to more general (possibly non-separable) objectives. However, these works do not provide a structural characterization of optimal policies under convex costs, nor do they yield optimal competitive guarantees in the general \OACC setting considered here.

\textit{Planar polynomial dynamical systems.} 
Polynomial and power-law dynamical systems are widely studied in biology, chemistry, and engineering \cite{Edelstein-Keshet2005,Donnell2013}. Theoretical work on planar systems focuses on qualitative behavior without explicit integration, which is often intractable. Existing results are either too general to yield system-specific structural insights \cite{Hirsch2013,Teschl2012,Dumortier2006}, or rely on assumptions that do not hold in our setting. For example, \citet{Craciun2013} and \citet{Craciun2019} require weak reversibility, which our system lacks, while \citet{Artes2021} and \citet{Shestopalov2021} focus on quadratic systems distinct from ours. To the best of our knowledge, no prior work provides a qualitative analysis tailored to the planar dynamical system that characterizes the design space of optimal online allocation under convex costs.

\section{Problem Formulation, Assumptions, and Preliminaries}

In this section, we formulate the online allocation with convex costs (\OACC) problem and introduce the technical assumptions and preliminaries underlying our algorithmic framework.

\subsection{Problem Formulation} \label{Sec:Problem Formulation}

We consider an online resource allocation problem over a set of $m$ offline supply nodes. Each supply node $j \in [m]$ incurs a cost $f_j(y): \mathbb{R}_+ \rightarrow \mathbb{R}_+$ when providing a total supply of $y$ units. 
A sequence of $ T $ online demand nodes  (or requests) arrive one at a time. Upon the arrival of demand node $ t $, we need to make an immediate and irrevocable allocation decision $ \mathbf{x}_t := (x_{t,1},\dots,x_{t,m}) \in \mathcal{X}_t \subset \mathbb{R}_+^m $. 
Each allocation $x_{t,j}$ denotes the amount of demand $t$ allocated to supply node $j$, which consumes $w_{t,j} x_{t,j}$ units of supply and generates a reward $v_{t,j} x_{t,j}$. The objective is to maximize the total reward minus the total cost of supply, i.e., 
\begin{equation} \label{eq:oacc}
\hspace{-1cm} (\OACC) \qquad 
\underset{\mathbf{x}_t \in \mathcal{X}_t}{\max}\ \sum_{t\in[T]} \sum_{j\in[m]} v_{t,j} x_{t,j} -   \sum_{j\in[m]} f_j\biggl(\sum_{t\in[T]} w_{t,j} x_{t,j} \biggr),
\end{equation}
where the feasible allocation set $\mathcal{X}_t = \{ \mathbf{x}_t \big| \sum_{j\in [m]} x_{t,j} \leq 1, x_{t,j}\geq 0, \forall j\in [m] \}$. 
A demand node $t$ may only be interested in a subset of supply nodes $ \mathcal{J}_t \subseteq [m] $. This is modeled by setting $v_{t,j} = 0, \forall j\in [m]\setminus \mathcal{J}_t$, which implies that allocating demand $t$ to supply $j$ yields no reward, and enforces $x_{t,j} = 0$. 
Unless otherwise stated, we assume that all $f_j$ are identical and drop the subscript in the baseline analysis.  However, our results can be naturally extended to the case of non-identical  $f_j$'s, with performance characterized in terms of the ``worst"  $f_j $. 

\OACC generalizes several well-studied online allocation problems:
\begin{itemize}
    \item \textit{Online knapsack with convex costs} \cite{Tan2020,Azar2016}. Here, each supply node represents a knapsack, and each online demand node corresponds to an item with value $v_{t,j}$ and weight $w_{t,j}$ for knapsack $j$. The cost function $f$ acts as soft capacity constraints (or production cost of resources), modeling the cost of allocating capacity in each knapsack. When item values and weights are uniform across knapsacks (i.e., $v_{t,j} = v_t$, $w_{t,j} = w_t$), \OACC reduces to the setting studied in online knapsack with convex costs \cite{Tan2020}, which is also aligned with a special instance of the online convex packing problem studied in \citet{Azar2016}. 

    \item \textit{Online combinatorial auctions with production costs} \cite{Huang2019,Tan2020a,Blum2011}. In this setting, each supply node represents a type of resource, and the cost function $f$ is interpreted as the production cost of that resource. Each demand node $t \in [T]$ corresponds to an online agent, where the tuple $\{(v_{t,j}, w_{t,j})\}_{\forall j}$ characterizes the \textit{type} of agent $t$, representing its valuation for different bundles of resources. The \OACC model thus encompasses the fractional version of online combinatorial auctions with production costs \cite{Huang2019,Tan2020a,Blum2011}. As noted in \cite{Huang2019}, extending these fractional allocation results to the integral setting is relatively straightforward and falls outside the scope of this paper.

    \item \textit{Online matching with concave returns (\OMCR)} \cite{Devanur2012}. In this setting, the supply nodes represent offline agents (e.g., advertisers), and each demand node corresponds to an online item (e.g., an Ad impression). Agent $j$ submits a bid $w_{t,j}$ for item $t$, and their overall return is modeled by a non-negative increasing concave function $g_j\left(y_j\right)$, where $y_j = \sum_{t\in [T]} w_{t,j}x_{t,j}$. 
    This can be converted into our \OACC problem as follows. Let $\tilde{v}_{t,j} = g_j'(0)w_{t,j}$ be the reward of agent $j$ for item $t$, and let $\tilde{f}_j(y_j) = g_j'(0)y_j - g_j(y_j)$ be the cost function of agent $j$. Note that $\tilde{f}_j$ is convex since $g_j$ is concave. The objective of \OMCR can be rewritten as
    $ \sum_{t\in[T]} \sum_{j\in[M]} \tilde{v}_{t,j} x_{t,j} - \sum_{j\in[M]} \tilde{f}_j(y_j) $,
    which is exactly the objective of \OACC with the above definitions of $v_{t,j}$ and $f_j$. 
    This indicates that \OACC is a more general model than \OMCR. However, \OMCR focuses on a different class of costs, which is disjoint with the one we consider in \OACC. We provide a detailed comparison between the two settings in Appendix~\ref{Apx:CompConcaveReturn}.
\end{itemize}

In Section \ref{sec:extension}, we further show that, under specific problem structures, our results extend to (i) \OACC with \textit{supply-oblivious arrivals}, where each demand node is interested in the entire set of supply nodes (i.e., $ \mathcal{J}_t = [m]$ for all $ t \in [T]$), and (ii) \OACC with \textit{(hard) supply constraints}, where each supply node has a hard capacity limit $B_j$ and must additionally satisfy the constraint $ \sum_{t\in[T]} w_{t,j} x_{t,j} \leq B_j $.

\subsection{Assumptions} 
Given an arrival instance $ \mathcal{I} = \{ (\mathbf{v}_1, \mathbf{w}_1), (\mathbf{v}_2, \mathbf{w}_2), \cdots, (\mathbf{v}_T, \mathbf{w}_T) \}$, where $ \mathbf{v}_t = \{v_{t,j}\}_{\forall j} $ and $ \mathbf{w}_t = \{w_{t,j}\}_{\forall j}$, let $ \OPT(\mathcal{I}) $ denote the optimal objective value of Problem~\eqref{eq:oacc} in the offline setting, and let $ \ALG(\mathcal{I}) $ denote the objective value achieved by an online algorithm.  Following the standard competitive analysis framework \cite{Borodin2005}, we define the competitive ratio of an online algorithm as  
\begin{equation}\label{definition_of_alpha}
\alpha := \max_{\text{all possible }\mathcal{I}} \frac{\OPT(\mathcal{I})}{\ALG(\mathcal{I})}, 
\end{equation}
where $\alpha \geq 1$, with values closer to 1 indicating better performance.   In the online setting, the supplier does not have prior knowledge of $ \mathcal{I} $ (including the number of arrivals $ T $) and only knows the cost function $f$ upfront. Consequently, the competitive ratio $\alpha$ depends solely on $f$.

In this paper, we focus on a class of convex cost functions whose marginal production cost has bounded \textit{elasticity}. Formally, the elasticity of the marginal cost $f'$, denoted $Ef'$, is defined as 
\begin{equation*} 
  Ef'(y) := \frac{yf''(y)}{f'(y)} \approx \frac{\% \text{ change in } f'(y)}{\% \text{ change in } y}. 
\end{equation*}
In economics, $Ef'$ measures the responsiveness of the marginal production cost to changes in the allocation level $y$: it represents the percentage change in $f'(y)$ induced by a one-percent change in $y$. For instance, suppose allocated resource rises by 1\%. If the elasticity of the marginal cost is 2, then a 1\% increase in allocated resources results in approximately a 2\% increase in the marginal production cost at that level.

Throughout the paper, unless otherwise stated, we impose the following assumption on the cost function.
\begin{assumption}
\label{ass1}
The cost function $f:\mathbb{R}_+ \to \mathbb{R}_+$ is increasing, smooth, and strictly convex with $f(0)=f'(0)=0$. There exist constants $\sigma \geq \tau \geq 1$ such that for all $y>0$, $Ef'(y)$ is non-decreasing and
    \begin{align*}
        \tau - 1 \leq Ef'(y) \leq \sigma - 1.
    \end{align*}
\end{assumption}
We say $f$ is \textbf{$(\tau,\sigma)$-elastic} if it satisfies Assumption \ref{ass1}, and furthermore, $f$ is \textbf{strongly $(\tau,\sigma)$-elastic} if it satisfies the stricter higher-order condition: $\tau - 2 \leq Ef''(y) \leq \sigma - 2$.
Note that $Ef''(y)$ itself is not necessarily non-negative for $y \ge 0$. However, as shown in Lemma \ref{Lem:ElasticComp} in Appendix, a strongly $(\tau,\sigma)$-elastic function guaranties that $Ef'$ remains bounded within $[\tau-1, \sigma-1]$, and thus, $(\tau, \sigma)$-elastic.

We now clarify the assumptions imposed on a $(\tau,\sigma)$-elastic cost function. First, we assume smoothness of $f$ for analytical convenience. All results can be extended to piecewise-smooth cost functions with minor modifications. Throughout the paper, for ease of presentation, we further assume that $\sigma>1$. The degenerate case $\sigma=\tau=1$ is trivial, as simple greedy algorithms are already 1-competitive. In Assumption~\ref{ass1}, we assume that the cost function satisfies $f'(0)=0$. If this is not the case, we overload the notation by redefining $f(y) \leftarrow f(y)-f'(0)y$, which does not affect the analysis. We additionally assume that $f$ is analytic at the origin, meaning that $f(y)$ coincides with its Taylor expansion in a neighborhood of $y=0$. This technical assumption is required to analyze the local behavior of $f$ near the origin.

We next introduce a special subclass of $(\tau,\sigma)$-elastic costs, referred to as $(\tau,\sigma)$-polynomial cost functions.

\begin{definition}[$ (\tau, \sigma) $-Polynomial]\label{def:tau_sigma_convex}
A polynomial function $ f(y)  $ is $ (\tau, \sigma) $-polynomial if it has a min-order $ \tau $  and max-order $ \sigma $ where $\sigma\geq \tau \geq 1 $ are integers, and all coefficients are nonnegative, i.e., 
$$ 
f(y) = \sum_{k=\tau}^{\sigma} c_k y^k, \quad c_k\geq 0 \text{ and } c_{\tau},c_{\sigma}> 0. 
$$  
\end{definition}

By Definition \ref{def:tau_sigma_convex}, a $(\tau,\sigma)$-elastic polynomial is also strongly $(\tau,\sigma)$-elastic.\footnote{For the strong elasticity, see Proposition~\ref{Prop:ts_poly_nondecreasing} in  Appendix. Beyond nonnegative polynomials, specific examples satisfying the strong-elasticity conditions are $y^2\sqrt{1+y}$ and $y^3/\ln(1+y)$, which satisfy the strong-elasticity bounds associated with $(2,2.5)$ and $(2,3)$, respectively.}
\citet{Azar2016} also discussed two similar function classes: one with bounded-elasticity on $Ef$ and the other with non-negative coefficient polynomial functions. 
In later sections, we illustrate the limiting case where $\sigma = +\infty$ (e.g., the exponential function $f(y) = e^y$ can be interpreted as a $(\tau, \sigma)$-polynomial with $\sigma = +\infty$).

\subsection{Preliminaries}
\label{competitive_analysis}

Below, we introduce the key idea of our algorithm based on pseudo-utility maximization, followed by a discussion of an ODE characterization of $ \alpha $-competitive algorithm design.  

\subsubsection{Pseudo-utility maximization.} 
We propose Algorithm~\ref{alg:PRM_MKwC}, dubbed $ \PUM\phi $, as our overarching algorithmic framework. The key idea of $ \PUM\phi $ is as follows. Rather than greedily maximizing the immediate value or reward from each request, $ \PUM\phi $ adopts a \textit{pseudo-cost function} that dynamically adjusts the immediate cost of allocation, guiding the algorithm to allocate resources so as to maximize the resulting pseudo-utility. More specifically, for each resource node $j \in [m]$, we define a \textit{reserve function} $\phi(y): \mathbb{R}_+ \rightarrow \mathbb{R}_+$, which maps the current resource utilization to a potentially larger value.  The marginal cost of resource at utilization $y$ is then given by $f'(\phi(y))$. For example, when $\phi(y) = 2y$, the resource marginal cost is considered as $f'(2y)$, i.e., the marginal cost at twice the current utilization.  The reserve function $ \phi $ determines the extent to which the algorithm ``looks ahead'' and \textit{reserves} resources for potentially higher-value future arrivals.  Intuitively, for a given set of supply cost functions, a steep and rapidly increasing $ \phi $ encourages greater reservation of resources for future requests, whereas a flatter $ \phi $ induces greedier allocations that prioritize current requests.   At the extremes, when $f \equiv 0$, no reservation is needed since accepting all requests is optimal. Conversely, if $f$ grows rapidly, resources must be allocated with extreme caution to remain profitable under high marginal costs. Thus, the performance of $ \PUM\phi $ fundamentally hinges on the careful design of reserve functions for each supply node, and this design is intrinsically shaped by the marginal cost function $f'$. As shown in Algorithm \ref{alg:PRM_MKwC} below, the pseudo-cost term in Eq.~\eqref{eq_agent_t} is indeed computed by integrating the marginal supply cost $f'(\phi(\eta))$ from the current utilization level $y_{t-1,j}$ to the utilization level after allocation to demand node $t$, i.e., $y_{t-1,j} + w_{t,j} x_{t,j}$.

\begin{algorithm}[t]
\caption{Pseudo-utility maximization for \OACC ($ \PUM\phi $)}
\label{alg:PRM_MKwC}

\DontPrintSemicolon
\LinesNumbered

\textbf{Inputs:} $f$ and $\phi$.\;
\textbf{Initialize:} $y_{0,1}=\dots=y_{0,m}=0$.\;

\While{a new demand node $t$ arrives}{
  Solve the following problem for demand node $t$:
  \begin{align}\label{eq_agent_t}
    & \mathbf{x}_t^{\ALG} = \arg\max_{\mathbf{x}_t} \  \sum_{j=1}^m v_{t,j} x_{t,j} - \sum_{j=1}^m \int_{y_{t-1,j}}^{y_{t-1,j}+w_{t,j} x_{t,j}} f'(\phi(\eta)) d \eta.
  \end{align}
  Update the total utilization $y_{t,j} = y_{t-1,j} + w_{t,j} x_{t,j}^{\ALG}, \quad \forall j\in [m]$.\;
}
\end{algorithm}

The concept of pseudo-utility maximization has proven effective across a variety of online optimization problems, including threshold-based algorithms for online knapsack (e.g., \cite{Chakrabarty2008,Tan2020,Sun2020}), posted pricing mechanisms for online combinatorial auctions (e.g., \cite{Blum2011,Huang2019,Tan2020a}), and inventory-balancing algorithms that maximize pseudo-revenue in online personalized assortment problems (e.g., \cite{Golrezaei2014,Ma2020}). In these settings, pseudo-cost functions are typically introduced as problem-specific constructions tailored to particular models. A central contribution of this work is to show that pseudo-utility maximization admits a complete structural characterization in the \OACC\ framework. In particular, we prove that the entire optimal design space of reserve functions for $\PUM\phi$ can be characterized by solving a \textit{nonlinear eigenvalue problem} (NEP) arising from an associated dynamical system.

\subsubsection{ODE characterization of $ \alpha $-competitive reserve functions}

The characterization of $ \alpha $-competitive posted-price mechanisms (in the context of online combinatorial auctions) has been established by many prior works \citep{Huang2019, Tan2020, Tan2020a}. However, the characterization is not stated using reserve functions, and some are given in the form of integral inequalities.
Here we present an ODE characterization corresponding to reserve functions $\phi$ that provides the sufficient and necessary condition for the design of $ \phi $ such that $ \PUM\phi $ is $ \alpha $-competitive. It connects the design of online algorithms for \OACC to the solution of a class of NEPs.

\begin{lemma}\label{unlimited_supply_SufNec}
Let $\alpha \geq 1$, and let $f$ be an increasing, smooth, and strictly convex cost function. An $\alpha$-competitive online algorithm exists if and only if there exists a reserve function $\phi$ satisfying \eqref{eq:SufNec_phi}. In particular, any such $\phi$ ensures that $\PUM\phi$ is $\alpha$-competitive:
\begin{subequations}\label{eq:SufNec_phi}
\begin{align}
& \phi'(y) = \alpha \cdot \frac{f'(\phi(y)) - f'(y)}{\phi(y) \cdot f''(\phi(y))}, & & \forall y \geq 0, \label{eq:mainODE}\\
& \textsf{Monotonicity: } \phi'(y) \geq 0, & & \forall y \geq 0, \label{eq:SufNec_phi_mono}\\
& \textsf{Dominance: } \phi(y) \geq y, & & \forall y \geq 0, \label{eq:SufNec_phi_superlin}\\ 
& \textsf{Initial condition: } \phi(0) = 0. \label{eq:mainODE_IntialValue}	
\end{align}
\end{subequations}
\end{lemma}

For completeness, the full proof of Lemma~\ref{unlimited_supply_SufNec} is provided in Appendix~\ref{Apx:Unlimited}. The lemma characterizes $\alpha$-competitive reserve functions as solutions to the ODE~\eqref{eq:mainODE} subject to three necessary structural conditions:
(i)~\emph{monotonicity}~\eqref{eq:SufNec_phi_mono}, which aligns with the irrevocability of online allocations;
(ii)~\emph{the initial condition}~\eqref{eq:mainODE_IntialValue}, which ensures the algorithm can accept low-value requests when demand is sparse; and
(iii)~\emph{dominance}~\eqref{eq:SufNec_phi_superlin}, which enforces $\phi(y) \geq y$ to prevent underpricing and ensure robustness over an infinite horizon.

To illustrate the derivation of \eqref{eq:mainODE}, consider the single-resource case ($m=1$). Composing the marginal cost function $f'$ with the reserve function $\phi$ defines the \emph{pricing function} $\varPhi(y) = f'(\phi(y))$, which acts as an admission threshold: a request is accepted only if its marginal value exceeds the unit price $\varPhi(y)$ at current utilization $y$. Consequently, $ \PUM\phi $ is $\alpha$-competitive, if the associated pricing function satisfies  
\begin{equation} \label{ineq:PricingFunction}
\int_{0}^{y} \varPhi(u) \mathrm{d}u - f(y) \geq \frac{1}{\alpha} f^{*}\big(\varPhi(y)\big), \quad \forall y \geq 0,
\end{equation} 
where $f^*(p) = \max_{z\geq 0} p z - f(z)$ is the convex conjugate of $f$.

In Eq.~\eqref{ineq:PricingFunction}, the left-hand side lower-bounds the performance of $ \PUM\phi $ upon allocating $y$ units, since allocating $y$ units yields at least $\int_{0}^{y} \varPhi(u)\,\mathrm{d}u$ in value. The right-hand side corresponds to the net profit of allocating all $y$ units at price $\varPhi(y)$, which upper-bounds the offline optimum when $\varPhi(y)$ is the highest marginal revenue achievable for any allocation of up to $y$ units. 
In the proof, we establish sufficiency by showing that the differential inequality $\phi'(y) \leq \alpha \cdot \frac{f'(\phi(y)) - f'(y)}{\phi(y) f''(\phi(y))}$, together with the initial condition $\phi(0) = 0$, implies \eqref{ineq:PricingFunction} and thus guarantees $\alpha$-competitiveness. Conversely, equality in the ODE \eqref{eq:mainODE} is necessary because a worst-case adversary can saturate \eqref{ineq:PricingFunction} across all $y \geq 0$.

\section{Optimal Design Space} \label{Optimal Online Algorithms via A Dynamical Systems Approach}

This section presents our main results. We begin by defining and characterizing the optimal design space for \OACC\ in Sections~\ref{sec_connecting_OACC_NEP}–\ref{sec_characterizing_solution_space}. Building on this characterization, we then discuss two applications: the construction of linear and mixed dynamic–static pricing schemes in Section~\ref{sec:Solution Space}, and the development of universally competitive algorithms for handling unknown or stochastic cost functions in Section~\ref{sec:universal_design_space}.

\subsection{Connecting \OACC to Dynamical Systems}
\label{sec_connecting_OACC_NEP}

From a technical standpoint, Eq. \eqref{eq:mainODE} is not an ODE that can be solved directly, primarily due to its \textit{nonlinearity} and the \textit{singularity} of the initial condition $\phi(0)=0$. Specifically, we express Eq.~\eqref{eq:mainODE} in the form of a \textit{nonlinear eigenvalue problem} (NEP):  $ \phi^{\prime}(y) = \alpha \cdot F(\phi, y) $, where 
\begin{equation} \label{Eq:NEPFunction}
    F(\phi,y) = \frac{f'(\phi(y)) - f'(y)}{\phi(y)\cdot f''(\phi(y))} = \frac{1}{Ef'(\phi)}\cdot \left(1 - \frac{f'(y)}{f'(\phi)}\right).
\end{equation}
This connects our problem to the broader study of NEPs with singular initial conditions \cite{Hirsch2013} and one can see the elasticity naturally appears.
Furthermore, we can observe that the function $F(\phi,y)$ fails to satisfy the Lipschitz condition in $\phi$. Consequently, the classical existence and uniqueness theorem for ODEs (see \cite{Walter1998}, p.~62) does not apply.
An alternative way to view this is to rewrite Eq.~\eqref{eq:SufNec_phi} as a two-dimensional dynamical system (see Appendix~\ref{Apx:Dynamical Systems and First Order ODE} for details):  
\begin{equation}
  \left\{\begin{aligned}
    & \frac{\mathrm{d}\phi}{\mathrm{d}t} = \alpha \left(f'\left(\phi\right)  - f'\left(y\right)\right),\\
    & \frac{\mathrm{d}y}{\mathrm{d}t} = \phi \cdot f''(\phi),
  \end{aligned}\right. \label{eq:MainDynSystem}
\end{equation}
where $y = y(t)$ and $\phi = \phi(t)$ are parameterized by $t \in \mathbb{R}$. The solutions to Eq.~\eqref{eq:SufNec_phi} correspond to the trajectories of this dynamical system that pass through the origin and satisfy the constraints Eqs.~\eqref{eq:SufNec_phi_mono}–\eqref{eq:mainODE_IntialValue}. That is, we need to find trajectories confined in the cone $\{(y,\phi) \mid y\geq 0, \phi\geq y \}$.

For dynamical systems, their trajectories can fill the whole plane. So the difficulty here is to identify whether there are trajectories satisfying our constraints.
Moreover, assuming we can find two trajectories living in the cone, by the nature of dynamical systems, the trajectories ``sandwiched'' by the two existing trajectories must be feasible as well (such trajectories are infinitely many). This inspires us to define a design space, which is constituted by all feasible solutions.

More formally, for any $ \alpha \geq 1 $ and any cost function $ f $ satisfying Assumption \ref{ass1}, we define the \textbf{design space} of $ \phi $, denoted by $ \mathcal{S}\left(\alpha, f\right) $, as the family of all feasible reserve functions satisfying Eq.~\eqref{eq:SufNec_phi}. By a slight abuse of notation, we also use $\mathcal{S}(\alpha, f)$ to refer to the planar region spanned by this design space whenever the distinction is clear from the context, i.e., 
$$
\mathcal{S}(\alpha, f) = \bigcup_{\phi \text{ satisfies \eqref{eq:SufNec_phi}}} \operatorname{graph}(\phi),
$$
where $\operatorname{graph}(\phi) = \{(y, \phi(y)) \mid y \geq 0\} \subset \mathbb{R}_+^2$.
Our goal is to investigate whether the design space $\mathcal{S}(\alpha, f)$ is well-defined and non-empty for various choices of $\alpha \geq 1$ and $(\tau, \sigma)$-elastic functions $f$, and, if so, to rigorously characterize its structure.
This introduces an additional layer of complexity, connecting our problem to the broader study of NEPs with singular initial conditions \cite{Hirsch2013}.

\subsection{Characterizing the Design Space $ \mathcal{S}\left(\alpha, f\right) $: Existence, Boundaries, and Geometric Interpretation}
\label{sec_characterizing_solution_space}

\noindent\textbf{Notations.} Before presenting our main results on the NEP in Lemma~\ref{unlimited_supply_SufNec}, we first introduce several key notations for the subsequent analysis. 
Define
\begin{equation} \label{eq_notations_alpha_Delta}
    \ataus = \tau^{\frac{\tau}{\tau-1}}, \qquad  
    \asigmas = \sigma^{\frac{\sigma}{\sigma-1}}. 
\end{equation}
Both $ \ataus $ and $ \asigmas $ are strictly increasing in $\sigma \geq \tau > 1$. Notably, $\lim_{x \to 1} x^{\frac{x}{x-1}} = e$, and thus $ \asigmas \geq \ataus  > e $ for all $ \sigma \geq \tau > 1 $.

\subsubsection{Existence of non-empty design space $\mathcal{S}\left(\alpha, f\right)$ for $ \alpha \geq \asigmas$} 
We begin by formally stating the lower bound of $ \alpha $ in the theorem below.
\label{sec:existence_design_space}

\begin{theorem}[Lower Bound]\label{theorem_lower_bound}
For any $ (\tau, \sigma) $-elastic $ f $, no online algorithm can achieve a competitive ratio of  $ \asigmas -\epsilon $ for any $ \epsilon > 0 $.
\end{theorem}

The proof of Theorem \ref{theorem_lower_bound} is given in Appendix \ref{Apx:LowerBound}. Among existing results, \citet{Tan2020a} and \citet{Huang2019} obtained the same lower bound, but only for power cost functions (i.e., single term polynomials) with $ \tau =  \sigma $; \citet{Azar2016} obtained an $ O(\frac{\tau}{\tau-1} \sigma) $-competitive algorithm, which is only order-optimal and is strictly worse than $ \asigmas $ for small $ \tau $ and $ \sigma $.  We remark that this lower bound result holds for $(\tau,\sigma)$-elastic costs, where $\tau$ and $\sigma$ can be real numbers. As one of our main results, Theorem \ref{Thm:UpperBoundAch} below establishes that $ \asigmas $ is an \textit{exact} matching lower bound, achievable by $ \PUM\phi$ through infinitely many reserve functions.

\begin{theorem}[Existence of Infinitely-Many $ \asigmas$-Competitive Solutions] \label{Thm:UpperBoundAch}
For any strongly $ (\tau, \sigma) $-elastic $ f $,  there exist infinitely many solutions to the NEP in Eq. \eqref{eq:SufNec_phi} for $\alpha \geq \asigmas$.
Furthermore, 
for any integer $k > 1$ and $\alpha>1$, define the \textbf{characteristic polynomial} $\CP(\alpha,k)$ as
\begin{equation}
z^{k} - \frac{\alpha}{k-1}\left(z^{k-1} - 1\right) , \quad \forall z\in \mathbb{R}. \label{eq:CharacPoly}
\end{equation}
Then these infinitely many solutions are tightly bounded
by linear functions $\chi^{(\tau)}_{-} y$ and $\chi^{(\tau)}_{+} y$, where $ \chi^{(\tau)}_{+} \geq \chi^{(\tau)}_{-} > 1 $ are the two positive roots to $\CP(\alpha,\tau)$.
\end{theorem}

We first briefly discuss the property of the characteristic polynomial defined in Eq. \eqref{eq:CharacPoly}, as it plays a key role in shaping the structure of the design space (for more details, see Lemma~\ref{Lem:PowerEquation} in the Appendix). Define $z_{\star} = k^{1/(k-1)}$ and $\alpha_{\star}(k) = z_{\star}^{k}$. For any $\alpha \geq \alpha_{\star}$, the polynomial $\CP(\alpha,k)$ admits two positive roots, denoted by $z_1$ and $z_2$, such that $ z_1 \geq z_{\star} \geq z_2 > 1$, where $z_1 = z_2 = z_{\star} $ (i.e., a double root) if and only if $\alpha = \alpha_{\star} $. Note that Eq.~\eqref{eq_notations_alpha_Delta} are the $\alpha_{\star}$'s corresponding to $\CP(\alpha,\tau)$ and $\CP(\alpha,\sigma)$, respectively.

The proof of Theorem \ref{Thm:UpperBoundAch} is fairly complicated and we place it in Appendix \ref{Apx:UpperBound}. 
Here we summarize some key insights and implications derived from the proof. 
Theorem \ref{Thm:UpperBoundAch} indicates that for any $ \alpha \geq \asigmas $, 
by the nature of dynamical systems, the infinitely many solutions to Eq.~\eqref{eq:SufNec_phi} never intersect, and all collectively constitute the design space $\mathcal{S}\left(\alpha, f\right)$.
On the other hand, by Theorem~\ref{theorem_lower_bound}, the design space $\mathcal{S}\left(\alpha, f\right)$ is empty for all $ \alpha < \asigmas $. 
The \textbf{optimal design space}, $\mathcal{S}(\asigmas, f)$, therefore comprises infinitely many optimal reserve functions for which $ \PUM\phi $ achieves the optimal $ \asigmas$-competitiveness.

Furthermore, for any $ \alpha \geq \asigmas $, the infinitely many solutions to Eq.~\eqref{eq:SufNec_phi} all begin with an initial slope of either $\chi^{(\tau)}_{+}$ or $\chi^{(\tau)}_{-}$, which are the two positive roots of the characteristic polynomial $\CP(\alpha,\tau)$. We refer to a reserve function as more (respectively, less) \textit{aggressive} if it has a higher (respectively, lower) slope or rate of increase. Theorem~\ref{Thm:UpperBoundAch} then implies that no $\asigmas$-competitive reserve function can be more aggressive than $\chi^{(\tau)}_{+} y$ or less aggressive than $\chi^{(\tau)}_{-} y$ at the onset of allocation.

\subsubsection{Characterizing and computing the boundaries of the design space}
Although Theorem~\ref{Thm:UpperBoundAch} shows that $\mathcal{S}(\alpha, f)$ is contained within the cone defined by the linear functions $\chi^{(\tau)}_{-} y$ and $\chi^{(\tau)}_{+} y$, this enclosure may be loose as $y \to \infty$. We therefore seek a precise characterization of the exact upper and lower boundaries of the design space for all $y \ge 0$, which we formalize through the following definition of the upper- and lower-bound solutions.

\begin{definition} 
\label{Def:UpperLowerBoundSol}
  For any $\alpha\geq \asigmas$, the upper-bound solution $\phi_{\mathrm{ub}}$ is defined as the upper bound of the solutions with initial slope of $\chi^{(\tau)}_{+}$, i.e.,  
  $$
  \phi_{\mathrm{ub}}(y) = \sup \big\{\phi(y) \mid \phi \text{ is a sol. to Eq. \eqref{eq:SufNec_phi}} \text{ and } \phi^{\prime}(0)=\chi^{(\tau)}_{+} \big \}. 
  $$
  In comparison, the lower-bound solution $\phi_{\mathrm{lb}}$ is defined as the lower bound of the solutions with initial slope of $\chi^{(\tau)}_{-}$, i.e., 
  $$ 
  \phi_{\mathrm{lb}}(y) = \inf \big\{\phi(y) \mid \phi \text{ is a sol. to Eq. \eqref{eq:SufNec_phi}} \text{ and } \phi^{\prime}(0)=\chi^{(\tau)}_{-} \big\}.  
  $$
\end{definition} 
 
We remark that the subscripts ``ub" and ``lb'' are abbreviations for ``upper-bound solution" and ``lower-bound solution,'' respectively. 
Note that $\phi_{\mathrm{ub}}$ and $\phi_{\mathrm{lb}}$ characterize the upper and lower bounds of the design space $ \mathcal{S}(\alpha, f)$, and thus are naturally dependent on the target competitive ratio $\alpha$ and the cost function $f$. For any $\alpha\geq \asigmas$, $\phi_{\mathrm{ub}}$ and $\phi_{\mathrm{lb}}$ are well defined, and moreover, by uniform convergence, they are solutions to Eq. \eqref{eq:SufNec_phi} as well (see Appendix \ref{Apx:Discussion on UpperLowerBoundSol} and Appendix \ref{Apx:UpperBound} for more details). Therefore, the upper- and lower-bound solutions $\phi_{\mathrm{ub}}$ and $\phi_{\mathrm{lb}}$ are $\alpha$-competitive and have an initial slope of $\chi^{(\tau)}_{\pm}$, respectively. Proposition~\ref{Prop:BoundResp} further establishes that each boundary solution tightly lies between two linear functions whose slopes are given by the roots of the characteristic polynomial in Eq.~\eqref{eq:CharacPoly}.

\begin{proposition}[\textsc{Tight Linear Bounds for $\phi_{\mathrm{ub}}$ and $\phi_{\mathrm{lb}}$}]
  \label{Prop:BoundResp}
  For any $\alpha\geq \asigmas$, we denote the two roots of $\CP(\alpha,\sigma)$ as $\Delta^{(\sigma)}_{+}$ and $\Delta^{(\tau)}_{-}$. Then, the upper-bound solution $\phi_{\mathrm{ub}}(y)$ is tightly bounded by the linear functions $\chi^{(\tau)}_{+} y$ and $\Delta^{(\sigma)}_{+} y$ for $y\geq 0$; the lower-bound solution $\phi_{\mathrm{lb}}(y)$ is tightly bounded by the linear functions $\chi^{(\tau)}_{-} y$ and $\Delta^{(\sigma)}_{-} y$ for $y\geq 0$.
\end{proposition}

The proof of Proposition \ref{Prop:BoundResp} is provided in Appendix \ref{Apx:BoundResp}.
Recall that $\chi^{(\tau)}_{+} > \Delta^{(\sigma)}_{+}$ denote the larger roots of $\CP(\alpha,\tau)$ and $\CP(\alpha,\sigma)$, respectively. As $\alpha \rightarrow +\infty$, both $\chi^{(\tau)}_{+}$ and $\Delta^{(\sigma)}_{+}$ diverge to $+\infty$, while $\chi^{(\tau)}_{-}$ and $\Delta^{(\sigma)}_{-}$ converge to $1$. Consequently, as $\alpha \rightarrow +\infty$, the design space $\mathcal{S}\left(\alpha,f\right)$ spanned by $\phi_{\mathrm{ub}}$ and $\phi_{\mathrm{lb}}$ expands and eventually encompasses the entire region between $\phi(y) = y$ and the $\phi$-axis. This indicates that even a trivial yet valid reservation design would still yield a competitive ratio $\alpha \rightarrow +\infty$.

It is worth noting that computing the boundary solutions numerically is nontrivial, primarily due to the singularity of the nonlinear eigenvalue problem in Eq. \eqref{eq:SufNec_phi}. In Appendix \ref{sec_numerical_method}, we analyze this difficulty, propose a stable numerical scheme for approximating $\phi_{\mathrm{ub}}$ and $\phi_{\mathrm{lb}}$, and provide error estimates for the approximation.

\begin{figure*}
  \centering
  \begin{subfigure}{0.27\textwidth}
    \includegraphics[width=1\linewidth]{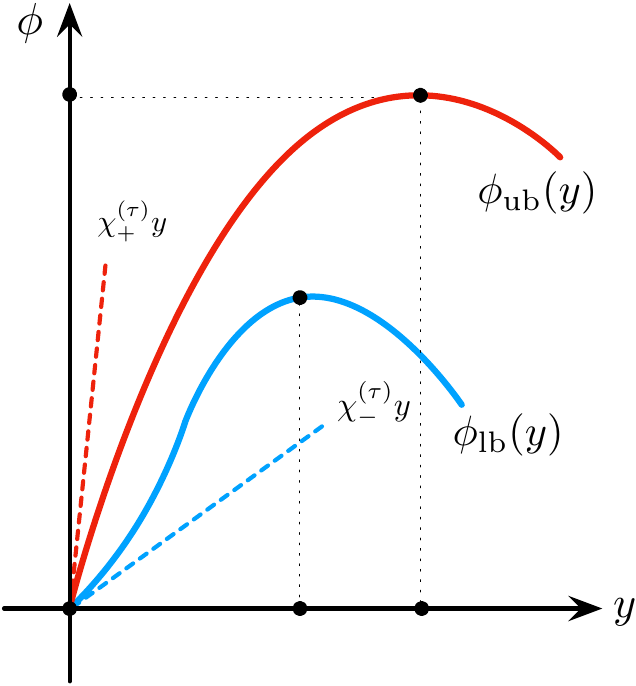}
    \caption{$\ataus\leq\alpha <\asigmas$}
    \label{fig:UpperBoundSol_Lowalpha}
  \end{subfigure}
  \qquad
  \begin{subfigure}{0.27\textwidth}
    \includegraphics[width=1\linewidth]{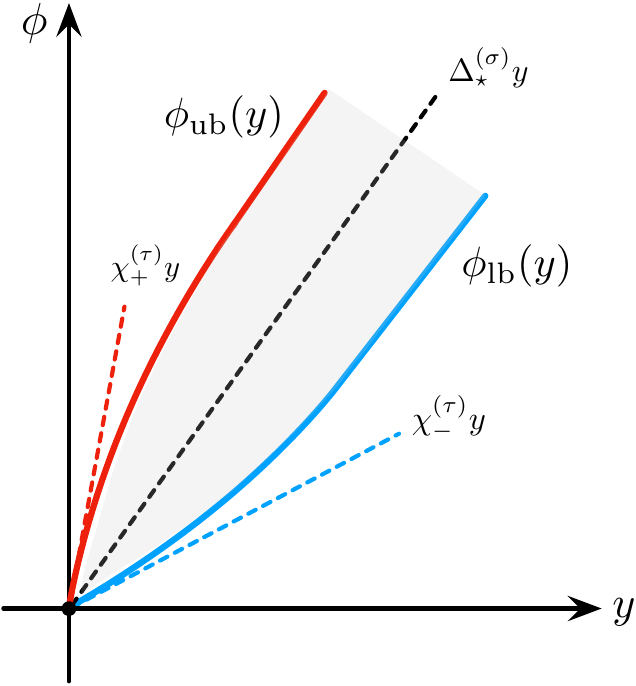}
    \caption{$\alpha =\asigmas$}
    \label{fig:UpperBoundSol_Exactalpha}
  \end{subfigure}
  \qquad
  \begin{subfigure}{0.27\textwidth}
    \includegraphics[width=1\linewidth]{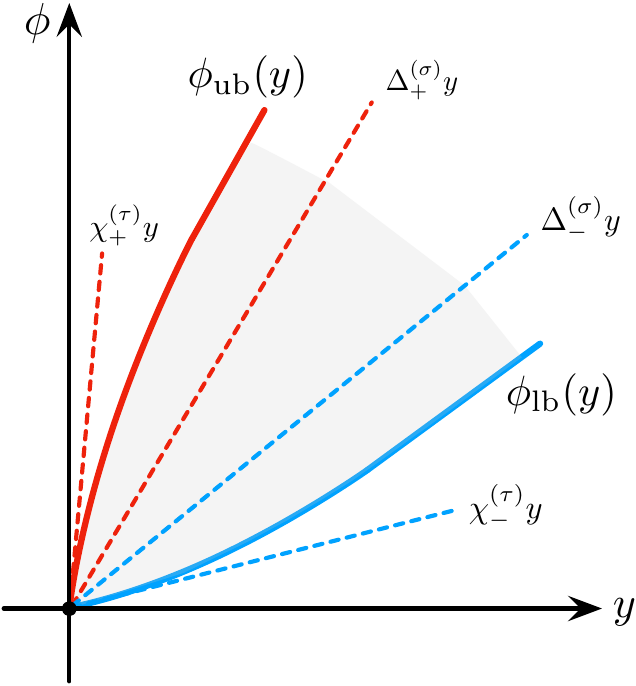}
    \caption{$\alpha >\asigmas$}
    \label{fig:UpperBoundSol_Highalpha}
  \end{subfigure}  
  \caption{Illustration of the design space $ \mathcal{S}(\alpha, f) $ as well as its upper/lower bounds, $ \phi_{\mathrm{ub}} $ and $ \phi_{\mathrm{lb}} $, for $\alpha\in [\ataus, +\infty)$. Infinitely many solutions will fill the design space (i.e., the gray area) if $\alpha\in [\asigmas, +\infty)$, where $ \phi_{\mathrm{ub}} $ and $ \phi_{\mathrm{lb}} $ correspond to the most and least aggressive $\alpha$-competitive reserve functions, respectively.}
\end{figure*}

\subsubsection{Geometric interpretation of the design space}
Figure~\ref{fig:UpperBoundSol_Lowalpha} illustrates that for $\ataus \leq \alpha < \asigmas$, no valid solution exists (noting that $\chi^{(\tau)}_{\pm} y$ coincides when $\alpha = \ataus$, denoted by $\chistar$). In this case, $\phi_{\mathrm{ub}}$ and $\phi_{\mathrm{lb}}$ are feasible within some bounded intervals and violate the constraints Eqs.~\eqref{eq:SufNec_phi_mono}-\eqref{eq:SufNec_phi_superlin} afterwards. 
Figures~\ref{fig:UpperBoundSol_Exactalpha} and \ref{fig:UpperBoundSol_Highalpha} show how $\phi_{\mathrm{ub}}$ and $\phi_{\mathrm{lb}}$ are tightly bounded, and how the design space expands as $\alpha$ increases. This widening illustrates the growing flexibility in designing the reserve function as the target competitive ratio becomes less stringent. Furthermore, increasing $\tau$ in Figure~\ref{fig:UpperBoundSol_Exactalpha} narrows the angle formed by $\chi^{(\tau)}_{\pm} y$, thereby shrinking the optimal design space $ \mathcal{S}(\asigmas, f)$. In fact, when $\tau = \sigma$, the optimal design space $ \mathcal{S}(\asigmas, f) $ collapses into a single line given by $\Dstar y$ (i.e., $ \phi_{\mathrm{ub}}(y) = \phi_{\mathrm{lb}}(y) = \Dstar y $), as illustrated in Figure \ref{fig:UpperBoundSol_Exactalpha}. This occurs because, under $\tau = \sigma$ and $ \alpha = \asigmas $, the two characteristic polynomials $\CP(\alpha,\tau)$ and $\CP(\alpha,\sigma)$ coincide and share the same double root, namely $ \chi^{(\tau)}_{+} = \chi^{(\tau)}_{-} = \chistar = \Delta^{(\sigma)}_{+} = \Delta^{(\sigma)}_{-} = \Dstar $.

We also remark that for any {$(2,\sigma)$-polynomial cost} function, one can employ a linearized system to examine the local behavior of Eq. \eqref{eq:SufNec_phi}, thereby gaining deeper insight into why no feasible solution exists when the competitive ratio $\alpha < \ataus$, as illustrated in Figure \ref{fig:UpperBoundSol_Lowalpha}. The detailed analysis is deferred to Appendix \ref{Linearized System}.

\section{Applications of the Optimal Design Space}
\label{sec:application}

In this section, we introduce three parallel applications of the characterized design space, including constructing new designs, handling uncertainty of cost functions, and applications in more structured setting.

\subsection{Constructing More Reserve Functions based on the Optimal Design Space} \label{sec:Solution Space}

A particularly appealing feature of characterizing the optimal design space $ \mathcal{S}(\asigmas, f) $ is that it enables the construction of infinitely many $\asigmas$-competitive reserve functions beyond the direct solutions to Eq. \eqref{eq:SufNec_phi}. We illustrate this through two representative examples: \textit{linear designs} and \textit{mix-dynamic–static designs}. The former recovers existing optimal designs from prior studies, while the latter offers meaningful insights into the interplay between dynamic and static pricing in economics. These two constructions are grounded in the conditions characterized by the proposition below.

\begin{proposition} \label{Prop:SolSpace}
    For $\alpha \geq \asigmas$, a continuous increasing function $\hat{\phi}(y)$ defined on $[0,+\infty)$ is an $\alpha$-competitive reserve function if it satisfies the following two conditions:
    \begin{itemize}
        \item $\phi_{\mathrm{lb}}(y) \leq \hat{\phi}(y) \leq \phi_{\mathrm{ub}}(y)$ for any $y\geq 0$,
        \item $\hat{\phi}'(y) \leq \alpha \cdot F(\hat{\phi}, y)$ \textbf{almost everywhere} on $[0,+\infty)$.
    \end{itemize}
\end{proposition}

The term ``almost everywhere” in Proposition~\ref{Prop:SolSpace} means that the statement holds except at \textit{countably many} points (or, more generally, on a subset of measure zero). Since $\hat{\phi}(y)$ is continuous and increasing, its derivative $\hat{\phi}'$ exists almost everywhere on $[0, +\infty)$ (see Corollary~3.7 in \citet{Stein2005}). In this case, $\hat{\phi}' \leq \alpha \cdot F(\hat{\phi}, y)$ holds and satisfies the constraints in Eq. \eqref{eq:SufNec_phi}. By the proof of Lemma \ref{unlimited_supply_SufNec}, this condition is sufficient for the solution to be $\alpha$-competitive.

\subsubsection{Linear designs} 
A direct consequence of Proposition~\ref{Prop:SolSpace} is that certain linear designs lie within the optimal design space.

\begin{corollary}[Linear Designs] \label{Cor:UniqueLin}
    For any strongly $ (\tau, \sigma) $-elastic $ f $, the following are true:
    \begin{itemize}
        \item If $ \alpha = \asigmas $, there exists a unique linear reserve function, $ \phi(y) = \Dstar y $, such that $ \PUM\phi $ is $\asigmas$-competitive.

        \item If $ \alpha > \asigmas $, there exists infinitely-many linear reserve functions, bounded above by $ \Delta_{+}^{(\sigma)} y $ and below by $ \Delta_{-}^{(\sigma)} y $, such that $ \PUM\phi $ is $\alpha$-competitive.
    \end{itemize}
\end{corollary}

The proof of Corollary~\ref{Cor:UniqueLin} is encompassed by the proof of Theorem \ref{Thm:UpperBoundAch} (see Appendix \ref{Apx:UpperBound}). Specifically, we show that any solution to Eq.~\eqref{eq:SufNec_phi} can intersect the linear functions described in Corollary~\ref{Cor:UniqueLin} only once and must subsequently deviate beyond them, thereby satisfying the second condition in Proposition~\ref{Prop:SolSpace}. We further observe that when $ \alpha = \asigmas $, the unique optimal linear reserve function $ \phi(y) = \Dstar y $ coincides with the designs proposed in \citet{Tan2020a} and \citet{Huang2019}, both achieving the optimal competitive ratio for power cost functions of the form $ f(y) = y^{\sigma} $ with $\sigma > 1$. Consequently, our results substantially generalize these prior works by extending the optimality of linear designs to the broader class of strongly $ (\tau, \sigma) $-elastic cost functions and by characterizing a richer family of optimal reserve functions.

\subsubsection{Mix-dynamic-static designs} 
Beyond the linear designs, another intriguing construction arises by connecting distinct optimal designs with horizontal segments, leading to what we term the “\textit{mix-dynamic-static designs}” (or simply ``\textit{mix designs}"). Specifically, consider $K$ non-intersecting designs $\phi_1, \dots, \phi_K$ within the optimal design space $\mathcal{S}(\asigmas, f)$, each being $\asigmas$-competitive and ordered in decreasing fashion. We introduce $2K - 1$ turning points $0 = p_0 < p_1 < \dots < p_{2K - 2}$ to connect these $K$ reserve functions, where the turning points with odd indices are selected arbitrarily (or by the designer). The reserve function $\hat{\phi}$ with a mix design is then constructed piecewise such that, for any $k \in [2K - 2]$ and $p_{k-1} \leq y < p_k$:
\begin{equation*} 
\hat{\phi}(y) = \Bigg \{
\begin{aligned}
    &\phi_{\frac{k+1}{2}}(y), \quad &&\text{ if } k \text{ is odd}, \\
    &\phi_{\frac{k}{2}+1}(p_k), \quad &&\text{ if } k \text{ is even}.
\end{aligned}
\end{equation*}
For $y \geq p_{2K-2}$, we set $\hat{\phi}(y)=\phi_K(y)$.

That is, for each $i \in [K]$, the function $\hat{\phi}(y)$ coincides with $\phi_i$ over the interval $[p_{2i-2}, p_{2i-1}]$ and remains constant at the level $\phi_i(p_{2i-1}) = \phi_{i+1}(p_{2i})$ over the interval $[p_{2i-1}, p_{2i}]$. By Proposition~\ref{Prop:SolSpace}, such a construction ensures that $\hat{\phi}$ is also an $\asigmas$-competitive, and hence optimal, reserve function. Figure~\ref{fig:ConnectedSol} illustrates an example of this construction using a mixture of three designs: the upper-bound solution $ \phi_{\mathrm{ub}}$, the unique optimal linear design $ \Dstar y $, and the lower-bound solution $ \phi_{\mathrm{lb}} $.

\begin{wrapfigure}{r}{0.35\textwidth}
  \centering
  \includegraphics[width=0.9\linewidth]{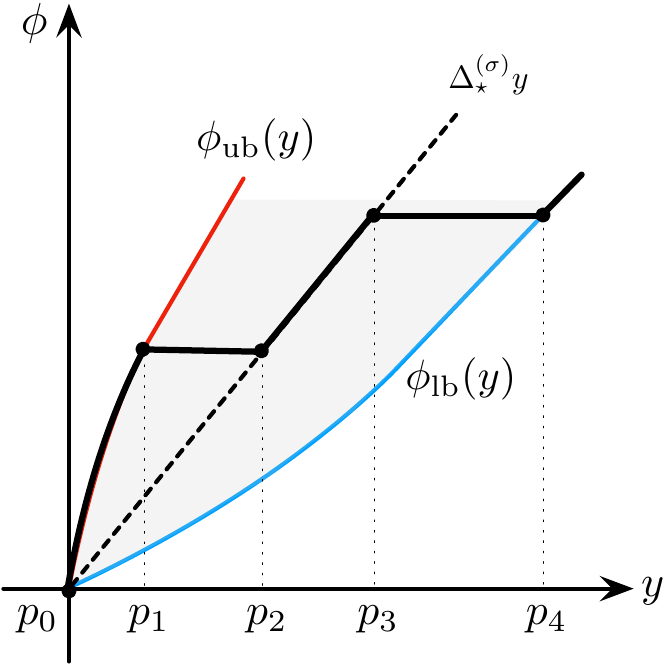} 
  \caption{Constructing an $ \asigmas $-competitive reserve function $ \hat{\phi}(y) $ based on $ \phi_1 = \phi_{\mathrm{ub}} $, $ \phi_2 = \Dstar y $, and $ \phi_3 = \phi_{\mathrm{lb}} $.}
  \label{fig:ConnectedSol}
\end{wrapfigure}
Importantly, the flexibility of the characterized optimal design space enables mixed dynamic–static pricing schemes without requiring the base functions to solve Eq.~\eqref{eq:SufNec_phi}. Any reserve function in $\mathcal{S}(\asigmas, f)$ can serve as a valid building block. In economics, this construction can be viewed as a structured mixture of \textit{dynamic} and \textit{static} pricing. We deploy multiple optimal dynamic pricing strategies and switch among them via intermediate static pricing phases based on the current utilization level. These transitions occur at designated turning points that connect distinct base designs within the optimal design space. Crucially, this modular composition preserves optimal worst-case performance while allowing different pricing behaviors across operational regimes. In contrast to prior results for online selection problems, where limited price adaptivity induces a monotone interpolation between fully static and fully dynamic solutions \cite{sun2024static,Perez2025_prophet_limited,Nekouyan2025_risk_pricing}, our result shows that dynamic and static pricing are not endpoints of a monotone tradeoff in \OACC, but modular components that can be flexibly combined without sacrificing optimality.  In Appendix~\ref{Apx:empirical_evaluation}, we provide empirical evidence demonstrating the practical benefits of this flexibility, particularly in benign environments where instances exhibit rich structural regularities that can be learned.

\subsection{Handling Unknown Costs with a Universal Design Space}
\label{sec:universal_design_space}

Another key advantage of the characterized optimal design space is its ability to handle more realistic and challenging \OACC settings in which cost functions are not fully known in advance.  In this section, we consider two such scenarios: (i) the cost function of each offline node is selected by an adversary from a known function set, and (ii) the cost function of each offline node is randomly sampled from a parameterized function class. Nearly all prior works (e.g., \cite{Huang2019, Azar2016, Devanur2012}) are not applicable to these settings, as they assume full knowledge of the cost functions in advance. In particular, algorithms such as Algorithm~\ref{alg:PRM_MKwC} require complete information about the cost functions. We will show that the characterized design space enables a unified analysis and yields effective algorithms for both scenarios.

\textbf{(Technique overview)}. For simplicity of exposition, we assume that all offline nodes share a common cost function class, denoted by $\mathcal{F}$. Our results extend naturally to heterogeneous function classes, with a similar dependence on the “worst” function set as discussed in Section~\ref{Sec:Problem Formulation}. In this section, our objective is to construct designs that are $\alpha$-competitive for every cost function in $\mathcal{F}$. To this end, we introduce an auxiliary cost function $\hat{f}$ derived from the class $\mathcal{F}$, and show that its $\alpha$-design space, $\mathcal{S}(\alpha, \hat{f})$, is contained in the $\alpha$-design space of every cost function in $\mathcal{F}$. By Proposition~\ref{Prop:SolSpace}, it follows that any $\alpha$-competitive design in $\mathcal{S}(\alpha, \hat{f})$ is also $\alpha$-competitive for all costs in $\mathcal{F}$. We refer to $\hat{f}$ as a \textbf{safe envelope} and to $\mathcal{S}(\alpha, \hat{f})$ as the $\alpha$-\textbf{universal design space} for $\mathcal{F}$.

\subsubsection{$\mathcal{F}$ is a finite set of $(\tau,\sigma)$-elastic functions}
We first consider the case in which the cost function is selected (possibly adversarially) from a finite function class. Fix parameters $\sigma \ge \tau \ge 1$, and let $\mathcal{F} = \{f_k\}_{k \in [K]}$ be a finite collection of $(\tau,\sigma)$-elastic functions. Prior to the start of the \OACC process, each agent $j$ is assigned a cost function drawn from $\mathcal{F}$ (again, possibly adversarially), and this function is unknown to the algorithm.

Since our objective is worst-case competitiveness, a natural approach is to identify the ``worst" cost function in $\mathcal{F}$ and design an $\alpha_{\star}^{(\sigma)}$-competitive policy accordingly. However, this approach fails for two reasons. First, the notion of a “worst” function is not well-defined within $\mathcal{F}$, as all functions share the same $(\tau,\sigma)$-elastic structure. Second, even if one selects a particular function in $\mathcal{F}$ and optimizes against it, the resulting design need not be optimal for the others. For example, consider $f_1(y) = 6y^2 + 20y^3 + y^4$ and $f_2(y) = 18y^2 + 6y^3 + 3y^4$. Their respective optimal design spaces do not contain one another (see Proposition~\ref{Prop:tausigmaPolyUniDS} for a formal statement). Consequently, a design that is optimal for $f_1$ may fail to be optimal for $f_2$.

Although there may be no single function in $\mathcal{F}$ uniformly dominates the others, we can construct a new cost function $\hat f$ that is \emph{pointwise safer} than every function in $\mathcal{F}$, hence the term \emph{safe envelope}. The design space associated with the safe envelope $\hat f$ will then be universally competitive across all possible realizations of $f \in \mathcal{F}$. To formalize this idea, let $k : (0,\infty) \to [K]$ be a \textit{safe selection function} that maps each allocation level $y > 0$ to the index of a function in $\mathcal{F}$.
We define the \textbf{elasticity-based safe selection} (ESS) function $k^* : (0,\infty) \to [K]$ by
$$
k^*(y) = \argmax_{k\in [K]} Ef_k''(y), \quad \forall y>0,
$$
where $Ef'' $ denotes the elasticity of the derivative of the marginal cost of $f'$. Let $k^{*}(y)$ be the smallest index attaining the maximum. Since $\mathcal{F}$ is finite, $k^*$ is piecewise constant and hence piecewise continuous. Intuitively, for each allocation level $y$, we select the cost function whose marginal behavior exhibits the largest elasticity. 
For the previous example, $f_1(y) = 6y^2 + 20y^3 + y^4$ and $f_2(y) = 18y^2 + 6y^3 + 3y^4$, one can compute that $k^*(y) = 1$ for $0\leq y\leq 1$, and $k^*(y)=2$ for $y > 1$.

Denote the switching set as
$$
\mathcal{Y} = \left\{ y>0 : k^{*} \text{ is not constant on any neighbourhood of } y \right\}.
$$
We write $\mathcal{Y} = \{y_1<y_2<\cdots\}$, set $y_0=0$, and denote by $k_s$ the constant value of $k^{*}$ on $(y_s,y_{s+1})$. The following proposition shows that the ESS function $k^*$ induces a safe envelope $\hat{f}$ that remains $(\tau,\sigma)$-elastic.
\begin{proposition}[Constructing a Safe Envelope] \label{Prop:f_hat_finite}
Let $\mathcal{F}=\{f_k\}_{k\in[K]}$, where each $f_k$ is a strongly $(\tau,\sigma)$-elastic cost. Then there exists a function $\hat{f}$ of the form
\begin{align}\label{eq:f_hat_finite}
    \hat{f}(y) = C(y)\, f_{k^{*}(y)}(y) + L(y), \qquad \forall y>0,
\end{align}
where $C$ is positive and piecewise constant and $L$ is piecewise linear, both sharing the switching
points of $k^{*}(y)$, such that:
\begin{enumerate}
    \item[(i)] $\hat{f}(0)=\hat{f}'(0)=0$ and $ \hat{f} $ is strongly $(\tau,\sigma)$-elastic;
    \item[(ii)] for every $f\in\mathcal{F}$,
    $$
    F_{\hat{f}}(\phi,y)\ \leq\ F_{f}(\phi,y), \qquad \forall\, \phi\geq y>0,
    $$
    where $F_g(\phi,y):=\bigl(g'(\phi)-g'(y)\bigr)/\bigl(\phi\, g''(\phi)\bigr)$ as in
    Eq.~\eqref{Eq:NEPFunction}.
\end{enumerate}
Moreover, the pair $(C,L)$ is unique subject to the normalisation $C\equiv 1$ and $L\equiv 0$ on
$[0,y_1]$. 
\end{proposition}

The pair $(C, L)$ is uniquely determined by a sequence of ODEs derived from the ESS function $k^*$ and the function class $\mathcal{F}$. On any interval where $k^*$ is constant, we have $E\hat{f}'' = Ef_{k^*}''$, which yields a differential equation governing $\hat{f}$ on that interval. Because elasticity is invariant under positive scaling, this equation determines $\hat{f}''$ only up to a positive multiplicative constant, which is exactly the role played by $C$; the linear function $L$ then supplies the two remaining degrees of freedom of $\hat{f}$. We glue consecutive segments by requiring $\hat{f}$, $\hat{f}'$, and $\hat{f}''$ to agree at each switching point, and normalize $C\equiv 1$ and $L \equiv 0$ on the first segment. Solving the ODEs on all intervals induced by $k^*$ in this manner produces $C$ and $L$. Matching the second derivatives is what makes $\hat{f}$ a safe envelope: the pointwise inequality $E\hat{f}'' \geq Ef''$ integrates up to the monotonicity of $\hat{f}''/f''$ from which part (ii) follows (see Appendix~\ref{Apx:UniDSProofs} for details).

\begin{theorem}[Universal Design Space for Finite $ \mathcal{F}$] \label{Thm:UniDSElastic}
    Let $\mathcal{F} = \{f_k\}_{k\in [K]}$, where for each $k\in [K]$, $f_k$ is a strongly $(\tau,\sigma)$-elastic cost.
    For any $\alpha\geq \asigmas$, the $\alpha$-design space  of $ \hat{f}$, $\mathcal{S}(\alpha, \hat{f})$, is an $\alpha$-universal design space for $\mathcal{F} = \{f_k\}_{k\in [K]}$.
\end{theorem}

Proposition~\ref{Prop:f_hat_finite} and Theorem~\ref{Thm:UniDSElastic} together yield two important implications. First, since the ESS function $k^*$ is well defined, the induced safe envelope $\hat{f}$ is also well defined and, moreover, is strongly $(\tau,\sigma)$-elastic. It follows that the design space $\mathcal{S}(\alpha,\hat{f})$ is well defined and nonempty, and any design in $\mathcal{S}(\asigmas,\hat{f})$ is $\asigmas$-competitive for every cost function in $\mathcal{F}$. Second, the ESS construction produces an extremal, pointwise-safe cost function associated with the class $\mathcal{F}$. In the special case where $k^*(y)$ is constant for all $y \ge 0$, the safe envelope $\hat{f}$ coincides with the function in $\mathcal{F}$ having the largest elasticity $E_{f'}$. In general, however, $\hat{f}$ need not belong to $\mathcal{F}$; rather, it acts as a safe envelope that guarantees universal competitiveness across all functions in $\mathcal{F}$. We defer the proofs of Proposition~\ref{Prop:f_hat_finite} and Theorem~\ref{Thm:UniDSElastic} to Appendix~\ref{Apx:UniDSProofs}.

\subsubsection{$\mathcal{F}$ is an infinite set of $(\tau,\sigma)$-polynomial functions}
We move on to consider the case where $\mathcal{F}$ consists of $(\tau,\sigma)$-polynomial cost functions parameterized by random coefficients. Formally, let $\mathcal{F}$ be a collection of $(\tau,\sigma)$-polynomial costs such that, for each $k=\tau,\dots,\sigma$, any $f\in\mathcal{F}$ has a (random) coefficient $c_k$ associated with the monomial $y^k$, where $c_k$ is supported on the interval $[\ubar{c}_k,\bar{c}_k]$ with $\ubar{c}_k\ge 0$. We say that a coefficient $c_k$ is degenerate if $\ubar{c}_k=\bar{c}_k$. The coefficients are realized by sampling before the online allocation process begins and are unknown to the online algorithm.

Since $\mathcal{F}$ may be an infinite set, the construction in Proposition \ref{Prop:f_hat_finite} and 
Theorem~\ref{Thm:UniDSElastic} no longer applies. As in the finite-set case, no single cost function $f\in\mathcal{F}$ can be used to derive a universally optimal design. We formalize this impossibility in Proposition~\ref{Prop:tausigmaPolyUniDS}. Specifically, consider a realization of the coefficients $c_k$ for $\tau \le k \le \sigma$. We say that an estimate $\hat{c}_k$ is \textbf{safe} if any design $\phi(y)$ derived from the $(\tau,\sigma)$-polynomial cost function with coefficients $\hat{c}_k$ lies in the $\alpha$-competitive design space of the $(\tau,\sigma)$-polynomial cost function with coefficients $c_k$. In Proposition~\ref{Prop:tausigmaPolyUniDS}, we show that it is impossible to find such safe estimates simultaneously for all coefficients.
\begin{proposition} \label{Prop:tausigmaPolyUniDS}
    Suppose that the supports of the random coefficients of a $(\tau,\sigma)$-polynomial cost function are known and nonnegative. If the supports are degenerate for all $\tau < k < \sigma$, then the estimates $\hat{c}_{\tau}=\ubar{c}_{\tau}$ and $\hat{c}_{\sigma}=\bar{c}_{\sigma}$ constitute a safe estimate. Otherwise, no safe estimates exist.
\end{proposition}

On the other hand, $(\tau,\sigma)$-polynomial cost functions are substantially simpler than general $(\tau,\sigma)$-elastic functions. In particular, to construct a safe envelope $ \hat{f}$, it is not necessary to compare elasticities across all functions in $\mathcal{F}$. Instead, it suffices to consider only finitely many cost functions corresponding to the boundary values of the coefficient supports. Theorem \ref{Thm:tausigmaPolyUniDS} formalizes this result.

\begin{theorem}[Universal Design Space for Infinite $\mathcal{F}$]
\label{Thm:tausigmaPolyUniDS}
Suppose that the supports of the random coefficients of a $(\tau,\sigma)$-polynomial
cost function are known and nonnegative, with $\ubar{c}_{\tau}>0$ and
$\ubar{c}_{\sigma}>0$. For each integer $d$ with $\tau\leq d\leq \sigma-1$, let
$f^{(d)}$ be the $(\tau,\sigma)$-polynomial cost function with coefficients
$$
\hat{c}_k=\ubar{c}_k \quad (\tau\leq k\leq d),
\qquad
\hat{c}_k=\bar{c}_k \quad (d<k\leq \sigma),
$$
and let $k^*$ denote the ESS function induced by
$\mathcal{G}=\left\{f^{(\tau)},f^{(\tau+1)},\dots,f^{(\sigma-1)}\right\}$.
Then the cost function $\hat{f}$ defined by Eq.~\eqref{eq:f_hat_finite} is strongly
$(\tau,\sigma)$-elastic and satisfies that for any $\alpha\geq\asigmas$, the
$\alpha$-design space of $\hat{f}$, $\mathcal{S}(\alpha,\hat{f})$, is an
$\alpha$-universal design space for $\mathcal{F}$.
\end{theorem}

Theorem~\ref{Thm:tausigmaPolyUniDS} shows that even for an (uncountable) infinite
function class, a universal design space can still be constructed. In particular, for
the $(\tau,\sigma)$-polynomial class, the safe envelope can be explicitly characterized
and, somewhat surprisingly, depends only on the $\sigma-\tau$ nested extremal
realizations $f^{(\tau)},\dots,f^{(\sigma-1)}$. The intuition is as follows. By
Proposition~\ref{Prop:tausigmaPolyUniDS}, we obtain safe bounds for the leading
coefficients $c_{\tau}$ and $c_{\sigma}$. The remaining difficulty lies in the
intermediate coefficients: as these vary, the elasticity changes pointwise across $y$.
We observe that overestimating an intermediate coefficient increases the elasticity for small $y>0$ while decreasing it for sufficiently large $y>0$, whereas underestimating it produces the opposite effect. 
However, the crossover between these two regimes occurs at a different location for each degree, so at any given scale the safe choice is to underestimate the low-degree coefficients and overestimate the high-degree ones, with the split point moving through $\tau,\dots,\sigma-1$ as the scale grows. Combining the corresponding $\sigma-\tau$ extremal functions in $\mathcal{F}$ then yields the safe envelope $\hat{f}$. We defer the proofs of Proposition~\ref{Prop:tausigmaPolyUniDS} and Theorem~\ref{Thm:tausigmaPolyUniDS} to Appendix~\ref{Apx:UniDSProofs}.

\subsection{Extensions: \OACC under More Structured Settings}
\label{sec:extension}

In this section, we extend our previous analysis to two structured variants of the \OACC problem. 

\subsubsection{\OACC with Supply-Oblivious Arrivals}
Prior to this section, we consider the standard adversarial model in which, for each arriving demand $t$, the adversary specifies a subset $\mathcal{J}_t \subseteq [m]$ of eligible supply nodes. The resulting competitive analysis is therefore governed by the ``worst'' cost function among the feasible nodes, leading to a dependence on the largest degree parameter $\sigma$ across resources. While this model captures general heterogeneity, it may be overly pessimistic for many practical settings. We now consider a weaker adversarial model, which we term \emph{supply-oblivious arrivals}. In this setting, each demand $t$ is indifferent to the identity of the serving supplier, although suppliers may have heterogeneous cost functions. Formally, we assume $v_{t,j} = v_t$, $w_{t,j} = w_t$, and $\mathcal{J}_t = [m]$ for all $t$ and $j \in [m]$. Thus, the adversary controls only the value $v_t$ and weight $w_t$ of each arrival, while the online algorithm determines how to allocate demand across supply nodes. This model captures resource allocation environments in which customers can be served by any supplier (e.g., fully connected markets), while suppliers may differ in production scale or technology. Our main result in this setting shows that the algorithm $\PUM\phi$ benefits from the weaker adversary: in the large-market regime, it asymptotically achieves a competitive ratio equal to the average of the node-wise optimal ratios, rather than being constrained by the worst cost function.

To demonstrate how $ \PUM\phi $ benefits from supply-oblivious arrivals, we establish the characterization of $ \alpha $-competitive reserve functions for this setting, which is similar to Lemma \ref{unlimited_supply_SufNec} and we skip the proof here. 

\begin{lemma}[Sufficiency]\label{sufficiency_unlimited_supply_fullycon} 
For any $ \alpha \geq 1 $, $ \PUM\phi $ is $ \alpha $-competitive if there exist $ \phi_j $, $j\in [m]$:
\begin{subequations}\label{eq:MKwC1_sufficiency_Phi}
\begin{align}
& \sum_{j=1}^{m}\phi_j'(y_j) \leq \alpha\cdot \sum_{j=1}^{m}\frac{f_j'(\phi_j(y_j))  - f_j'(y_j) }{\phi_j(y_j) \cdot f_j''(\phi_j(y_j))}, & & \forall y_j \geq  0,\ j\in [m],\label{eq:eq:MKwC1_sufficiency_Phi_mainODE}\\
& \textsf{Monotone: } \phi_j'(y_j) \geq 0, & & \forall y_j \geq  0,\ j\in [m],\\
& \textsf{Superlinear: } \phi_j(y_j) \geq  y_j, & & \forall y_j \geq  0,\ j\in [m],\\
& \textsf{Boundary condition: } \phi_j(0) =  0. & & \forall j\in [m].	
\end{align}
\end{subequations}
\end{lemma}
Similarly, we write the equality version of Eq. \eqref{eq:eq:MKwC1_sufficiency_Phi_mainODE} as $\sum_{j}\phi_j^{\prime} = \alpha \cdot \sum_{j} F_j(\phi_j, y_j)$, where 
$$
F_j(\phi_j,y_j) = \frac{f_j'(\phi_j(y_j))  - f_j'(y_j) }{\phi_j(y_j)\cdot f_j''(\phi_j(y_j))}.
$$

By Lemma~\ref{sufficiency_unlimited_supply_fullycon}, the performance of the algorithm~$\PUM\phi$ is determined by the reserve function chosen for each supply node. To build intuition, consider the case of two supply nodes. If we independently design an $\alpha_\star^{(\sigma_1)}$-competitive reserve function for one node and an $\alpha_\star^{(\sigma_2)}$-competitive reserve function for the other, one might expect the combined system to achieve a competitive ratio lying between these two quantities. The following theorem formalizes this intuition and shows that, in fact, an averaging phenomenon emerges asymptotically.

\begin{theorem}[Asymptotic Averaging Convergence] \label{Thm:AverageCRforMK}
    For each $j \in [m]$, let $f_j$ be a $(\tau_j,\sigma_j)$-polynomial cost function. If each $\phi_j$ in Algorithm~\ref{alg:PRM_MKwC} is selected from the optimal design space  $\mathcal{S}(\alpha_{\star}^{(\sigma_j)}, f_j)$, then Algorithm~\ref{alg:PRM_MKwC} asymptotically achieves competitive ratio
    \begin{align*}
        \frac{1}{m} \sum_{j\in [m]} \alpha_{\star}^{(\sigma_j)}
    \end{align*}
    in the large-market regime, where both the maximum valuation  $\max_t \{v_t\} $ and the total demand $\sum_t w_t$ grow unbounded.
\end{theorem}

Theorem~\ref{Thm:AverageCRforMK} reveals a qualitative distinction between the supply-oblivious and standard \OACC\ settings. In the standard model, performance is governed by the worst supply node, yielding competitive ratio $\max_j \alpha_{\star}^{(\sigma_j)}$. In contrast, under supply-oblivious arrivals, the fully connected structure eliminates bottleneck effects in the large-market regime: the competitive ratio converges to the average of the node-wise optimal ratios, thereby strictly improving the worst-node bound.  The key to Theorem~\ref{Thm:AverageCRforMK} is to establish a monotonicity property of $F_j $ across supply nodes, which enables the application of Chebyshev's sum inequality \cite{Hardy1991} to derive a sharper aggregate bound. The detailed proof is provided in Appendix~\ref{Apx:AverageCRforMK}.

\subsubsection{\OACC with (Hard) Supply Constraints} \label{(Hard) Supply Constraints}

We proceed to the second variant of \OACC with \textit{(hard) supply constraints} and demonstrate how the upper-bound solution $\phi_{\mathrm{ub}}$ can be leveraged to strengthen the result in \citet{Tan2020a}, by transforming the existence result (see Definition 3 in \citet{Tan2020a}) into an explicit expression derived from $\phi_{\mathrm{ub}}$.
For clarity of exposition, we focus on a special instance of \OACC, given by:  
\begin{equation}
\label{eq:MainPrimal}
\underset{x_t \in [0,1]}{\max}\ \sum_{t=1}^T v_t x_t \;-\; f\left(\sum_{t=1}^T w_t x_t\right) \qquad \text{s.t.}\ \sum_{t=1}^T w_t x_t \leq B.
\end{equation}  
Hereinafter, we assume w.l.o.g. that the total supply limit is normalized (i.e., $B = 1$). In this setting, we are also given the highest value among upcoming requests.\footnote{Without this information, the problem becomes trivial, as no online algorithm can be competitive with a bounded competitive ratio.}  
Let $v_{\rm max} = \max\{v_t\}_{\forall t}$ and define  
$$ 
\rho = f'^{-1}(v_{\max}), \quad  \text{or equivalently,} \quad f'(\rho) = v_{\max}. 
$$
Here, $ \rho $ represents the utilization level at which the marginal supply cost equals the maximum value $ v_{\rm max} $. 

Similar to Lemma~\ref{unlimited_supply_SufNec} and Lemma~\ref{sufficiency_unlimited_supply_fullycon}, we can derive the following necessary and sufficient condition on $ \PUM\phi $ being $ \alpha $-competitive (for the proof, see Appendix~\ref{Apx:Limited_supply_SufNec}). 

\begin{lemma}\label{Limited_supply_SufNec}
For $ \alpha \geq 1 $ and $\rho>0$, let $b = \min\{\rho,1\}.$ 
An $\alpha$-competitive online algorithm exists if and only if there exists a reserve function $\phi$ satisfying \eqref{eq:SufNec_phi_limited}. In particular, any such $\phi$ ensures that $\PUM\phi$ is $\alpha$-competitive:
\begin{subequations} \label{eq:SufNec_phi_limited}
\begin{align}\normalfont
& \phi'(y) = \alpha\cdot \frac{f'(\phi(y))  - f'(y) }{\min\{\phi(y),1\}\cdot f''(\phi(y))}, & & \forall y \in [0, b], \label{eq:eq:SufNec_phi_limited_ODE}\\
& \textsf{Monotone: } \phi'(y) \geq 0, & & \forall y \in [0, b],\\
& \textsf{Superlinear: } \phi(y) \geq  y, & & \forall y \in [0, b],\\
& \textsf{Boundary conditions: } \phi(0) =  0, \phi(b) \geq \rho.	\label{eq:SufNec_phi_limited_BoundaryCond}
\end{align}
\end{subequations}
\end{lemma}

The optimal competitive ratio established in \citet{Tan2020a} is defined as the smallest value of $\alpha$ for which Eq.~\eqref{eq:SufNec_phi_limited} admits at least one solution. However, \citet{Tan2020a} did not provide a formal analysis regarding the existence or computation of such solutions. Unlike the setting without hard supply constraints, the existence of solutions to Eq.~\eqref{eq:SufNec_phi_limited} depends not only on the competitive ratio $\alpha$ but also on the value of $\rho$. As a result, the associated design space for Eq.~\eqref{eq:SufNec_phi_limited} is fundamentally different from the one characterized in Section~\ref{sec_characterizing_solution_space}. Nevertheless, in what follows, we demonstrate that Eq.~\eqref{eq:SufNec_phi_limited} can be solved by extending the upper-bound solution $\phi_{\mathrm{ub}}$ from Eq.~\eqref{eq:SufNec_phi}.

\paragraph{Core idea: extending $ \phi_{\mathrm{ub}} $ from Eq. \eqref{eq:SufNec_phi} to solve Eq. \eqref{eq:SufNec_phi_limited}.} Let $\alpha = \asigmas$, and consider the question of for which values of $\rho$ Eq.~\eqref{eq:SufNec_phi_limited} admits at least one feasible solution. If we relax the boundary condition in Eq.~\eqref{eq:SufNec_phi_limited_BoundaryCond} to the initial condition $\phi(0) = 0$, we observe that Eq.~\eqref{eq:SufNec_phi_limited} becomes identical to Eq.~\eqref{eq:SufNec_phi} for values of $\phi(y) \leq 1$. This implies that for small $\rho$, any solution to Eq.~\eqref{eq:SufNec_phi} can be extended to a solution of Eq.~\eqref{eq:SufNec_phi_limited} (see Appendix~\ref{Apx:ExtendSol} for construction details). By Theorem~\ref{Thm:UpperBoundAch}, there exist infinitely many such extensions. Among them, we can identify the largest $\rho$ for which the extended solution satisfies $\phi(b) \geq \rho$. Notably, if a solution $\phi_1$ lies above another solution $\phi_2$, then the corresponding extension of $\phi_1$ will also lie above that of $\phi_2$. This implies that the upper-bound solution $\phi_{\mathrm{ub}}$ to Eq.~\eqref{eq:SufNec_phi} determines the largest value of $\rho$ for which Eq.~\eqref{eq:SufNec_phi_limited} admits a feasible solution. A similar analysis applies for any $\alpha \geq \ataus$. In this case, we define a solution to Eq.~\eqref{eq:SufNec_phi_limited} as an upper-bound solution if it is obtained by extending the upper-bound solution to Eq.~\eqref{eq:SufNec_phi}. This extended upper-bound solution characterizes the conditions of the largest $ \rho $ under which Eq.~\eqref{eq:SufNec_phi_limited} is feasible. The following definition formalizes this idea.

\begin{definition}[Maximal Utilization Level]\label{Def:UpperBoundInput}
  Given a $(\tau,\sigma)$-polynomial cost function $ f $. For any $\alpha\geq 1$, let $P:[\ataus,+\infty)\rightarrow (0,+\infty]$ be defined as follows: 
  $$
  P(\alpha) = \sup_{y\geq 0} \phi_{\mathrm{ub}}(y,\alpha), \qquad \forall  \alpha\geq \ataus,
  $$
  where $\phi_{\mathrm{ub}}(y,\alpha)$ is the upper-bound solution to Eq. \eqref{eq:SufNec_phi_limited} corresponding to parameter $\alpha$. 
\end{definition}

Clearly, $P(\alpha)$ is increasing with respect to $\alpha$ as the upper-bound solution $\phi_{\mathrm{ub}}(y,\alpha_1) > \phi_{\mathrm{ub}}(y,\alpha_2)$ for $\alpha_1 > \alpha_2 \geq \ataus$ (this can be obtained by analyzing the upper-bound solutions to Eq. \eqref{eq:SufNec_phi}). 
With the definition of $P(\alpha)$, we can now present the following theorem for the existence and construction of the optimal reserve function to Eq. \eqref{eq:SufNec_phi_limited}.

\begin{theorem}[\textsc{Design Space for Eq. \eqref{eq:SufNec_phi_limited}}] \label{Thm:ExistenceSol_Limited}
    For any $\alpha\geq \ataus$ and $\rho>0$, we have
    \begin{itemize}
        \item Case-I: $\rho > P(\alpha)$. There exists no solution to Eq. \eqref{eq:SufNec_phi_limited} , namely, there exists no $ \alpha $-competitive reserve function if $ \rho $ is strictly larger than the maximal utilization level $ P(\alpha) $.
        \item Case-II: $\rho = P(\alpha)$. There exists a unique solution to Eq. \eqref{eq:SufNec_phi_limited} \textbf{extended} from $ \phi_{\mathrm{ub}}$ of Eq. \eqref{eq:SufNec_phi}. 
        \item Case-III: $\rho < P(\alpha)$. There exist infinitely many solutions to Eq. \eqref{eq:SufNec_phi_limited}.
    \end{itemize} 
\end{theorem}

The proof of Theorem~\ref{Thm:ExistenceSol_Limited} is provided in Appendix~\ref{Apx:ExistenceSol_Limited}. In Case-II, the extension of the upper-bound solution follows the construction in Appendix~\ref{Apx:ExtendSol}. A direct implication of Theorem~\ref{Thm:ExistenceSol_Limited} is that, for each $\rho > 0$, we can identify the best achievable competitive ratio: 
\begin{equation} \label{eq:BestCRforrho}
    \alpha_{\star}(\rho) = \left\{
    \begin{aligned}
        &\ataus, \quad &&\text{ if } \rho \leq P(\ataus), \\
        &P^{-1}(\rho), \quad &&\text{ if } \rho > P(\ataus).
    \end{aligned}
    \right. 
\end{equation}
Together with Lemma~\ref{Limited_supply_SufNec}, this result implies that $ \PUM\phi $, when equipped with the optimal reserve function $ \phi $ specified by Case-II in Theorem~\ref{Thm:ExistenceSol_Limited}, is $ \alpha_{\star}(\rho) $-competitive and therefore optimal.

\begin{wrapfigure}{R}{0.6\textwidth}
  \centering
  \begin{subfigure}[c]{0.28\textwidth}
    \includegraphics[width=0.975\linewidth]{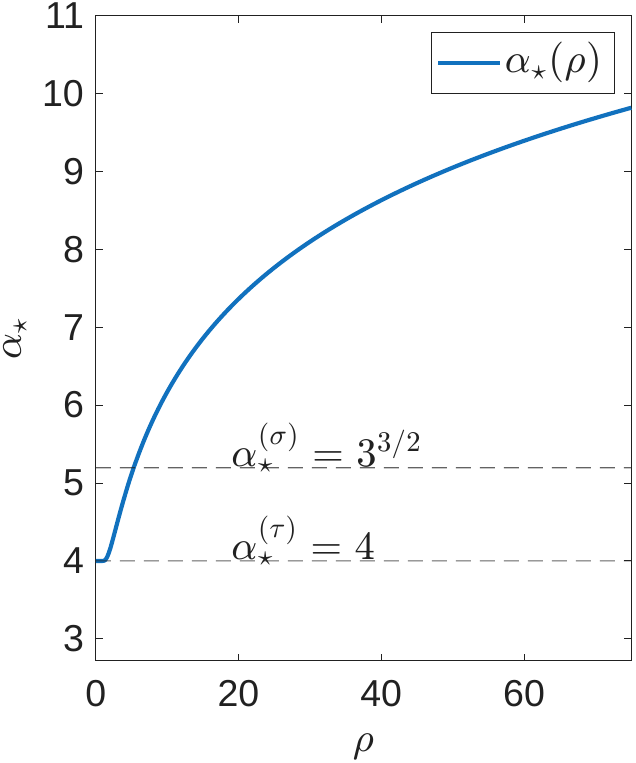}
    \caption{\footnotesize With Supply Constraints}
    \label{fig:23rhoalpha_limited}
  \end{subfigure}
  \hfill
  \begin{subfigure}[c]{0.28\textwidth}
    \includegraphics[width=\linewidth]{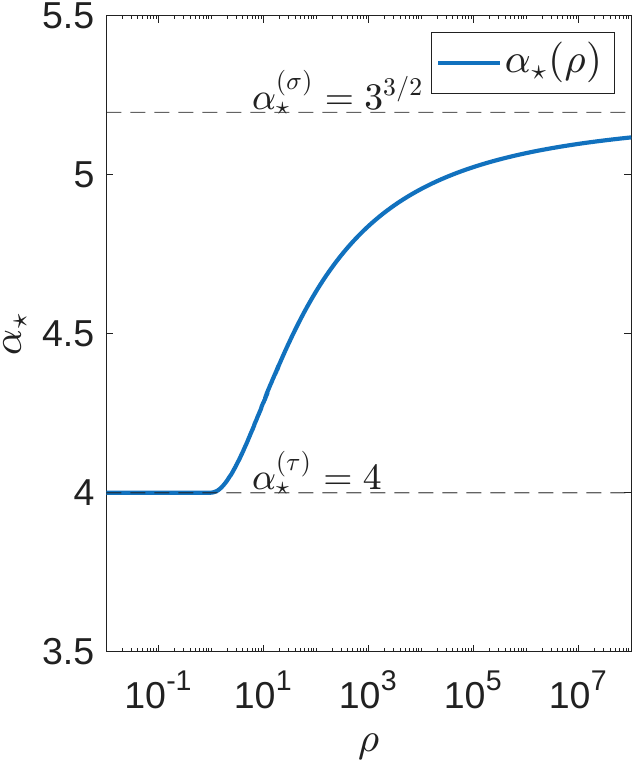}
    \caption{\footnotesize  Without Supply Constraints}
    \label{fig:23rhoalpha}
  \end{subfigure}  
  \caption{$\alpha_{\star}(\rho)$ for input $\rho>0$ with $ f(y) = y^3 + y^2 $.}
\end{wrapfigure} 
\paragraph{Case study: $f(y) = y^3 + y^2$.} 
See Figure~\ref{fig:23rhoalpha_limited} for an illustration of $ \alpha_{\star}(\rho) $ when the cost function is $f(y) = y^3 + y^2$. As shown, $\alpha_{\star}(\rho)$ grows unboundedly with $\rho$, which aligns with intuition: in the limited supply setting, as the maximum value $v_{\rm max}$—and thus $\rho$—increases, the uncertainty faced by the online algorithm also increases, making it increasingly difficult to compete with \OPT. Motivated by this case study, we also examine the role of additional information about $\rho$ in the context of \OACC. Specifically, Eq.~\eqref{eq:oacc} can be interpreted as a special case of the limited supply setting with $B \rightarrow +\infty$, where, by normalization, all input values tend to zero. In this sense, having knowledge of $\rho$ in the unlimited supply setting effectively corresponds to a limited supply instance with a small $\rho$. We show that such knowledge enables the design of improved reserve functions, as the online algorithm no longer needs to compete with \OPT indefinitely—mirroring the behavior in the limited supply case. See Appendix~\ref{Apx: Impact of Addition Information} for details. Figure~\ref{fig:23rhoalpha} illustrates how prior knowledge of $\rho$ at various levels can yield competitive ratios strictly better than $\asigmas$.

\paragraph{Implications for cost functions with $ \sigma = \infty $.} Before concluding this section, we briefly discuss the case where the $(\tau,\sigma)$-polynomial cost function has an unbounded maximum degree, i.e., $\sigma = +\infty$, such as $f(y) = e^y = \sum_{k=0}^{\infty} \frac{y^k}{k!}$. Clearly, based on our lower bound result in Theorem~\ref{theorem_lower_bound}, no algorithm with a bounded competitive ratio exists when the maximum degree of $f(y)$ is unbounded, as in the case of $f(y) = e^y$ where $\sigma = \infty$.  Nevertheless, the parameter-dependent optimal competitive ratio $ \alpha_{\star}(\rho) $ remains well-defined even if $ \sigma = + \infty $. Here, by "parameter-dependent competitive ratio," we mean that the competitive ratio explicitly depends on the additional parameter $\rho$. Thus, we can derive $ \alpha_{\star}(\rho) $ as defined in Eq.~\eqref{eq:BestCRforrho} and conclude that $ \PUM\phi $ is $ \alpha_{\star}(\rho) $-competitive for \OACC with an exponential cost function and maximum value $v_{\max} = f'(\rho)$.

\section{Conclusions and Future Work}

In this paper, we studied online allocation with convex costs (\OACC), extending limited-supply models to settings where resources can be produced dynamically at increasing marginal cost. Our main contribution is a structural understanding of the space of optimal reserve functions, revealing that optimality is not confined to a single algorithm but arises from a rich and explicitly characterizable design space. This viewpoint not only unifies a wide range of prior algorithms, but also shows that infinitely many distinct designs achieve the same best-possible competitive guarantee. Beyond unification, this structural perspective exposes new phenomena. In particular, it enables flexible combinations of dynamic and static pricing without sacrificing optimality and supports the construction of robust algorithms that remain competitive under cost uncertainty. 

Several intriguing directions remain open. First, it is unclear whether the identified design space is complete—that is, whether every $\alpha$-competitive reserve function must lie within it. A similar question arises for the completeness of the universal design space. Resolving these questions would further deepen our understanding of the theoretical structure underlying \OACC and its variants. Second, a systematic investigation of more expressive yet feasible designs within the optimal space may uncover novel forms of phase-adaptive behavior and richer structural features beyond current mix-type constructions, particularly those involving mixtures of multiple noncontiguous phases not linked by plateau regions (i.e., discontinuous price jumps). Finally, our work highlights the promise of integrating nonlinear dynamical systems theory with online algorithm design to achieve structural characterizations of optimal policies. We anticipate that this perspective will extend beyond \OACC and contribute to further progress in broader classes of online optimization problems.

\section*{Acknowledgment}
Xiaoqi Tan acknowledges support from Alberta Machine Intelligence Institute (Amii), Alberta Major Innovation Fund, and NSERC Discovery Grant RGPIN-2022-03646.
Bo Sun acknowledges support from NSERC Discovery Grant RGPIN-2026-06424.

\bibliographystyle{ACM-Reference-Format}
\bibliography{25-OACC}

\appendix

\section{Proofs for Characterization of $\alpha$-competitive Reserve Functions}

\subsection{Preliminaries}

\subsubsection{Online primal-dual framework.} 
The online primal-dual (OPD) framework has proved to be extremely useful for a wide variety of online optimization problems \cite{Buchbinder2009}, and the proof of the sufficiency part of Lemma \ref{unlimited_supply_SufNec} is also based on the OPD approach.
For any instance of the \OACC problem, we define the ``offline'' (fractional) version of the problem as the primal program, which is captured by Eq. \eqref{eq:oacc}. Meanwhile, we consider the corresponding Fenchel dual program of Eq. \eqref{eq:oacc}, which provides an upper bound on any feasible solution to the instance.
\begin{equation}
  \def\arraystretch{1.5}
  \begin{array}{l|l}
  \multicolumn{1}{c|}{\text{ Primal }} & \multicolumn{1}{c}{\text{ Dual }} \\ \hline
  \text{maximize } \sum\limits_{t=1}^T \sum\limits_{j=1}^m v_{t,j} x_{t,j} - \sum\limits_{j=1}^m f(y_j) &  \text{ minimize }\ \sum\limits_{j=1}^m f^*(\lambda_j) + \sum\limits_{t=1}^T \mu_t \\
  \text{subject to: } & \text{ subject to: } \\
  \quad \forall j \in [m]: \quad \sum\limits_{t=1}^T w_{t,j} x_{t,j} \leq y_j  & \quad \forall t \in [T]: \quad \mu_t\geq v_{t,j}-w_{t,j}\lambda_j \\
  \quad \forall t \in [T]: \quad 0 \leq \sum\limits_{j=1}^m x_{t,j} \leq 1 & \quad \forall j \in [m]: \quad \lambda_j \geq 0 \\
    & \quad \forall t \in [T]: \quad \mu_t\geq 0
  \end{array}\label{MKwC1_PrimalDualProblem}
\end{equation}
where $f^*$ is given by: 
$$
f^*(\lambda_j) = \max_{y_j\in \mathbb{R}_+} \lambda_j y_j - f(y_j). 
$$
Note that $ f^*(y) $ is convex, monotone, and differentiable, and $ f^*(0) = 0 $. 
In this paper, we design the reserve function $\phi$ by solving an ordinary differential equation (ODE) derived from the online primal-dual analysis.

\subsubsection{Properties of convex conjugate $ f^* $}
Throughout the whole section, we assume w.l.o.g. that $w_{t,j} = 1$ holds for any $t\in [T]$ and $j\in [m]$ (for general weights $w_{t,j}$, one can simply substitute $v_{t,j}$ by $v_{t,j}/w_{t,j}$ in the proofs and all arguments remain the same).
For \OACC with hard constraints, we use a barrier function to take care of the supply limit. For a $(\tau,\sigma)$-elastic cost $f$, where $\tau>1$, define
$$
\bar{f}(y) = 
\begin{cases}
f(y) & \text{ if } y \in [0, 1]\\
+\infty & \text{ if } y \in (1, +\infty).
\end{cases}
$$
We then show some properties of the convex conjugate function of the cost function $\bar{f}$, given by$\bar{f}^*(\lambda) = \max\limits_{y\in \mathbb{R}_+} \lambda y - \bar{f}(y)$.
\begin{lemma} \label{Lem:ModifiedCost}
    The convex conjugate of the cost function $ f^*(\lambda) $ has the following properties
    \begin{itemize}
        \item $ \bar{f}^*(\lambda) $ is convex, differentiable, and monotone, and $ f^*(0) = 0 $. 
        
        \item The derivative of $ \bar{f}^*{'}(\lambda) $ is given as follows:
              \begin{align}
                  \bar{f}^*{'}(\lambda) = 
                  \begin{cases}
                  f'^{-1}(\lambda) & \text{ if } \lambda \in [0, f^{\prime}(1)]\\
                  1                & \text{ if } \lambda \in (f^{\prime}(1), +\infty),
                  \end{cases}
              \end{align}
    \end{itemize}
\end{lemma}
\begin{proof}
    The proof of these properties are elementary. For example, 
    \begin{align}
        f^*(\lambda) = 
        \begin{cases}
                  \lambda f^{'-1}(\lambda) - f(f^{'-1}(\lambda)) & \text{ if } \lambda\in \left[0, f^{\prime}(1)\right]\\
                  \lambda - f(1)            & \text{ if } \lambda\in \left(f^{\prime}(1), +\infty\right),
                  \end{cases}
    \end{align}
    Taking the derivative on both sides leads to the property in the lemma. 
\end{proof}

\subsection{Proof of Lemma \ref{unlimited_supply_SufNec}} \label{Apx:Unlimited}

To prove Lemma \ref{unlimited_supply_SufNec}, we need to show both the sufficiency and necessity of $\PUM\phi$ for achieving $ \alpha $-competitiveness. Our proof is divided into two parts. First, we first prove the the sufficiency and then the necessity. 

\begin{theorem}[Sufficiency]\label{sufficiency_unlimited_supply} 
For any $ \alpha \geq 1 $, $ \PUM\phi $ is $ \alpha $-competitive if there exists $ \phi $:
\begin{subequations}\label{eq:sufficiency_Phi}
\begin{align}
& \phi'(y) \leq \alpha\cdot \frac{f'(\phi(y))  - f'(y) }{\phi(y)\cdot f''(\phi(y))}, & & \forall y \geq  0, \\
& \textsf{Monotone: } \phi'(y) \geq 0, & & \forall y \geq  0,\\
& \textsf{Superlinear: } \phi(y) \geq  y, & & \forall y \geq  0,\\
& \textsf{Initial condition: } \phi(0) =  0.	
\end{align}
\end{subequations}
\end{theorem}
\begin{proof}

  For each arriving request $t$, let $\mu_t \geq 0$ and $\{\zeta_{t,j}\}_{j\in[m]}$ denote the KKT
  multipliers associated with the constraints $\sum_{j\in[m]} x_{t,j}\leq 1$ and $x_{t,j}\geq 0$,
  respectively, in the allocation problem Eq.~\eqref{eq_agent_t} solved by
  Algorithm~\ref{alg:PRM_MKwC}. Since $\varPhi = f'\circ\phi$ is continuous and non-decreasing,
  the objective of Eq.~\eqref{eq_agent_t} is concave and continuously differentiable in
  $\mathbf{x}_t$; as all constraints defining $\mathcal{X}_t$ are affine, the KKT conditions are
  necessary and sufficient for the optimality of $\mathbf{x}_t^{\ALG}$, and such multipliers exist.
  We set the dual variables of Eq.~\eqref{MKwC1_PrimalDualProblem} as
  $$
  \lambda_{t,j} = \varPhi(y_{t,j}), \qquad \forall j \in [m],
  $$
  together with the multipliers $ \mu_t $ above. Stationarity of Eq.~\eqref{eq_agent_t} at
  $ \mathbf{x}_t^{\ALG} $ reads
  \begin{equation} \label{eq:KKT_stationarity}
      v_{t,j} - \varPhi\left(y_{t-1,j}+ x_{t,j}^{\ALG}\right) - \mu_t + \zeta_{t,j} = 0,
      \qquad \forall j \in [m],
  \end{equation}
  that is, $ v_{t,j} - \lambda_{t,j} = \mu_t - \zeta_{t,j} \leq \mu_t $ for all $ j\in[m] $. Since
  $ y_{t,j} $ is non-decreasing in $ t $ and $ \varPhi $ is non-decreasing, we have
  $ \lambda_{t,j}\leq \lambda_{T,j} $, and therefore $ \mu_t \geq v_{t,j}-\lambda_{T,j} $ for all
  $ t\in[T] $ and $ j\in[m] $. Together with $ \mu_t\geq 0 $ and $ \lambda_{T,j}\geq 0 $, this
  shows that $ \{\lambda_{T,j}\}_{j\in[m]} $ and $ \{\mu_t\}_{t\in[T]} $ form a feasible dual
  solution.

  We now prove that the incremental inequality $P_t-P_{t-1}\geq \frac{1}{\alpha}\left(D_t-D_{t-1}\right)$ holds for $t\in [T]$. 
  By Algorithm \ref{alg:PRM_MKwC}, if $\varPhi(y_{t-1,j})> v_{t,j}$, request $t$ is rejected by node $j$, i.e., $x_{t,j} = 0$. In this case, node $j$ does not contribute to $P_t$ and $D_t$.
  Therefore, in the following we focus on those $j\in [m]$ such that $\varPhi(y_{t-1,j})\leq v_{t,j}$. 

  \begin{equation}
    \begin{aligned}
      \OPT 
      &\leq \sum_{t=1}^{T} \mu_t + \sum_{j=1}^{m} \sum_{t=1}^{T}\left(f^*(\lambda_{t,j}) - f^*(\lambda_{t-1,j})\right) \\
      &= \sum_{t=1}^{T} \mu_t + \sum_{j=1}^{m} \sum_{t=1}^{T}\left(f^*(\varPhi(y_{t,j})) - f^*(\varPhi(y_{t-1,j}))\right) \\
      &= \sum_{t=1}^{T} \mu_t + \sum_{j=1}^{m} \sum_{t=1}^{T} \int_{y_{t-1,j}}^{y_{t,j}} f^{*\prime}(\varPhi(u))\cdot \varPhi^\prime(u) \Id u
    \end{aligned} \label{ineq:Glb1}
  \end{equation}
  By Eq. \eqref{eq:sufficiency_Phi}, we know that the following inequality holds for all $y_{t-1, j}\geq 0$: 
  \begin{equation}
    f^{*\prime}(\varPhi(y_{t-1,j})) \cdot\varPhi^{\prime}(y_{t-1,j}) \leq \alpha\left(\varPhi(y_{t-1,j})-f^{\prime}(y_{t-1,j})\right). \label{ineq:DiffCondPrf2}
  \end{equation}
  Combining Eq. \eqref{ineq:DiffCondPrf2} with Eq. \eqref{ineq:Glb1}, we have: 
  \begin{equation}
    \begin{aligned}      
      \OPT
      &\leq \sum_{t=1}^{T} \mu_t + \sum_{j=1}^{m} \sum_{t=1}^{T} \int_{y_{t-1,j}}^{y_{t,j}} \alpha\left(\varPhi(u) - f^\prime(u)\right) \Id u \\
      &\leq \alpha \sum_{t=1}^{T} \mu_t + \alpha \sum_{j=1}^{m} \sum_{t=1}^{T} \int_{y_{t-1,j}}^{y_{t,j}} \varPhi(u)  \Id u - \alpha \sum_{j=1}^{m} f(y_{T,j}) \\
      &\leq \alpha \sum_{t=1}^{T} \left(\mu_t + \sum_{j=1}^{m} \varPhi(y_{t,j})\cdot x_{t,j}\right) - \alpha \sum_{j=1}^{m} f(y_{T,j}).
    \end{aligned}  \label{ineq:Glb2}
  \end{equation}
  
  It remains to express $ \mu_t $ in terms of the allocation made by
  Algorithm~\ref{alg:PRM_MKwC}. By complementary slackness, $ \zeta_{t,j}\, x^{\ALG}_{t,j} = 0 $ for
  every $ j\in[m] $, so multiplying Eq.~\eqref{eq:KKT_stationarity} by $ x^{\ALG}_{t,j} $ gives
  $$
  \left(v_{t,j} - \varPhi(y_{t,j})\right)\, x^{\ALG}_{t,j} = \mu_t\, x^{\ALG}_{t,j},
  \qquad \forall j\in[m].
  $$
  Summing over $ j\in[m] $ and invoking the complementary slackness condition
  $ \mu_t\left(1-\sum_{j\in[m]} x^{\ALG}_{t,j}\right) = 0 $, forces
  $ \mu_t \sum_{j\in[m]} x^{\ALG}_{t,j} = \mu_t $. Then we obtain
  $$
  \mu_t = \sum_{j=1}^{m}\left(v_{t,j} - \varPhi(y_{t,j})\right)\cdot x^{\ALG}_{t,j}.
  $$
  Substituting $ \mu_t $ into Eq. \eqref{ineq:Glb2}, we have: 
  $$
  \OPT \leq \alpha \sum_{j=1}^{m}\left(\sum_{t=1}^{T} v_{t,j} x_{t,j} - f(y_{T,j})\right) = \alpha \ALG. 
  $$ 
  This thus completes the proof of the sufficiency.
\end{proof}

We then prove the necessity. Our proof is based on the following family of instances, parameterized by $p > 0$:
$$
\mathcal{I}(p) = \{(v_1,v_2,\cdots,v_T)\mid 0< v_1\leq v_2 \leq \cdots v_T\leq p, T\in \mathbb{N}\},
$$
where for each request $t$, $\mathcal{J}_t = \{1\}$. 
Note that $\mathcal{I}(p)$ includes the instance where requests with the same value can appear arbitrarily many times. 

For an online algorithm, define the \textit{utilization function} $\psi: \mathbb{R}_+ \rightarrow [0,+\infty]$ as 
$$
\psi(p) = \sup \Big\{\text{The final utilization of the supply node after executing instance }I\mid I\in \mathcal{I}(p) \Big\}. 
$$
We observe the following:
\begin{enumerate}
  \item For a non-zero cost function $f$, there exists an online algorithm such that $\psi$ is a proper function, 
  i.e., $\exists p>0$ such that $\psi(p)<+\infty$. 
  \item For a proper function $\psi(p)$, it is increasing with respect to $p$. 
\end{enumerate}
Moreover, we can derive a lower bound for any online algorithm's result using its utilization function, formalized as the following proposition:
\begin{proposition}\label{prop:ALGinf}
  For any online algorithm with a utilization function $\psi$, assuming $\psi(p) < +\infty$,
  the infimum of the algorithm's result for any instance in $\mathcal{I}(p)$ is given by
  \begin{equation}
    \inf_{I\in \mathcal{I}(p)} \ALG(I) = \int_{0}^{p} u\Id \psi(u) - f(\psi(p)). \label{eq:PsiALG}
  \end{equation} 
\end{proposition}
\begin{proof}[Proof of Proposition \ref{prop:ALGinf}]
For fixed $p$, it is easy to see that for any instance in $\mathcal{I}(p)$, the primal objective 
achieved by an online algorithm with utilization function $\psi$ is no less than Eq. \eqref{eq:PsiALG}.
For example, for an instance in $\mathcal{I}(p)$ with a minimal value at $L>0$, the primal objective 
for this instance is no less than
$$
L\psi(L) + \int_{L}^{p} u\Id \psi(u) - f(\psi(p)),
$$
which is no less than Eq. \eqref{eq:PsiALG}.  Based on this observation, we only need to prove that 
  $$
  \inf_{I\in \mathcal{I}(p)} \ALG(I) \leq \int_{0}^{p} u\Id \psi(u) - f(\psi(p)).
  $$
  For a fixed $p$, we construct a family of instances to approximate $\inf_{I\in \mathcal{I}(p)} \ALG(I)$. 

  For any small $\varepsilon >0$, let $N, M>0$ be arbitrarily large integers such that 
  $M> \psi(p)$ and $\frac{p}{N} <\varepsilon$. We define an instance $I_\varepsilon$ as follows: 
  requests arrive in groups, with $N$ groups in total, and requests in the $k$-th group 
  have the same value $\frac{kp}{N}$, where $1 \leq k \leq N$. For instance $I_\varepsilon$, we have 
  \begin{equation}
    \ALG(I_\varepsilon) = \sum_{k=1}^{N} \frac{kp}{N} \left(\psi(\frac{kp}{N})-\psi(\frac{(k-1)p}{N})\right) - f(\psi(p)). \label{eq:UsumRSInt}
  \end{equation}
  Taking $\varepsilon\rightarrow 0$ in Eq. \eqref{eq:UsumRSInt}, we obtain the infimum of $ \ALG(I_\varepsilon)$: 
  $$
  \inf_{I\in \mathcal{I}(p)} \ALG(I) \leq \lim_{\varepsilon\rightarrow 0} \ALG(I_\varepsilon) = \int_{0}^{p} u\Id \psi(u) - f(\psi(p)).
  $$
  We thus complete the proof of Proposition \ref{prop:ALGinf}.
\end{proof}
\begin{theorem}[Necessity] \label{thm:necessity_IVP}
For any $ \alpha \geq 1 $, if there is an $ \alpha $-competitive online algorithm for \OACC, then there must exist 
a strictly-increasing reserve function $ \phi $:
\begin{align}\label{eq_necessity_phi}
\begin{cases} 
\phi'(y) = \alpha\cdot \frac{f'(\phi(y))  - f'(y) }{\phi(y)\cdot f''(\phi(y))}\geq 0, \quad \text{for all } y \geq 0, \\
\phi(0) =  0. 
\end{cases}
\end{align}
\end{theorem}
\begin{proof}[Proof of Theorem \ref{thm:necessity_IVP}]
  Considering all instances in $\mathcal{I}(p)$, for any $\alpha$-competitive online algorithm, it must have a utilization function $\psi$ such that 
  \begin{equation}
    \int_{0}^{p} u\Id \psi(u) - f(\psi(p)) \geq \frac{1}{\alpha} f^*(p). \label{ineq:Necspsi}
  \end{equation}
  We prove this by contradiction; otherwise, for any utilization function $\psi$, there must exist $\delta>0$ such that  
  $$
  \int_{0}^{p} u\Id \psi(u) - f(\psi(p)) < \frac{1}{\alpha} f^*(p) - \delta.
  $$

  For any $\epsilon >0$, there exists $\varepsilon >0$ such that the same instance $I_\varepsilon$ 
  in the proof of Proposition \ref{prop:ALGinf} satisfying
  $$
  \ALG(I_\varepsilon) \leq \int_{0}^{p} u\Id \psi(u) - f(\psi(p)) + \epsilon.
  $$
  Extend $I_\varepsilon$ to $\tilde{I_\varepsilon}$ by serving more than $f^{*\prime}(p)$ requests with the same value $p$. 
  Since $\psi(p)$ has already reached the final utilization for $I_\varepsilon$, 
  we have $ \ALG(\tilde{I_\varepsilon}) = \ALG(I_\varepsilon)$. By $f^{*\prime}(\lambda) = f^{\prime -1}(\lambda)> 0$, 
  for $\lambda \in \mathbb{R}_+$, the optimal offline result for $\tilde{I_\varepsilon}$ is given by 
  $$
  \OPT(\tilde{I_\varepsilon}) = \max_{y} \left(py-f(y)\right) = f^*(p). 
  $$
  Therefore, we have 
  $$
    \frac{\ALG(\tilde{I_\varepsilon})}{\OPT(\tilde{I_\varepsilon})} 
    \leq \frac{\int_{0}^{p} u\Id \psi(u) - f(\psi(p)) + \epsilon}{f^*(p)} 
    < \frac{1}{\alpha} + \frac{\epsilon - \delta}{f^*(p)}, 
  $$
  which contradicts the algorithm being $\alpha$-competitive when $\epsilon \leq \delta$. 

  By integration by parts, Eq. \eqref{ineq:Necspsi} simplifies to: 
  \begin{equation}
    p\psi(p) - \int_{0}^{p} \psi(u) \Id u - f(\psi(p)) \geq \frac{1}{\alpha} f^*(p). \label{ineq:SimNecspsi}
  \end{equation}
  Define 
  $$
  \Psi(p) = \inf \big\{\psi(p)\mid \psi \text{ satisfies Eq. } \eqref{ineq:SimNecspsi} \big\}, 
  $$
  and we prove that Eq. \eqref{ineq:SimNecspsi} holds with equality for $\Psi$, i.e., 
  \begin{equation}
    p\Psi(p) - \int_{0}^{p} \Psi(u) \Id u - f(\Psi(p)) = \frac{1}{\alpha} f^*(p). \label{eq:NecsPsi}
  \end{equation}

  For given $p>0$, we first show $\Psi$ is a solution to Eq. \eqref{ineq:SimNecspsi}. 
  For any $\eta>0$, there exists $\psi$ satisfying Eq. \eqref{ineq:SimNecspsi} such that $\psi(p) \leq \Psi(p) + \eta$. 
  Since $f$ is monotone, we have 
  $$
  \begin{aligned}
         & p\left(\Psi(p) + \eta\right) - \int_{0}^{p} \Psi(u) \Id u - f(\Psi(p))\\
   \geq\  & p\psi(p) - \int_{0}^{p} \psi(u) \Id u - f(\psi(p)) \\
   \geq\  & \frac{1}{\alpha} f^*(p). 
  \end{aligned}
  $$
  Thus, $\Psi$ satisfies Eq. \eqref{ineq:SimNecspsi} by letting $\eta\rightarrow 0$. 
  Further, Eq. \eqref{ineq:SimNecspsi} must hold with equality because otherwise we could further lower 
  the value of $\Psi$ while maintaining Eq. \eqref{ineq:SimNecspsi}, contradicting our choice of $\Psi$. 

  We now prove that $\Psi(p)$ is strictly increasing in $p$.  
  Suppose, for contradiction, that $\Psi(p)$ is not strictly increasing at some $p^*$. 
  Consider Eq. \eqref{eq:NecsPsi} at $p^*$. The derivative of the left hand side is 0 while 
  the derivative of the right hand side is strictly positive since $f^{*\prime}(p^*) = f^{\prime -1}(p^*)> 0$. 
  Therefore, for a sufficiently small increment at $p^*$, the equality in Eq. \eqref{eq:NecsPsi} no longer 
  holds, contradicting the assumption that $\Psi$ satisfies Eq. \eqref{eq:NecsPsi}. 

  Differentiating Eq. \eqref{eq:NecsPsi}, we have
  \begin{align}\label{eq_necessity_Psi}
    \begin{cases} 
    \Psi'(p) = \frac{1}{\alpha}\cdot \frac{f^*(\Psi(p))}{p - f'(\Psi(p))}, \quad \text{for all } p\geq 0, \\
    \Psi(0) =  0. 
    \end{cases}
  \end{align}

  Denote the inverse of $\Psi(p)$ by $\Phi(y) \triangleq \Psi^{-1}(y) = p$, which is strictly increasing. 
  Substituting $\Phi$ into Eq. \eqref{eq_necessity_Psi}, we obtain
  the following ODE for $\Phi$: 
  \begin{align}\label{eq_necessity_Phi}
    \begin{cases} 
    \Phi'(y) = \alpha\cdot \frac{\Phi(y) - f'(y)}{f^{*\prime}(\Phi(y))}, \quad \text{for all } y\geq 0, \\
    \Phi(0) =  0. 
    \end{cases}
  \end{align}
  The existence of a solution $\phi(y)$ to Eq.~\eqref{eq_necessity_phi} follows by the transformation $ \varPhi(y) = f'(\phi(y)) $, or equivalently, $ \phi(y) = f'^{-1}(\varPhi(y))$.  
\end{proof}

\subsection{Proof of Lemma \ref{Limited_supply_SufNec}} \label{Apx:Limited_supply_SufNec}

To prove Lemma \ref{Limited_supply_SufNec}, we need to show both the sufficiency and necessity of $\PUM\phi$ for achieving $ \alpha $-competitiveness. Our proof is divided into two parts. First, we first prove the the sufficiency and then the necessity. 

\begin{theorem}[Sufficiency]\label{sufficiency_case_0}
For any $ \alpha \geq 1 $, let $b = \min\{\rho,1\}$. $ \PUM\phi $ is $ \alpha $-competitive if $ \phi $ satisfies:
\begin{subequations} \label{eq:limited_Sufphi}
\begin{align}\normalfont
& \phi'(y) = \alpha\cdot \frac{f'(\phi(y))  - f'(y) }{\min\{\phi(y),1\}\cdot f''(\phi(y))}, & & \forall y \in [0, b], \\
& \textsf{Monotone: } \phi'(y) \geq 0, & & \forall y \in [0, b],\\
& \textsf{Superlinear: } \phi(y) \geq  y, & & \forall y \in [0, b],\\
& \textsf{Boundary conditions: } \phi(0) =  0, \phi(b) = \rho	
\end{align}
\end{subequations}
\end{theorem}
\begin{proof}

  Similarly, for each arriving request $t$, let $\mu_t \geq 0$, $\{\zeta_{t,j}\}_{j\in[m]}$ and
  $\{\kappa_{t,j}\}_{j\in[m]}$ denote the KKT multipliers of $\sum_{j\in[m]} x_{t,j}\leq 1$,
  $x_{t,j}\geq 0$ and $y_{t-1,j}+x_{t,j}\leq 1$, respectively, in the allocation problem
  Eq.~\eqref{eq_agent_t}. Since $\phi(b)=\rho$, we have $\varPhi(1) = f'(\rho) \geq v_{\rm max}$,
  so no maximizer of Eq.~\eqref{eq_agent_t} has $y_{t-1,j}+x_{t,j} = 1$ with $x_{t,j}>0$; the
  capacity constraint is therefore never active and $\kappa_{t,j} = 0$ for all $t$ and $j$. As in
  the proof of Theorem~\ref{sufficiency_unlimited_supply}, the objective is concave and
  continuously differentiable and the constraints are affine, so the KKT conditions are necessary
  and sufficient. We set $ \lambda_{t,j} = \varPhi(y_{t,j}) $ and keep the multipliers $ \mu_t $
  above; stationarity reads
  \begin{equation} \label{eq:KKT_stationarity_LS}
      v_{t,j} - \varPhi\left(y_{t-1,j}+ x_{t,j}^{\ALG}\right) - \mu_t + \zeta_{t,j} = 0,
      \qquad \forall j \in [m],
  \end{equation}
  and, exactly as before, $ \mu_t \geq v_{t,j}-\lambda_{t,j} \geq v_{t,j}-\lambda_{T,j} $, so the
  dual solution is feasible.

  We now prove the incremental inequality $P_t-P_{t-1}\geq \frac{1}{\alpha}\left(D_t-D_{t-1}\right)$ holds for $t\in [T]$. 
  Following similar argument in the proof of Theorem~\ref{sufficiency_unlimited_supply}, we focus on those $j\in [m]$ such that $\varPhi(y_{t-1,j})\leq v_{t,j}$. 
  Note that $\phi(y) \leq 1$ if and only if $\varPhi(y) \leq f^{\prime}(1)$. 
  We only need to consider Eq.~\eqref{eq:limited_Sufphi} for $y\geq 0$ such that $\phi(y) \leq \rho$ given $v_{\rm max} \leq f'(1)$, and for $y\geq 0$ such that $\phi(y) \leq 1$ given $v_{\rm max} > f'(1)$. Thus, we only need to consider Eq.~\eqref{eq:limited_Sufphi} for $y\in [0,b]$.
  
  Combining Eq. \eqref{eq:limited_Sufphi} and Lemma \ref{Lem:ModifiedCost}, we have 
  \begin{equation}
    f^{*\prime}(\varPhi(y_{t-1,j})) \cdot\varPhi^{\prime}(y_{t-1,j}) \leq \alpha\left(\varPhi(y_{t-1,j})-f^{\prime}(y_{t-1,j})\right), 
  \end{equation}
  for any $y\in [0,b]$.  Similarly to the proof of Theorem \ref{sufficiency_unlimited_supply}, we have: 
  \begin{equation}
    \begin{aligned}      
      \OPT
      &\leq \sum_{t=1}^{T} \mu_t + \sum_{j=1}^{m} \sum_{t=1}^{T}\left(f^*(\lambda_{t,j}) - f^*(\lambda_{t-1,j})\right) \\
      &= \sum_{t=1}^{T} \mu_t + \sum_{j=1}^{m} \sum_{t=1}^{T} \int_{y_{t-1,j}}^{y_{t,j}} f^{*\prime}(\varPhi(u))\cdot \varPhi^\prime(u) \Id u \\
      &\leq \sum_{t=1}^{T} \mu_t + \sum_{j=1}^{m} \sum_{t=1}^{T} \int_{y_{t-1,j}}^{y_{t,j}} \alpha\left(\varPhi(u) - f^\prime(u)\right) \Id u \\
      &\leq \alpha \sum_{t=1}^{T} \left(\mu_t + \sum_{j=1}^{m} \varPhi(y_{t,j})\cdot x_{t,j}\right) - \alpha \sum_{j=1}^{m} f(y_{T,j}).
    \end{aligned}  \label{ineq:Glb2_limited}
  \end{equation}

  By complementary slackness, $ \zeta_{t,j}\, x^{\ALG}_{t,j} = 0 $ for every $ j\in[m] $, so
  Eq.~\eqref{eq:KKT_stationarity_LS} gives
  $ \left(v_{t,j} - \varPhi(y_{t,j})\right) x^{\ALG}_{t,j} = \mu_t x^{\ALG}_{t,j} $. Summing over
  $ j\in[m] $ and using $ \mu_t\big(1-\sum_{j\in[m]} x^{\ALG}_{t,j}\big) = 0 $, we obtain
  $$
  \mu_t = \sum_{j=1}^{m}\left(v_{t,j} - \varPhi(y_{t,j})\right)\cdot x^{\ALG}_{t,j}.
  $$
  Substituting into Eq. \eqref{ineq:Glb2_limited}, we have: 
  $$
  \OPT \leq \alpha \sum_{j=1}^{m}\left(\sum_{t=1}^{T} v_{t,j} x_{t,j} - f(y_{T,j})\right) = \alpha \ALG. 
  $$ 
  Thus, we complete the proof. 
\end{proof}

We now prove the necessity.

\begin{theorem}[Necessity] \label{theorem_necessary_limited_supply}
Given \OACC with hard supply constraints, if there is an $ \alpha $-competitive online algorithm, then there must exist  a strictly-increasing reserve function $ \phi $:
\begin{align}\label{eq_necessity_phi}
\begin{cases} 
\phi'(y) = \alpha\cdot \frac{f'(\phi(y))  - f'(y) }{\min\{\phi(y),1\}\cdot f''(\phi(y))}, \quad y\geq 0, \\
\phi(0) =  0, \phi(\min\{\rho,1\}) \geq \rho.
\end{cases}
\end{align}
\end{theorem}

\begin{proof}
  For given $v_{\rm max}>0$, we consider any instance in $\mathcal{I}(v_{\rm max})$. By Proposition \ref{prop:ALGinf}, an online algorithm with a utilization function $\psi$ that is $\alpha$-competitive satisfies
  $$
  \inf_{I\in \mathcal{I}(v_{\rm max})} \ALG(I) = \int_{0}^{V} u\Id \psi(u) - f(\psi(V)) \geq \frac{1}{\alpha} \OPT(v_{\rm max}), 
  $$
  where $\OPT(v_{\rm max})$ is the offline optimal value and $V = \min\{v_{\rm max}, \psi^{-1}(1)\}$.\footnote{For simplicity, here we assume $\psi$ is strictly increasing. For those $\psi$ which may not be so, $\psi^{-1}(1)$ may become a set of numbers, and we set $V$ to be the minimum among $\{v_{\rm max}\}\cup \psi^{-1}(1)$.} Note that there are two differences from the unlimited supply case (i.e., \OACC without hard supply constraints in Eq. \eqref{eq:oacc}):
  \begin{enumerate}
    \item When calculating $\inf_{I\in \mathcal{I}(v_{\rm max})} \ALG(I)$, the integral interval may not extend to $v_{\rm max}$ due to the supply limit.  
    \item $\OPT(v_{\rm max}) = \bar{f}^*(v_{\rm max})$, which is less than $f^*(v_{\rm max})$ when $v_{\rm max}$ is large.
  \end{enumerate} 
  We first show that it is sufficient to consider algorithms with the utilization function $\psi$ such that $\psi^{-1}(1) \geq v_{\rm max}$, i.e., $\psi(v_{\rm max})\leq 1$. Otherwise, we consider 
  $$
  \tilde{\psi}(p) = \psi\left(\lambda p\right), \quad \forall p\geq 0, 
  $$
  where $\lambda = \psi^{-1}\left(1\right)/v_{\rm max} < 1$. For $\tilde{\psi}$, we have $\tilde{V} = \min\{v_{\rm max}, \tilde{\psi}^{-1}(1)\} = v_{\rm max}$, and  
  $$
  \begin{aligned}
    \inf_{I\in \mathcal{I}(v_{\rm max})} \widetilde{\ALG}(I) 
    &= \int_{0}^{v_{\rm max}} u\Id \psi\left(\lambda u\right) - f\left(\psi\left(\lambda v_{\rm max}\right)\right) \\
    &= \frac{1}{\lambda}\int_{0}^{\lambda v_{\rm max}} u\Id \psi(u) - f\left(1\right) \\
    &= \frac{1}{\lambda}\left(\int_{0}^{\lambda v_{\rm max}} u\Id \psi(u) - f\left(1\right)\right) + \left(\frac{1}{\lambda}-1\right) f\left(1\right) \\
    &\geq \frac{1}{\lambda} \OPT(\lambda v_{\rm max}) + \left(\frac{1}{\lambda}-1\right) f\left(1\right) \\
    &= \frac{1}{\lambda} \max_{y\in [0,1]} \left(\lambda v_{\rm max} \cdot y - \lambda f(y) + \left(1-\lambda\right)\left(f(1) - f(y)\right)\right) \\
    &\geq \OPT(v_{\rm max}).
  \end{aligned}
  $$
  This implies that an online algorithm with the utilization function $\tilde{\psi}$ is also $\alpha$-competitive, and thus we add a constraint, 
  \begin{equation}
    \psi(v_{\rm max})\leq 1 \label{ineq:added_constraint}
  \end{equation}
  to the utilization function in the sequel. Consequently, we have 
  $$
  \begin{aligned}
    \inf_{I\in \mathcal{I}(v_{\rm max})} \ALG(I) 
    &= \int_{0}^{v_{\rm max}} u\Id \psi(u) - f(\psi(v_{\rm max})) \\
    &= v_{\rm max}\psi(v_{\rm max}) - \int_{0}^{v_{\rm max}} \psi(u) \Id u - f(\psi(v_{\rm max})),
  \end{aligned}
  $$
  which shares the same form as the unlimited supply case. 

  In the limited supply case, we have 
  $$
  \OPT(v_{\rm max}) = \bar{f}^*(v_{\rm max}). 
  $$
  Similar to the proof of Theorem \ref{sufficiency_case_0}, if an online algorithm with the utilization function $\psi$ is $\alpha$-competitive, then there exists a strictly increasing $\psi$ that satisfies
  $$
  v_{\rm max}\psi(v_{\rm max}) - \int_{0}^{v_{\rm max}} \psi(u) \Id u - f(\psi(v_{\rm max})) = \frac{1}{\alpha} \bar{f}^*(v_{\rm max}), \quad \forall v_{\rm max} \geq 0.
  $$ 
  Note that $\bar{f}^*$ on the right hand side has two phases but continuous at $v_{\rm max}=f^{\prime}(1)$. Therefore, for a given $v_{\rm max}$, there exists a strictly increasing utilization function $\psi$ satisfying
  \begin{equation}
    \left\{
    \begin{aligned}
      &p\psi(p) - \int_{0}^{p} \psi(u) \Id u - f(\psi(p)) = \frac{1}{\alpha} \bar{f}^*(p), \quad p\geq 0\\
      &\psi(0) = 0, \quad \psi(v_{\rm max}) \leq 1,
    \end{aligned}\right.
    \label{eq:NecsPsi_limit_pre}
  \end{equation}
  where the second boundary condition is the constraint added by us. We now further simplify Eq. \eqref{eq:NecsPsi_limit_pre}. 

  \paragraph{Case I:} $v_{\rm max} \leq f^{\prime}(1)$. The utilization function $\psi$, as a solution to Eq. \eqref{eq:NecsPsi_limit_pre}, satisfies the boundary condition 
  \begin{equation}
    \psi(v_{\rm max}) \leq \rho, \label{eq:boundary_case1}
  \end{equation}
  and thus $\psi(v_{\rm max}) \leq \rho \leq 1$, i.e., the constraint Eq. \eqref{ineq:added_constraint} is automatically satisfied. Otherwise, $\psi$ is a solution to 
  \begin{equation}
    \left\{
    \begin{aligned}
      &p\psi(p) - \int_{0}^{p} \psi(u) \Id u - f(\psi(p)) = \frac{1}{\alpha} f^*(p), \quad p\in [0,f^{\prime}(1)],\\
      &\psi(0) = 0, \quad \psi(v_{\rm max}) > \rho.
    \end{aligned}\right.
    \label{eq:NecsPsi_limit_case1_pre_contrd}
  \end{equation}
  Since $\psi(v_{\rm max}) > \rho = f^{\prime-1}(v_{\rm max})$, there exist a small $\varepsilon>0$ such that $\psi(p) > f^{\prime-1}(p)$ for $p\in [v_{\rm max}-\varepsilon, v_{\rm max}]$. Denote $F_{\psi}(p) = p\psi(p) - \int_{0}^{p} \psi(u) \Id u - f(\psi(p))$, and thus 
  $$
  F_{\psi}(p) = \frac{1}{\alpha} f^*(p), \quad \forall p\geq 0.
  $$
  Since
  $$
  F_{\psi}^{\prime}(p) = \psi^{\prime}(p)\left(p-f^{\prime}\left(\psi\left(p\right)\right)\right) < 0, \quad \forall p\in [v_{\rm max}-\varepsilon, v_{\rm max}], 
  $$
  we have $F_{\psi}^{\prime}(p)$ is decreasing on the interval $[v_{\rm max}-\varepsilon, v_{\rm max}]$. On the other hand, $f^*(p)$ is increasing with respect to $p\geq 0$. This implies, for $p\in [v_{\rm max}-\varepsilon, v_{\rm max}]$, 
  $$
  F_{\psi}(p) > \frac{1}{\alpha} f^*(p),
  $$
  which contradicts the fact that $\psi$ is a solution to Eq. \eqref{eq:NecsPsi_limit_case1_pre_contrd}. Therefore, there exists a utilization function $\psi$ satisfying
  \begin{equation}
    \left\{
    \begin{aligned}
      &p\psi(p) - \int_{0}^{p} \psi(u) \Id u - f(\psi(p)) = \frac{1}{\alpha} f^*(p), \quad \forall p\in [0,f^{\prime}(1)],\\
      &\psi(0) = 0, \quad \psi(v_{\rm max}) \leq \rho.
    \end{aligned}\right.
    \label{eq:NecsPsi_limit_case1_pre}
  \end{equation}

  By differentiating Eq. \eqref{eq:NecsPsi_limit_case1} and applying the transformations in Theorem \ref{sufficiency_case_0}, $\varPhi \triangleq \psi^{-1}$ and $ \varPhi(y) = f'(\phi(y)) $, we have
  \begin{equation}
    \left\{
    \begin{aligned}
      &\phi'(y) = \alpha\cdot \frac{f'(\phi(y))  - f'(y) }{\phi(y)\cdot f''(\phi(y))},\quad y\geq 0,  \\
      &\phi(y) \leq 1, \\
      &\phi(0) = 0, \quad \phi(\rho) \geq \rho,
    \end{aligned}\right. 
    \label{eq:NecsPsi_limit_case1}
  \end{equation}
  where $\phi(y)\leq 1$ is from $p\leq f^{\prime}(1)$ in Eq. \eqref{eq:NecsPsi_limit_case1_pre} and the boundary condition is derived from the boundary condition in Eq. \eqref{eq:NecsPsi_limit_case1_pre}. 

  \paragraph{Case II:} $v_{\rm max} > f^{\prime}(1)$. By Eq. \eqref{eq:NecsPsi_limit_pre}, for a given $v_{\rm max}$, there exists a utilization function $\psi$ satisfying
  \begin{equation}
    \left\{
    \begin{aligned}
      & p\psi(p) - \int_{0}^{p} \psi(u) \Id u - f(\psi(p)) = \frac{1}{\alpha} f^*(p), \quad p\in [0,f^{\prime}(1)), \\
      & p\psi(p) - \int_{0}^{p} \psi(u) \Id u - f(\psi(p)) = \frac{1}{\alpha} \left(p - f(1)\right), \quad p\in [f^{\prime}(1),+\infty), \\
      & \psi(0) = 0, \quad \psi(v_{\rm max}) \leq 1.
    \end{aligned}\right.
    \label{eq:NecsPsi_limit_case2_pre}
  \end{equation}
  Similar to Case I, we can rewrite Eq. \eqref{eq:NecsPsi_limit_case2_pre} using the reserve function $\phi$: 
  \begin{equation}
    \left\{
    \begin{aligned}
      &\phi'(y) = \alpha\cdot \frac{f'(\phi(y))  - f'(y) }{\phi(y)\cdot f''(\phi(y))}, \quad y\in [0,\zeta), \\
      &\phi'(y) = \alpha\cdot \frac{f'(\phi(y))  - f'(y) }{f''(\phi(y))}, \quad y\in [\zeta,+\infty), \\
      &\phi(0) = 0, \quad \phi(1) \geq \rho,
    \end{aligned}\right.
    \label{eq:NecsPsi_limit_case2}
  \end{equation}
  where $\phi(\zeta)=1$ and the second boundary condition is derived from the constraint Eq. \eqref{ineq:added_constraint} and $p\leq v_{\rm max}$ in Eq. \eqref{eq:NecsPsi_limit_case2_pre}. Incorporating the equations in Eq. \eqref{eq:NecsPsi_limit_case2} and combining with Eq. \eqref{eq:NecsPsi_limit_case1}, we obtain Eq. \eqref{eq_necessity_phi}. 
  Therefore, we complete the proof. 
  \end{proof}

\section{Proof of Theorem \ref{theorem_lower_bound} (Lower Bound)} \label{Apx:LowerBound}

We establish the lower bound on the competitive ratio for strongly $(\tau, \sigma)$-elastic costs, and thus general $(\tau,\sigma)$-elastic costs.
In the proof, we will analyze the asymptotic behavior of $F(\phi, y)$ as $y \rightarrow 0^+$ and $y \rightarrow +\infty$, corresponding to the parameters $\tau$ and $\sigma$, respectively.

\subsection{Left Asymptotic Analysis} 
As $ y \rightarrow 0^+ $, we use the Taylor expansion to analyze Eq.~\eqref{eq:mainODE}.
We claim that for the Taylor expansion of $f'(y)$, the lowest degree of the expansion must be at least $\tau-1$.
We prove this by contradiction. Suppose that the lowest degree of $f'(y)$ is $m - 1 < \tau -1$. The corresponding coefficient is positive since $f$ is convex. On the other hand, the lowest degree of $f''(y)$ is $m - 2 < \tau - 2$, and the function $f''(y)/y^{\tau-2}$ must be decreasing for $y>0$ sufficiently small. This contradicts the fact that $f$ is strongly $(\tau,\sigma)$-elastic. Thus, the claim holds. 
We assume w.l.o.g. that $\tau>1$ and
$$
f^\prime(y) = c_{\tau -1}y^{\tau - 1} + o(y^{\tau - 1}), \quad c_{\tau-1} > 0, 
$$
where $o(y^{\tau - 1})$ denotes the higher-order terms of $y^{\tau - 1}$. For $\tau = 1$ or the case where the lowest degree of the expansion of $f'$ is strictly larger than $\tau - 1$, we consider the lowest term of $f''$, which must have a positive coefficient, and the sequel analysis still applies.
Then we have 
$$
f^{\prime\prime}(y) = (\tau - 1) c_{\tau - 1} y^{\tau-2} + o(y^{\tau - 2}).
$$

Let
\begin{align*}
\chi := \lim\limits_{y\rightarrow 0^+}  \phi'(y).
\end{align*}
By Eq.~\eqref{eq:mainODE} and the expansion of $f'$ and $f''$, we have
\begin{align*}
\chi  =  \ & \lim\limits_{y\rightarrow 0^+} \alpha\cdot \frac{f'\left(\phi\right) - f'(y) }{\phi\cdot f''(\phi)} \\
=\ & \lim\limits_{y\rightarrow 0^+} \alpha\cdot \frac{c_{\tau -1}\phi^{\tau - 1} + o(\phi^{\tau - 1}) - c_{\tau -1}y^{\tau - 1} - o(y^{\tau - 1})}{(\tau - 1)c_{\tau - 1} \phi^{\tau-1} + o(\phi^{\tau - 1})}\\
=\ &\lim\limits_{y\rightarrow 0^+} \frac{\alpha}{\tau - 1}\cdot \left(1 - \left(\frac{y}{\phi}\right)^{\tau-1}\right) \\
=\ &\frac{\alpha}{\tau - 1}\cdot \left(1 - \frac{1}{\chi^{\tau-1}}\right),
\end{align*}	
leading to
\begin{align}\label{eq_cp} 
\underbrace{ \chi^\tau - \frac{\alpha}{\tau-1}\cdot  \chi^{\tau-1} + \frac{\alpha}{\tau-1}}_{\CP(\alpha,\tau)} = 0,
\end{align}
where $ \CP(\alpha,\tau) $ is the characteristic polynomial defined in Eq. \eqref{eq:CharacPoly}.

To guarantee that $ \phi'(0) $ is well defined, $\CP(\alpha,\tau) = 0 $ must have at least one real root. By some algebra, $ \alpha \geq \ataus =   \tau^{\frac{\tau}{\tau-1}} $ follows (for a formal discussion, see Lemma \ref{Lem:PowerEquation} in Appendix~\ref{Apx:UpperBound}). 

\subsection{Right Asymptotic Analysis} \label{Apx:Right Asymptotic Analysis}
When $ y \rightarrow +\infty $, $ \phi(y) $ should approach infinity because of the unlimited supply and the monotone cost function $ f $. 
On the other hand, the ratio $f'(\phi(y))/f'(y)$ should be bounded even if $y$ tends to infinity; otherwise the reserve function is reserving too much for the future so that their pseudo marginal cost is infinitely larger than the true marginal cost. 
Since $f$ is (tightly) $(\tau,\sigma)$-elastic, we have for sufficiently large $y>0$, there exists some small $\epsilon>0$ such that
$$
Ef'(y) \geq \tau - 1 + \epsilon,
$$
which leads to $f'(y)/y^{\tau - 1 + \epsilon}$ is increasing. 
This implies
$$
\frac{f'(\phi(y))}{f'(y)} \geq \left(\frac{\phi(y)}{y}\right)^{\tau - 1 + \epsilon} .
$$
Therefore, the ratio $\phi(y)/y$ is bounded and so is $\phi'(y)$. Let
\begin{align*}
\Delta := \lim\limits_{y\rightarrow +\infty}  \phi'(y),
\end{align*}
By Eq.~\eqref{eq:mainODE}, we have
\begin{align*}
\Delta =\ & \lim\limits_{y\rightarrow \infty} \alpha\cdot \frac{f'\left(\phi\right) - f'(y) }{\phi\cdot f''(\phi)} \\
=\ & \lim\limits_{y\rightarrow \infty} \alpha\cdot \frac{1 - \frac{f'(y)}{f'(\phi)}}{\frac{\phi\cdot f''(\phi)}{f'(\phi)}}.
\end{align*}	
Since $ f(y) $ is (tightly) $ (\tau,\sigma) $-elastic,  we have
\begin{align*}
\lim\limits_{y\rightarrow \infty}  \frac{\phi\cdot f''(\phi)}{f'(\phi)} 
&= \lim\limits_{y\rightarrow \infty} Ef'(y) = \sigma-1, \\
\lim\limits_{y\rightarrow \infty}  \ln\frac{f'(\phi)}{f'(y)} 
&= \lim\limits_{y\rightarrow \infty}  \int_y^{\phi(y)} \frac{Ef'(u)}{u} \Id u \\
&= \lim\limits_{y\rightarrow \infty}(\sigma - 1)\ln\frac{\phi(y)}{y} \\
&= (\sigma-1)\ln \Delta.
\end{align*}
Combining the above results leads to the following equation:
\begin{align*}
\Delta = \alpha\cdot \frac{1 - \frac{1}{\Delta^{\sigma - 1}}}{\sigma - 1}, 
\end{align*}
leading to the following characteristic polynomial $ \CP(\alpha,\sigma) $:
\begin{align}\label{eq_cp_sigma} 
\Delta^\sigma - \frac{\alpha}{\sigma-1}\cdot  \Delta^{\sigma-1} + \frac{\alpha}{\sigma-1} = 0.
\end{align}
By some similar algebra, $ \alpha \geq \asigmas =   \sigma^{\frac{\sigma}{\sigma-1}} $ follows. 
Similar to the analysis above, when $ \alpha = \asigmas  $, $ \phi'(\infty) $ is given by
\begin{align*}
\phi'(\infty) =  \Dstar = \asigmas/\sigma = \sigma^{\frac{1}{\sigma-1}}.
\end{align*}

Combining the left asymptotic analysis (i.e., $ y \rightarrow 0^+$) for $ \tau $ with the right asymptotic analysis (i.e., $ y \rightarrow +\infty $) for $ \sigma $:
\begin{align*}
\alpha \geq \max\left\{\ataus,  \asigmas \right\} =  \max \Big\{ \tau^{\frac{\tau}{\tau-1}},   \sigma^{\frac{\sigma}{\sigma-1}} \Big\} = \sigma^{\frac{\sigma}{\sigma-1}}.
\end{align*}
We thus complete the proof of the lower bound in Theorem \ref{theorem_lower_bound}.

\section{Proofs Related to the Design Space} \label{Apx:UpperBound}

In this section, we first give a lemma on the characteristic polynomial and the impact of different degrees for Eq.~\eqref{eq:SufNec_phi}. Then we present the proof of Theorem~\ref{Thm:UpperBoundAch} in Sections \ref{Apx:Proof of the Existence of a Solution}-\ref{Apx:Proof for Tight Linear Bounds}, and the proof of Proposition~\ref{Prop:BoundResp} in Section~\ref{Apx:BoundResp}. 
Specifically, to prove Theorem~\ref{Thm:UpperBoundAch}, we first show the existence of a solution to Eq.~\eqref{eq:SufNec_phi}, and then, the existence of infinitely many solutions. Finally, we show that the linear functions in Theorem~\ref{Thm:UpperBoundAch} are the tight linear bounds of the solutions to Eq.~\eqref{eq:SufNec_phi}.

\subsection{Preliminaries}

We first give a formal definition of the linear tight bound. 
\begin{definition}\label{Def:TightBound}
    Let $\mathcal{G}$ be a subset of continuous functions defined on $\mathbb{R}_+$. 
    For a constant $k^* \geq 0$, we say that the function $y \mapsto k^* y$ (or simply the scalar $k^*$) is the \emph{tight linear upper (respectively, lower) bound} for $\mathcal{G}$ if the following conditions hold:
    \begin{itemize}
        \item \textbf{Bound condition:} For every $f \in \mathcal{G}$, we have 
        $$
            k^* y \geq f(y) \quad (\text{respectively, } k^* y \leq f(y)) \quad \text{for all } y \in (0, +\infty).
        $$
        \item \textbf{Tightness condition:} For any scalar $k$ such that $0 \leq k < k^*$ (respectively, $k > k^*$), there exist $f_0 \in \mathcal{G}$ and $y_0 > 0$ such that 
        $$
            f_0(y_0) > k y_0 \quad (\text{respectively, } f_0(y_0) < k y_0).
        $$
    \end{itemize}
\end{definition}
When the context is clear, we simply refer to $k^*$ as a \emph{tight} upper (respectively, lower) bound. 
If $\mathcal{G}$ admits a tight lower bound $k_1$ and a tight upper bound $k_2$, 
then the region bounded by the two lines $y = k_1x$ and $y = k_2x$ forms the minimal cone that bounds $\mathcal{G}$.

In the proof of the lower bound (see Appendix~\ref{Apx:LowerBound}), the characteristic polynomial plays a central role. We now formally state its properties in the following lemma. As will become evident, these properties also fundamentally influence the existence of infinitely many solutions in the design space. 
\begin{lemma}
  \label{Lem:PowerEquation}
  For any integer $k> 1$ and any $\alpha>1$, let $z_\star = k^{1/(k-1)}$ and $\alpha_\star = z_\star^{k}$. The following equation of $x\in \mathbb{R}$ 
  \begin{equation}
    x^{k} - \frac{\alpha}{k-1}\left(x^{k-1} - 1\right) = 0 \label{eq:PowerEquation}
  \end{equation}
  has solutions:
  \begin{enumerate}
    \item $x_1\geq z_\star\geq x_2>1$, if $k$ is even and $\alpha\geq \alpha_\star$,
    \item $x_1\geq z_\star\geq x_2>1>0>x_3>-1$, if $k$ is odd and $\alpha \geq \alpha_\star$, 
  \end{enumerate} 
  where $x_1 = x_2 = z_\star$ if and only if $\alpha = \alpha_\star$. For $\alpha<\alpha_\star$, there is no positive solution to Eq. \eqref{eq:PowerEquation}, and only one negative solution $-1<x_1<0$ provided that $k$ is odd. 

  Specifically, for fixed $k$, $x_1$ and $x_2$ are dependent of $\alpha$, denoted as $x_1(\alpha)$ and $x_2(\alpha)$. We have $x_1(\alpha)$ and $x_2(\alpha)$, respectively, are increasing and decreasing, respectively, for $\alpha \geq \alpha_\star$. Moreover, as $\alpha\rightarrow+\infty$, we have $x_1(\alpha)\rightarrow +\infty$ and $x_2(\alpha)\rightarrow 1$. 
  
  On the other hand, for fixed $\alpha >1$, $x_1$ is dependent of $k$, denoted as $x_1(k)$. For integer $k\geq 1$ such that $\alpha_\star\leq \alpha$, we have $x_1(k)$ is strictly decreasing 
\end{lemma}
\begin{proof}
  Define $g_k(x)$, $x\in\mathbb{R}$, as follows
  $$
  g_k(x) = \frac{(k-1)x^{k}}{x^{k-1}-1}. 
  $$
  To solve Eq. \eqref{eq:PowerEquation}, it is equivalent to find $x$ such that $x^{k-1} \neq 1$ and $g_k(x) = \alpha$. We have the derivative of $g_k$ given by
  $$
  g_k^{\prime}(x) = \frac{(k-1)x^{k-1}\left(x^{k-1}-k\right)}{\left(x^{k-1}-1\right)^2}, \quad x^{k-1}\neq 1.
  $$
  If $k$ is even, $g_k$ is divided into 2 branches $(-\infty,1)\cup (1,+\infty)$, and $g_k$ has a local maxima at $x = 0$, where $g_k(0)=0$, and a local minima at $x = z_\star$, where $g_k(z_\star) = \alpha_\star$.
  If $k$ is odd, $g_k$ is divided into 3 branches $(-\infty,-1)\cup(-1,1)\cup (1,+\infty)$, and $g_k$ has a local maxima at $x = -z_\star$ where $g_k(-z_\star) = -\alpha_\star$, and a local minima at $x = z_\star$, where $g(z_\star) = \alpha_\star$; for $x\in (-1,1)$, $g_k(x)$ is strictly decreasing.
  Combining these two cases, one can find the solutions to Eq. \eqref{eq:PowerEquation} and verify how $x_1(\alpha)$ and $x_2(\alpha)$ depend on $\alpha$.

  To prove the third part, for fixed $\alpha$, we calculate 
  $$
  g_{k+1}(x) - g_k(x) = \frac{x^{k}\left(x^{k} - k x + k-1\right)}{\left(x^{k-1}-1\right)\left(x^{k}-1\right)}. 
  $$
  Since $x_1(k) \geq z_\star = k^{1/(k-1)}$, we have 
  $$
  \begin{aligned}
    &\left(x_1\left(k\right)\right)^{k} - k x_1(k) + k - 1 \\
    = &\left(\left(x_1\left(k\right)\right)^{k-1}- k\right)\cdot x_1(k) + k -1 \\
    \geq & k-1 >0.
  \end{aligned}
  $$
  This implies 
  $$
  g_{k+1}\left(x_1\left(k\right)\right) - \alpha = g_{k+1}\left(x_1\left(k\right)\right) - g_k\left(x_1\left(k\right)\right) > 0. 
  $$
  Then we have 
  $$
  g_{k+1}\left(x_1\left(k+1\right)\right) = \alpha < g_{k+1}\left(x_1\left(k\right)\right). 
  $$
  By previous analysis, we have $g_{k+1}(x)$ is increasing for $x\geq x_1\left(k+1\right)$, and thus 
  $$
  x_1\left(k+1\right) < x_1\left(k\right). 
  $$
  Combining this together, we complete the proof. 
\end{proof}

Note that the results of $x_1$ and $x_2$ in Lemma \ref{Lem:PowerEquation} hold for $k\in \mathbb{R}_+$. There will be some troubles in defining $x_3$ when $k$ is not an integer, but we only use roots $x_1$ and $x_2$ of Eq. \eqref{eq:PowerEquation} throughout the paper, i.e., the positive roots of characteristic polynomials.

\begin{lemma} \label{Lem:ElasticComp}
  For a strongly $(\tau,\sigma)$-elastic function $f$ (see the definition right after Assumption~\ref{ass1}), we have for any $y>0$, 
  \begin{align*}
  \tau-1 \leq &Ef'(y) \leq \sigma -1.
  \end{align*}
  Furthermore, if two strongly $(\tau,\sigma)$-elastic cost functions $f_1$ and $f_2$ such that
  $$
  Ef_1''(y) \geq Ef_2''(y), \quad \forall y>0
  $$
  we then have 
  $$
  F_1(\phi,y) \leq F_2(\phi,y), \quad \forall \phi \geq y > 0,
  $$ 
  where $F_1$ and $F_2$ are defined as in Eq. \eqref{Eq:NEPFunction} with respect to $f_1$ and $f_2$, respectively.  
\end{lemma}
\begin{proof}
    We now prove the first part of this lemma. If $f$ is strongly $(\tau,\sigma)$-elastic, one can show for all $y > 0$, the following two functions are monotone (not necessarily strictly):
    $$
    \begin{aligned}    
      \frac{f^{\prime\prime}(y)}{y^{\tau - 2}} & \text{ is increasing}, \\
      \frac{f^{\prime\prime}(y)}{y^{\sigma - 2}} & \text{ is decreasing}.
    \end{aligned}
    $$
    Since $\frac{f^{\prime\prime}(y)}{y^{\tau - 2}}$ is increasing for all $y>0$, we have for any $0<t<y$,
    \begin{align*}
    & f''(t) \leq f''(y) \left(\frac{t}{y}\right)^{\tau - 2} \\
    \Rightarrow \quad & \int_{0}^{y} f''(t) \Id t \leq f''(y) \int_{0}^{y} \left(\frac{t}{y}\right)^{\tau - 2} \Id t \\
    \Rightarrow \quad & f'(y) \leq \frac{f''(y)}{\tau - 1} y,
    \end{align*}
    which implies that $Ef'(y) \geq \tau - 1$. Similarly, one can show that $Ef'(y) \leq \sigma - 1$.

    For the second part, denote $g_i(y) = \ln\left(f_k''(y)\right)$ for $y>0$ and $k\in \{1,2\}$.
    since $Ef_1''(y) \geq Ef_2''(y)$ for $y>0$, we have $g_1'(y) \geq g_2'(y)$. Let $h(y) = \exp\left(g_1(y) - g_2(y)\right)$, and $h(y)$ is non-decreasing. This implies
    $$
    \frac{f_1^{\prime\prime}(y)}{f_2^{\prime\prime}(y)} \leq  \frac{f_1^{\prime\prime}(\phi)}{f_2^{\prime\prime}(\phi)}. 
    $$
    Let $v = y/\phi \leq 1$, we have
    $$
    \begin{aligned}
        & \frac{f_1^{\prime\prime}(v\phi)}{f_1^{\prime\prime}(\phi)} \leq  \frac{f_2^{\prime\prime}(v\phi)}{f_2^{\prime\prime}(\phi)} \\
        \Rightarrow \quad & \int_{v}^{1} \frac{f_1^{\prime\prime}(u\phi)}{f_1^{\prime\prime}(\phi)} \Id u \leq  \int_{v}^{1} \frac{f_2^{\prime\prime}(u\phi)}{f_2^{\prime\prime}(\phi)} \Id u \\
        \Rightarrow \quad & \frac{f_1^{\prime}(\phi) - f_1^{\prime}(v\phi)}{\phi f_1^{\prime\prime}(\phi)} \leq \frac{f_2^{\prime}(\phi) - f_2^{\prime}(v\phi)}{\phi f_2^{\prime\prime}(\phi)} \\
        \Rightarrow \quad & F_1(\phi,y) \leq F_2(\phi,y).
    \end{aligned}
    $$
    Thus, we complete the proof.
\end{proof}

With this property, we can still show the following comparison result regarding the NEP in Eq.~\eqref{eq:SufNec_phi}.
\begin{corollary} \label{Cor:ElasticFBound}
  If $f$ is strongly $(\tau,\sigma)$-elastic, then the corresponding $F(\phi,y)$ satisfies that 
  \begin{equation} \label{Eq:ElasticFBound}
    \frac{\phi^{\sigma-1} - y^{\sigma-1}}{(\sigma-1)\phi^{\sigma-1}}
  \leq F(\phi,y) \leq \frac{\phi^{\tau-1} - y^{\tau-1}}{(\tau-1)\phi^{\tau-1}}, \quad \forall \phi \geq y > 0.
  \end{equation}
\end{corollary}
\begin{proof}
  Let $f_1 = f$ and $f_2(y) = y^{\tau}$. 
  Since $f$ is strongly $(\tau,\sigma)$-elastic, by the definition, we have $f_1$ and $f_2$ satisfy the condition in Lemma \ref{Lem:ElasticComp}. Thus, we have $F(\phi,y) \leq F_2(\phi,y)$ for all $\phi \geq y > 0$. Note that for $f_2(y) = y^{\tau}$, we have
  $$
    F_2(\phi,y) = \frac{\phi^{\tau-1} - y^{\tau-1}}{(\tau-1)\phi^{\tau-1}}.
  $$
  This implies the upper bound in Eq. \eqref{Eq:ElasticFBound}. Similarly, by letting $f_2(y) = y^{\sigma}$, one can show the lower bound in Eq. \eqref{Eq:ElasticFBound}.
\end{proof}

\subsection{Local Phase Portraits and Proof of the Existence of Solutions} \label{Apx:Proof of the Existence of a Solution}
Recall that $\alpha \geq \asigmas$. 
We first analyze the behavior of the solutions to Eq. \eqref{eq:SufNec_phi} near the origin in order to probe potential (globally) feasible solutions. Then we show that there exist solutions that are feasible.

We claim that for the Taylor expansion of $f'(y)$, the lowest degree of the expansion must be at least $\tau-1$.
We prove this by contradiction. Suppose that the lowest degree of $f'(y)$ is $m - 1 < \tau -1$. The corresponding coefficient is positive since $f$ is convex. On the other hand, the lowest degree of $f''(y)$ is $m - 2 < \tau - 2$, and the function $f''(y)/y^{\tau-2}$ must be decreasing for $y>0$ sufficiently small. This contradicts the fact that $f$ is strongly $(\tau,\sigma)$-elastic. Thus, the claim holds. 
We assume w.l.o.g. that $\tau>1$ and
$$
f^\prime(y) = c_{\tau -1}y^{\tau - 1} + o(y^{\tau - 1}), \quad c_{\tau-1} > 0, 
$$
where $o(y^{\tau - 1})$ denotes the higher-order terms of $y^{\tau - 1}$. For $\tau = 1$ or the case where the lowest degree of the expansion of $f'$ is strictly larger than $\tau - 1$, we consider the lowest term of $f''$, which must have a positive coefficient, and the sequel analysis still applies.
Then we have 
$$
f^{\prime\prime}(y) = (\tau - 1) c_{\tau - 1} y^{\tau-2} + o(y^{\tau - 2}).
$$

We first focus on the following ODE only with the initial condition $ \phi(0) = 0 $:
\begin{equation}
  \left\{
  \begin{aligned}
      &\phi^\prime(y) = \alpha \cdot F(\phi,y), \quad \forall y\geq 0, \\
      &\phi(0) = 0.
  \end{aligned}
  \right.
  \label{eq:PolyDiffEq}
\end{equation} 
We consider the following dynamical system with the origin as an equilibrium point: 
\begin{equation}
  \left\{\begin{aligned}
    & \frac{\mathrm{d}\phi}{\mathrm{d}t} = \alpha \left(f'(\phi) - f'(y)\right),\\
    & \frac{\mathrm{d}y}{\mathrm{d}t} = \phi f''(\phi),
  \end{aligned}\right. \label{eq:DynSystem}
\end{equation}
where $y = y(t)$ and $\phi = \phi(t)$ are parameterized with respect to $t\in \mathbb{R}$. Then any trajectory of Eq. \eqref{eq:DynSystem} with the origin being its limit point is a solution to Eq. \eqref{eq:PolyDiffEq}. We prove that such solutions exist by giving a qualitative description of the local phase portrait near the origin. 

For the special case of $\tau = 2$, we can study Eq. \eqref{eq:DynSystem} via its linearized system (see Appendix \ref{Linearized System}). Here we discuss the general case where $\tau >2$. 
Since the linear part of the system Eq. \eqref{eq:DynSystem} vanishes at the origin when $\tau>2$, we apply the technique of homogeneous blow-up (for details, see Chapter 3 in \cite{Dumortier2006}), which uses polar coordinates to blow-up the singularity onto a circle in order to study the local phase portrait. More precisely, applying transformation $(y,\phi) = (r\cos\theta,r\sin\theta)$, where $(r,\theta)\in \mathbb{R}_{\geq 0}\times \mathbb{S}^1$, we can write Eq. \eqref{eq:DynSystem} as: 
\begin{equation}
  \left\{\begin{aligned}
    & \frac{\mathrm{d}r}{\mathrm{d}t} = \left(\alpha \left(\sin^{\tau}\theta - \cos^{\tau - 1}\theta\sin \theta\right) + (\tau-1)\cos\theta\sin^{\tau-1}\theta\right) c_{\tau-1}r^{\tau-1} + o(r^{\tau - 1}),\\
    & \frac{\mathrm{d}\theta}{\mathrm{d}t} = \left(\alpha \cos\theta \left(\sin^{\tau-1}\theta - \cos^{\tau-1}\theta\right) - (\tau-1)\sin^{\tau}\theta\right) c_{\tau-1}r^{\tau} + o(r^{\tau}).
  \end{aligned}\right. \label{eq:DynSystemPolar}
\end{equation}
One can then apply a time transformation $r^{\tau-2}\mathrm{d}t = \mathrm{d}s$ to further simplify Eq. \eqref{eq:DynSystemPolar}. Here we skip the details on analyzing the elementary singularities on the circle $\{0\}\times \mathbb{S}^1$, and present the local phase portraits in Figure \ref{fig:blowup}. There are two types of phase portrait because the elementary singularities on $\{0\}\times \mathbb{S}^1$ are different for odd and even $k$ due to Lemma \ref{Lem:PowerEquation}. 
\begin{figure}
  \centering
  \begin{subfigure}{0.45\textwidth}
    \includegraphics[width=1\linewidth]{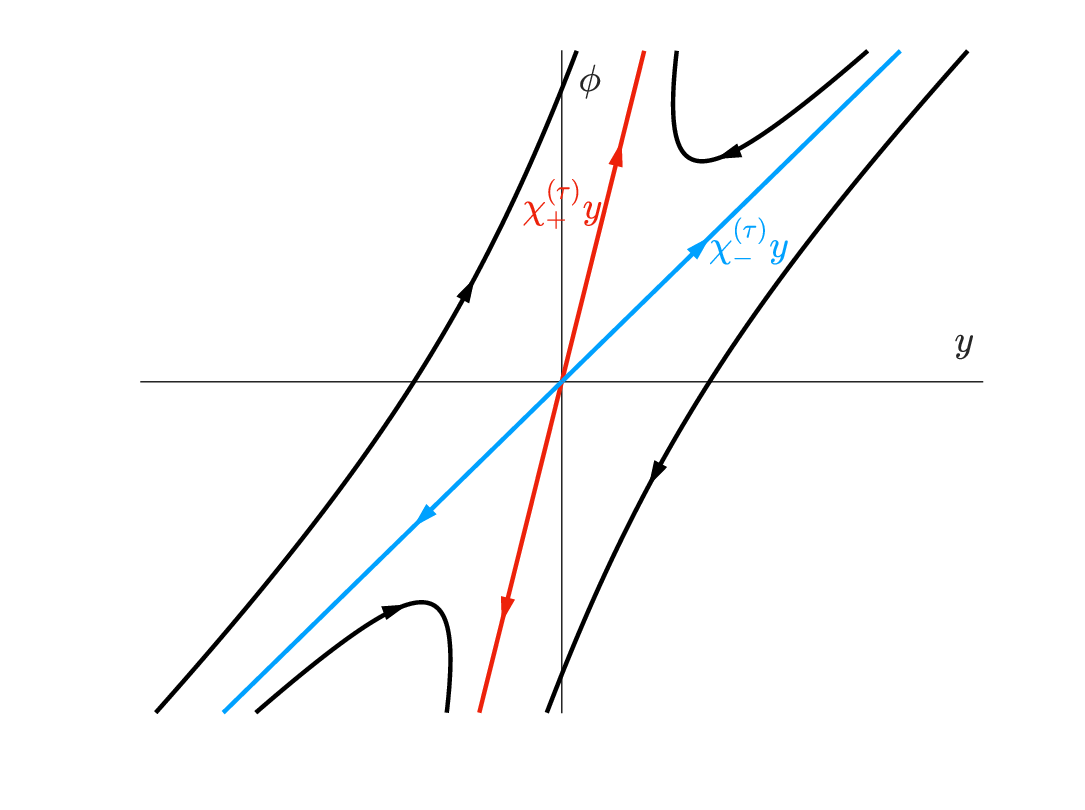}
    \caption{When $k$ is Odd} \label{fig:blowup_odd}
  \end{subfigure}
  ~
  \begin{subfigure}{0.45\textwidth}
    \includegraphics[width=1\linewidth]{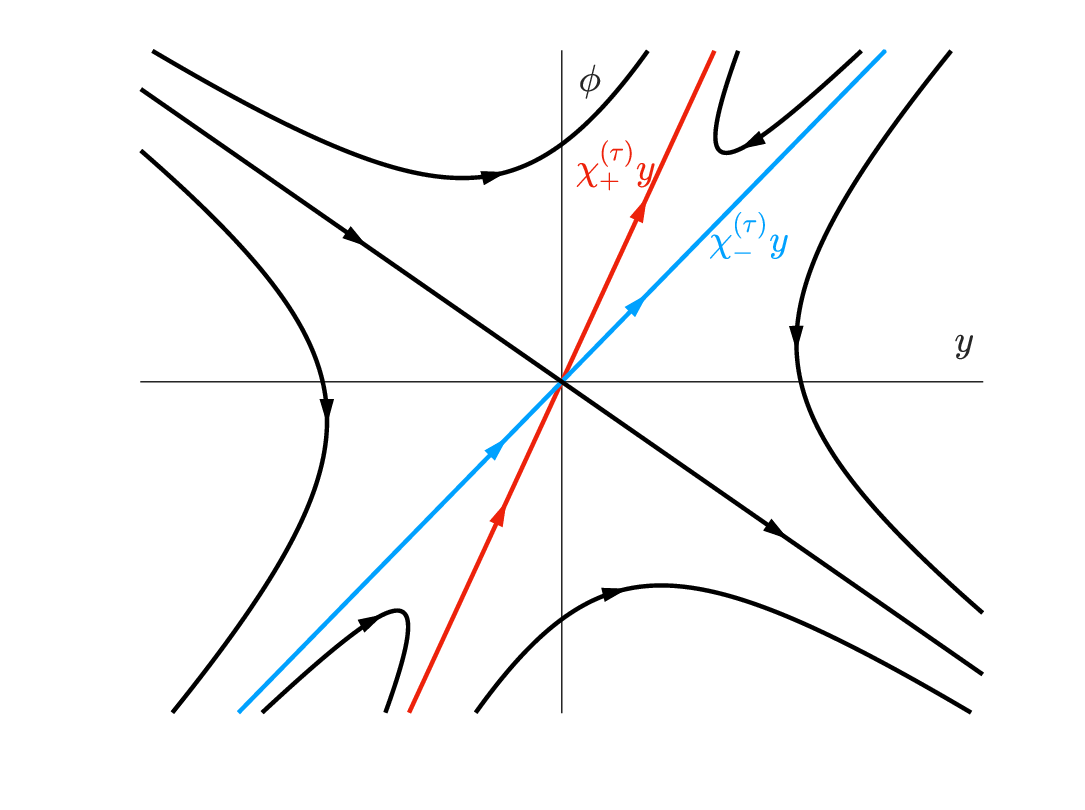}
    \caption{When $k$ is Even} \label{fig:blowup_even}
  \end{subfigure}
  \caption{We show two types of local phase portraits of dynamical system Eq. \eqref{eq:DynSystem}. Both indicate that there are trajectories starting from the origin with initial slope of $\chi^{(\tau)}_{\pm}$.}
  \label{fig:blowup}
\end{figure}

Obviously, there are trajectories with the origin being their limit point, which are solutions to Eq. \eqref{eq:PolyDiffEq}. We next prove that among these solutions, there exists at least one that is superlinear (and thus it becomes a feasible reserve function to Eq. \eqref{eq:SufNec_phi}). 

We consider trajectories with an initial slope of $\chi^{(\tau)}_+$ at the origin. We claim that these trajectories must be lower bounded by the linear function $\Delta^{(\sigma)}_{+} y$, and hence, they are superlinear. We prove this by contradiction. Otherwise, there exists a trajectory, denoted as $\tilde{\phi}$, such that $\tilde{\phi}(y_0) < \Delta^{(\sigma)}_{+} y_0$ for some $y_0>0$. By Lemma \ref{Lem:PowerEquation}, we have $\chi^{(\tau)}_+>\Delta^{(\sigma)}_{+}>1$. Then there exists $y_1\in (0,y_0)$ such that $\tilde{\phi}(y_1) = \Delta^{(\sigma)}_{+} y_1$ and 
\begin{equation}
  \tilde{\phi}(y) < \Delta^{(\sigma)}_{+} y, \quad \forall y_1<y\leq y_0. \label{ineq:SuperlinContr}
\end{equation}
We know that $\varphi(y) = \Delta^{(\sigma)}_{+} y$ is a solution to the following equation:
\begin{equation}
  \left\{\begin{aligned}
    & \varphi^\prime = \alpha\cdot \frac{\varphi^{\sigma-1} - y^{\sigma-1}}{\left(\sigma-1\right)\varphi^{\sigma-1}}, \quad y\geq y_1,\\
    & \varphi(y_1) = \tilde{\phi}(y_1).
  \end{aligned}\right. \label{eq:CompsigmaODE}
\end{equation}
Comparing Eq. \eqref{eq:PolyDiffEq} and Eq. \eqref{eq:CompsigmaODE}, by Corollary \ref{Cor:ElasticFBound}, we have
\begin{equation}
  \alpha \cdot F(\phi,y) \geq \alpha \cdot \frac{\left(\phi^{\sigma-1} - y^{\sigma-1}\right)}{\left(\sigma-1\right)\phi^{\sigma-1}}, \label{ineq:SimpLowOrdODE}
\end{equation}
where the equality holds if and only if the cost function has only one term, i.e., $\tau = \sigma$ (this is discussed in \cite{Huang2019} and \cite{Tan2020a}). Here we focus on the case where $\tau < \sigma$. In this case, Eq. \eqref{ineq:SimpLowOrdODE} becomes a strict inequality, and by the ODE comparison theorem (see Comparison Theorem  \cite{Walter1998}, p.~90), we have 
$$
\tilde{\phi}(y) > \Delta^{(\sigma)}_{+} y, \quad \forall y_1<y<y_0,
$$
which contradicts to Eq. \eqref{ineq:SuperlinContr}. Thus, trajectories with an initial slope of $\chi^{(\tau)}_+$ are lower bounded by $\Delta^{(\sigma)}_{+} y$ for any $y\geq 0$, and are, of course, superlinear.

\subsection{Proof of Infinitely Many Solutions} \label{Apx:Proof of Infinitely Many Solutions}

We first prove the following lemma. 
\begin{lemma}
  \label{Lem:LowerBoundforFeasibleSol}
  For any $Y>0$ and $\alpha \geq \asigmas$, define $\phi_Y(y)$ as the solution to the following ODE:
  \begin{equation}
    \label{eq:LowerBoundforFeasibleSol}
    \left\{
    \begin{aligned}
      &\phi^{\prime}(y) = \alpha \cdot F(\phi, y), \quad \forall y\geq 0, \\
      &\phi(Y) = Y.
    \end{aligned}
    \right.
  \end{equation}
  Then we have $\phi_Y(y) > y$ for $y\in (0,Y)$ with $\phi_Y(0) = 0$ and $\phi_Y^{\prime}(0) = \chi^{(\tau)}_-$. 
\end{lemma}
\begin{proof}
  We prove $\phi_Y(y) > y$ for $y\in (0,Y)$ by contradiction. Otherwise, $\phi_Y(y)$ intersects with $\phi(y) = y$ on the interval $(0,Y)$. Let $\phi_Y(\tilde{Y}) = \tilde{Y}$ and, assume w.l.o.g. that there is no other intersecting point within the interval $(\tilde{Y},Y)$. Note that $F(y,y) = 0$. Since $\phi_Y^{\prime}(\tilde{Y}) = 0$, we have $\phi_Y(y)<y$ for $y\in (\tilde{Y},Y)$. On the other hand, we have $\phi_Y^{\prime}(y) < 0$ as $y \rightarrow Y$, which contradicts $\phi_Y(Y) = Y$. Therefore, $\phi_Y(y) > y$ for $y\in (0,Y)$. This implies $\phi_Y(0) \geq 0$. However, $\phi_Y(0)$ can not be strictly larger than 0; otherwise, $\phi_Y$ must intersect with solutions to Eq. \eqref{eq:SufNec_phi} with initial slope of $\chi^{(\tau)}_+$ at 0. Moreover, $\phi_Y^{\prime}(0) \neq \chi^{(\tau)}_+$ as solutions with initial slope of $\chi^{(\tau)}_+$ are lower bounded by $\Delta^{(\sigma)}_{+} y$ for $y\geq 0$. Thus, $\phi_Y^{\prime}(0) = \chi^{(\tau)}_-$. 
\end{proof}
We remark that a corollary of Lemma \ref{Lem:LowerBoundforFeasibleSol} is that the (point-wise) upper bound of the solution to Eq. \eqref{eq:LowerBoundforFeasibleSol} for any $Y>0$, is a solution to Eq. \eqref{eq:SufNec_phi}.

We now prove there are infinitely many solutions to Eq. \eqref{eq:SufNec_phi} satisfy that their initial slope is $\chi^{(\tau)}_-$ and each of them intersects with $\Delta^{(\sigma)}_{-} y$ at some $y>0$. 

For any $Y>0$, define $\phi_{Y,-}$ as the solution to the following ODE:
$$
\left\{
\begin{aligned}
  &\phi^{\prime}(y) = \alpha \cdot F(\phi, y), \quad \forall y\geq 0, \\
  &\phi(Y) = \Delta^{(\sigma)}_{-} Y.
\end{aligned}
\right.
$$
We have $\phi_{Y,-}(0)\leq 0$; otherwise, $\phi_{Y,-}$ intersects with solutions to Eq. \eqref{eq:SufNec_phi} with initial slope of $\chi^{(\tau)}_+$. On the other hand, $\phi_{Y,-}(0)\geq 0$; otherwise, $\phi_{Y,-}$ intersects with $\phi_Y$, solutions to Eq. \eqref{eq:LowerBoundforFeasibleSol}. Thus, we have $\phi_{Y,-}(0) = 0$. Moreover, $\phi_{Y,-}^{\prime}(0) \neq \chi^{(\tau)}_+$ as solutions with initial slope of $\chi^{(\tau)}_+$ are lower bounded by $\Delta^{(\sigma)}_{-} y$ for $y\geq 0$. Thus, $\phi_Y^{\prime}(0) = \chi^{(\tau)}_-$. 

Furthermore, we have $\phi_{Y,-}(y) > \phi_Y(y)$ for $y\in (0,Y)$; otherwise, $\phi_{Y,-}$ intersects with $\phi_Y$. Then by Lemma \ref{Lem:LowerBoundforFeasibleSol}, $\phi_{Y,-}(y) \geq \phi_Y(y) \geq y$ for $y\in [0,Y]$. Meanwhile, by the previous proof, $\phi_{Y,-}$ is lower bounded by $\Delta^{(\sigma)}_{-} y$ for $y> Y$. Thus, $\phi_{Y,-}$ is superlinear.

\subsection{Proof for Tight Linear Bounds} \label{Apx:Proof for Tight Linear Bounds}

We prove the tight upper bound and the lower bound separately.
\begin{lemma}
    Solutions to Eq. \eqref{eq:SufNec_phi} for $\alpha \geq \asigmas$ are tightly upper bounded by the linear function $\chi^{(\tau)}_{+} y$ for $y\geq 0$.
\end{lemma}
\begin{proof}
  We only need to prove that for solutions to Eq. \eqref{eq:SufNec_phi} with an initial slope of $\chi^{(\tau)}_{+}$ at 0, the linear function $\chi^{(\tau)}_{+} y$ gives the tight upper bound. 
  By Theorem~\ref{theorem_lower_bound}, any solution to Eq.~\eqref{eq:SufNec_phi} has a bounded derivative on $\mathbb{R}_+$, and thus it must be upper bounded by some linear function (the same for the lower bound).
  Denote $k^*>0$ as the tight upper bound. We now prove that $k^* = \chi^{(\tau)}_{+}$. 

  First of all, we have $k^* \geq \chi^{(\tau)}_{+}$. Denote $\phi$ as a solution to Eq. \eqref{eq:SufNec_phi} with an initial slope of $\chi^{(\tau)}_{+}$. Then for any $0<k< \chi^{(\tau)}_{+}$,
  there exist a neighborhood at 0 such that $\phi(y) > ky$. 
  
  On the other hand, we prove $k^* \leq \chi^{(\tau)}_{+}$ by contradiction. Otherwise, there exists $k > \chi^{(\tau)}_{+}$ such that $ky$ is not an upper bound. This implies, for some $\phi$, there exists $y_0 > 0$ such that $\phi(y_0) = y_0$ and $\phi(y)< ky$ for $y\in (0,y_0)$. 
  By Lemma \ref{Lem:PowerEquation}, there exists $\gamma>\alpha$ such that $\varphi(y) = ky$ is a solution to 
  $$
  \left\{
  \begin{aligned}
    &\varphi^{\prime}(y) = \gamma \cdot \frac{c_{\tau-1}\left(\varphi^{\tau-1} - y^{\tau-1}\right)}{c_{\tau-1}\left(\tau-1\right)\varphi^{\tau-1}}, \quad \forall y\geq 0, \\
    &\varphi(y_0) = ky_0.
  \end{aligned}
  \right.  
  $$
  By Corollary \ref{Cor:ElasticFBound}, we have 
  $$
  \alpha\cdot F(\phi,y) < \gamma\cdot F(\phi,y) \leq \gamma \cdot \frac{c_{\tau-1}\left(\phi^{\tau-1} - y^{\tau-1}\right)}{c_{\tau-1}\left(\tau-1\right)\phi^{\tau-1}},
  $$
  for $\phi>y>0$. Applying the ODE comparison theorem (see Comparison Theorem \cite{Walter1998}, p. 90), we have $\phi(y) < ky$ for $y > y_0$. Therefore, we have $ky$ is a linear upper bound, which leads to a contradiction. Thus, $k^* \leq \chi^{(\tau)}_{+}$. We then obtain that $k^* = \chi^{(\tau)}_{+}$.
\end{proof}
We also prove the tight lower bound.
\begin{lemma}
    Solutions to Eq. \eqref{eq:SufNec_phi} for $\alpha \geq \asigmas$ are tightly lower bounded by the linear function $\chi^{(\tau)}_{-} y$ for $y\geq 0$.
\end{lemma}
\begin{proof}
  We only need to prove that for solutions to Eq. \eqref{eq:SufNec_phi} with an initial slope of $\chi^{(\tau)}_{-}$ at 0, the linear function $\chi^{(\tau)}_{-} y$ gives the tight lower bound. Denote $k^*>0$ as the tight lower bound. We now prove that $k^* = \chi^{(\tau)}_{-}$. 

  First of all, we have $k^* \leq \chi^{(\tau)}_{-}$. Denote $\phi$ as a solution to Eq. \eqref{eq:SufNec_phi} with an initial slope of $\chi^{(\tau)}_{-}$. Then for any $k> \chi^{(\tau)}_{-}$,
  there exist a neighborhood at 0 such that $\phi(y)> ky$ for $y\in (0,y_0)$. 
  
  On the other hand, we prove $k^* \geq \chi^{(\tau)}_{-}$ by contradiction. Otherwise, there exists $0<k<\chi^{(\tau)}_{-}$ such that $ky$ is not a lower bound. This implies, for some solution $\phi$, there exists $y_0>0$ such that $\phi(y_0) = ky_0$.  
  By Lemma \ref{Lem:PowerEquation}, there exists $\gamma>\alpha$ such that $\varphi(y) = ky$ is the solution to 
  $$
  \left\{
  \begin{aligned}
    &\varphi^{\prime}(y) = \gamma \cdot \frac{c_{\tau-1}\left(\varphi^{\tau-1} - y^{\tau-1}\right)}{c_{\tau-1}\left(\tau-1\right)\varphi^{\tau-1}}, \quad \forall y\geq 0, \\
    &\varphi(y_0) = k y_0.
  \end{aligned}
  \right.  
  $$
  By Corollary \ref{Cor:ElasticFBound}, we have 
  $$
  \alpha\cdot F(\phi,y) < \gamma\cdot F(\phi,y) \leq \gamma \cdot \frac{c_{\tau-1}\left(\phi^{\tau-1} - y^{\tau-1}\right)}{c_{\tau-1}\left(\tau-1\right)\phi^{\tau-1}},
  $$
  for $\phi>y>0$. Applying the ODE comparison theorem (see Comparison Theorem \cite{Walter1998}, p. 90), we have 
  $$
  y \leq \phi(y) < ky, \quad \forall y > y_0,
  $$
  where the first inequality holds because $\phi(y)$ is a feasible solution to Eq. \eqref{eq:SufNec_phi}. However, by Eq. \eqref{eq:SufNec_phi}, we have $\phi^{\prime}(y) \rightarrow \Delta^{(\sigma)}_{\pm} > k$ as $y\rightarrow +\infty$. This implies there exists sufficiently large $y_1>0$ such that $\phi(y_1)> ky_1$, which leads to a contradiction. Thus, $k^* \geq \chi^{(\tau)}_{-}$. We then obtain that $k^* = \chi^{(\tau)}_{-}$.
\end{proof}

We thus complete the proof of Theorem \ref{Thm:UpperBoundAch}. 

\subsection{Proof of Proposition \ref{Prop:BoundResp} (Tight Linear Bounds for $\phi_{\mathrm{ub}}$ and $\phi_{\mathrm{lb}}$)} \label{Apx:BoundResp}

  See Definition \ref{Def:TightBound} for the definition of tight linear bounds. 
  In Theorem \ref{Thm:UpperBoundAch}, we have shown that $\chi^{(\tau)}_{\pm} y$ give the tight bounds for the solutions to Eq. \eqref{eq:SufNec_phi}. We now show that $\Delta^{(\sigma)}_{-} y$ is the tight upper bound for $\phi_{\mathrm{lb}}$. 
  
  First, we show that $\Delta^{(\sigma)}_{-} y$ is an upper bound for $\phi_{\mathrm{lb}}$. We prove this by contradiction. Otherwise, since $\Delta^{(\sigma)}_{-} > \chi^{(\tau)}_{-}$, there exists $y_0>0$ such that $\phi_{\mathrm{lb}}(y_0) = \Delta^{(\sigma)}_{-} y_0$. Define $\phi$ as the solution to the following ODE 
  $$
  \left\{
  \begin{aligned}
    &\phi^{\prime}(y) = \alpha \cdot F(\phi, y), \quad \forall y\geq 0, \\
    &\phi(2y_0) = 2\Delta^{(\sigma)}_- y_0.
  \end{aligned}
  \right.
  $$
  This implies $\phi(y)< \Delta^{(\sigma)}_- y$ for $y\in (0,2y_0)$, and thus, $\phi(y_0) < \phi_{\mathrm{lb}}(y_0)$. We then have $\phi(y) < \phi_{\mathrm{lb}}(y)$ for $y>y_0$; otherwise $\phi$ and $\phi_{\mathrm{lb}}$ must intersect.
  Therefore, for any $y>0$,
  $$
  \phi(y) < \phi_{\mathrm{lb}}(y),
  $$
  which contradicts the definition of the lower-bound solution $\phi_{\mathrm{lb}}$. Thus, $\Delta^{(\sigma)}_{-} y$ is an upper bound for $\phi_{\mathrm{lb}}$. For the tightness, similar to the proof of Theorem \ref{theorem_lower_bound} (see Appendix~\ref{Apx:LowerBound}), we have 
  $$
  \lim_{y\rightarrow +\infty} \frac{\phi_{\mathrm{lb}}(y)}{y} = \Delta^{(\sigma)}_{\pm}. 
  $$
  Since $\Delta^{(\sigma)}_{-} y$ is an upper bound for $\phi_{\mathrm{lb}}$, the above limit must be $\Delta^{(\sigma)}_{-}$, which implies the tightness. 

  We have proved that $\Delta^{(\sigma)}_{+} y$ is the lower bound for $\phi_{\mathrm{ub}}$ in Appendix \ref{Apx:Proof of the Existence of a Solution}. Following the same analysis, one can prove the tightness. Thus, we complete the proof of Proposition~\ref{Prop:BoundResp}.

\section{Stable Numerical Methods for Computing the Boundary of  $\mathcal{S}(\asigmas,f)$}
\label{sec_numerical_method}

In general, the upper- and lower-bound solutions, $ \phi_{\mathrm{ub}} $ and $ \phi_{\mathrm{lb}} $, do not admit closed-form expressions, leaving open the question of how to compute them numerically. In what follows, we first discuss why off-the-shelf ODE solvers may fail in this context, and then present stable numerical methods with theoretically justified bounds on their approximation error.

Recall that Eq.~\eqref{eq:mainODE} exhibits a singularity at the origin, which poses challenges for both theoretical analysis and stable numerical computation. To gain a clearer understanding of this difficulty, we consider the two-dimensional dynamical system Eq.~\eqref{eq:MainDynSystem}, which is equivalent to Eq.~\eqref{eq:SufNec_phi}.
Recall that the solutions to Eq.~\eqref{eq:mainODE} correspond to the trajectories of this dynamical system that pass through the origin. As a non-Hamiltonian system, it is notoriously difficult (see \citet{Dumortier2006}) to identify first integrals\footnote{For a two-dimensional dynamical system such as Eq.~\eqref{eq:MainDynSystem}, the existence of a first integral—an implicit solution that remains constant along trajectories—completely characterizes the system’s phase portrait.} of Eq.~\eqref{eq:MainDynSystem}. Although one may attempt to solve Eq.~\eqref{eq:MainDynSystem} numerically using standard ODE solvers, identifying the upper- and lower-bound solutions among the infinitely many possible trajectories is nontrivial. Owing to the feasibility constraints in Eq.~\eqref{eq:SufNec_phi}, it is not immediately evident whether a trajectory initialized at the origin satisfies the required conditions. As a result, standard ODE solvers typically yield a single trajectory, which may not be feasible. This motivates the development of numerically stable approximation methods specifically tailored to compute the upper- and lower-bound solutions.

\subsection{Approximate Methods to Compute the Lower- and Upper-Bound Solutions} 
We start by presenting a parameterized method that can approximate the lower-bound solution $\phi_{\mathrm{lb}}$ arbitrarily closely over any bounded interval. For any $\xi>0$, let $\phi_{\mathrm{lb},\xi}$ be the solution to
\begin{align}\label{eq:LowerSol}
  \begin{cases} 
  \phi'(y) = \asigmas \cdot F(\phi, y), \quad \text{for all } 0 \leq y \leq \chi^{(\tau)}_- \xi, \\
  \phi(\xi) =  \chi^{(\tau)}_- \xi.
  \end{cases}
\end{align}
That is, $\phi_{\mathrm{lb},\xi}$ is the solution to Eq.~\eqref{eq:mainODE} which passes the point $\left(\xi, \chi^{(\tau)}_- \xi\right)$. 
Since $\phi_{\mathrm{lb},\xi}$ is not a solution to Eq.~\eqref{eq:SufNec_phi}, it is not $\asigmas$-competitive and may violate constraints Eq.~\eqref{eq:SufNec_phi_mono} and Eq.~\eqref{eq:SufNec_phi_superlin} (as a solution to Eq.~\eqref{eq:MainDynSystem}, $\phi_{\mathrm{lb},\xi}(0)=0$).
However, Proposition \ref{Prop:LowerSol} below shows that we can use $\phi_{\mathrm{lb},\xi}$ to approximate $\phi_{\mathrm{lb}}$. The proof is given in Appendix \ref{Apx:NumMethdProofs}.
\begin{proposition}[Approximate Computation of $\phi_{\mathrm{lb}}$]
  \label{Prop:LowerSol}
  For any small $\varepsilon>0$ and any $Y>0$, there exists $\xi>0$ such that the reserve function $\phi_{\mathrm{lb},\xi}$ is $\asigmas$-competitive on $[0,Y]$ and satisfies
  $$
  |\phi_{\mathrm{lb},\xi}(y) - \phi_{\mathrm{lb}}(y)| < \varepsilon, \quad \forall y\in[0,Y].
  $$
  Specifically, the reserve function $\phi_{\mathrm{lb},\xi}$ is $\asigmas$-competitive for some interval that contains $\left[0,\chi^{(\tau)}_- \xi\right]$.
\end{proposition}
Additionally, if we assume that the request value of an instance is bounded, we can estimate the competitive ratio for the reserve function $\phi_{\mathrm{lb},\xi}$ with respect to the infinite supply. This is given by Corollary~\ref{Cor:LowerSol_vmax}, and the proof is provided in Appendix \ref{Apx:NumMethdProofs}. 
\begin{corollary} \label{Cor:LowerSol_vmax}
  For any bounded instance with the highest request value of $v_{\rm max}$, $ \PUM\phi $ with $ \phi = \phi_{\mathrm{lb},\xi}$ is $\tilde{\alpha}$-competitive, where 
  \begin{equation*}
      \tilde{\alpha} = \max\left\{\frac{f^*\left(v_{\rm max}\right)}{f^*\left(f'\left(\chi^{(\tau)}_- \xi\right)\right)},1\right\} \cdot \asigmas.
  \end{equation*}
\end{corollary}

For the upper-bound solution $\phi_{\mathrm{ub}}$, we have the following approximate method. For any small $\eta>0$, let $\phi_{\mathrm{ub},\eta}$ be the solution to 
\begin{align}\label{eq:UpperSol}
\begin{cases} 
\phi'(y) = \asigmas \cdot F(\phi, y), \quad \text{for all } y \geq 0, \\
\phi(0) =  \eta.
\end{cases}
\end{align}
We have the following proposition to characterize how $\phi_{\mathrm{ub},\eta}$ approaches $\phi_{\mathrm{ub}}$.

\begin{proposition}[Approximate Computation of $\phi_{\mathrm{ub}}$]
  \label{Prop:AprUpSol}
  For any small $\varepsilon>0$ and any $Y>0$, there exists $\eta>0$ such that the reserve function $\phi_{\mathrm{ub},\eta}$ is $\asigmas$-competitive with a constant $\beta$ and satisfies
  $$
  |\phi_{\mathrm{ub},\eta}(y) - \phi_{\mathrm{ub}}(y)| < \varepsilon, \quad \forall y\in[0,Y],
  $$  
  where 
  $$
  \beta = \frac{1}{\asigmas}\left(\eta f^{\prime}(\eta) - f(\eta)\right).
  $$
  That is, $ \PUM\phi $ with $ \phi = \phi_{\mathrm{ub},\eta}$ is $ \asigmas $-competitive with $ \beta $ additive loss, namely, it satisfies $\ALG(\mathcal{I}) \geq \frac{1}{\alpha}\OPT(\mathcal{I}) - \beta$ for any instance $ \mathcal{I} $. 
\end{proposition}
The proof of Proposition~\ref{Prop:AprUpSol} is provided in Appendix~\ref{Apx:NumMethdProofs}. This approximation method is numerically stable, and thus a smaller value of $\eta$ yields a more accurate approximation of the upper-bound solution $\phi_{\mathrm{ub}}$, as well as a smaller constant $\beta$. However, due to the inherent limitations in the accuracy of numerical ODE solvers, $\eta$ cannot be made arbitrarily small. A similar limitation arises in the context of Proposition~\ref{Prop:LowerSol}.

\subsection{Proofs Related to Stable Numerical Methods} \label{Apx:NumMethdProofs}

We begin by reviewing a prior result that characterizes pricing functions achieving an $ \alpha $-competitive ratio up to an additive constant. We then approximate the lower-bound solution and derive estimates for the corresponding approximate solutions. A similar analysis is carried out for the upper-bound solutions.

Prior work by \citet{Huang2019} establishes the following characterization of when a pricing function is $ \alpha $-competitive with an additive constant $\beta$. Recall that the pricing function is defined as $\varPhi(y) = f'(\phi(y))$.
\begin{lemma}
For any $ \alpha \geq 1 $, $ \PUM\phi $ is $ \alpha $-competitive if and only if its corresponding pricing function $\varPhi$ is increasing and satisfies the following inequality for some constant $\beta$:
    \begin{equation} \label{ineq:PricingFunctionCharac}
        \int_{0}^{y} \varPhi(u)\Id u - f(y) \geq  \frac{1}{\alpha} f^{*}\left(\varPhi\left(y\right)\right) - \beta, \quad \forall y\geq 0.
    \end{equation}
\end{lemma}

This lemma helps us give an estimation of the approximation solutions corresponding to the upper-bound solution.

\begin{proof}[Proof of Proposition \ref{Prop:LowerSol}]
  
  We prove that for sufficiently large $\xi$, $\phi_{\mathrm{lb},\xi}$ can be arbitrarily close to $\phi_{\mathrm{lb}}$ on $[0,Y]$. Let $X_0 = \varepsilon/\Dstar$. Divide $\mathbb{R}_{\geq 0} = [0, X_0)\cup [X_0, +\infty)$. For the first interval, if $\xi \geq X_0$, we have $\phi_{\mathrm{lb},\xi}(y) \leq \Dstar y$ for $y\in [0,X_0)$; otherwise, $\phi_{\mathrm{lb},\xi}(y)$ must intersect with $\Dstar y$ for twice, which contradicts $\phi_{\mathrm{lb},\xi}(y)>\Dstar y$ after intersecting with $\Dstar y$ (see the proof of Theorem \ref{Thm:UpperBoundAch}). This implies
  $$
  |\phi_{\mathrm{lb},\xi}(y) - \phi_{\mathrm{lb}}(y)| \leq \Dstar y < \varepsilon, \quad \forall y\in[0,X_0). 
  $$
  We then prove that the inequality holds for the interval $[X_0, +\infty)$. 
  
  Let $Y_L = \phi_{\mathrm{lb}}\left(X_0\right)$. For any $Y_0\in \left[\chi^{(\tau)}_- X_0, Y_L\right]$, define $\phi_0(Y_0,y)$ as the solution to 
  $$
  \begin{aligned}
  \begin{cases} 
  \phi'(y) = \alpha \cdot F(\phi, y), \quad \text{for all } y \geq 0, \\
  \phi(X_0) =  Y_0.
  \end{cases}
  \end{aligned}
  $$
  By the continuous dependence of solutions on initial values (see Theorem on Continuous Dependence \cite{Walter1998}, p. 145), there exists $Y_\xi\in \left[\chi^{(\tau)}_- X_0, Y_L\right)$ such that 
  $$
  \phi_{\mathrm{lb}}\left(y\right) - \phi_0(Y_\xi,y) < \varepsilon, \quad \forall y\geq X_0.
  $$
  Since $Y_\xi < Y_L$, we claim that there exists $y_0>0$ such that $\phi_0(Y_\xi,y_0) = \chi^{(\tau)}_- y_0$, i.e., the solution satisfying the initial condition $\phi(X_0) =  Y_\xi$ intersects with the line $\chi^{(\tau)}_- y$ at $y_0$. Otherwise, for any $y>0$, we have 
  $$
  \phi_0(Y_\xi,y) > \chi^{(\tau)}_- y,
  $$
  and thus, $\phi_0(Y_\xi,y)$ is a feasible solution to Eq.~\eqref{eq:SufNec_phi}. On the other hand, since $\phi_{\mathrm{lb}}$ is the lower-bound solution, we have for any $y>0$,
  $$
  \phi_{\mathrm{lb}}(y) < \phi_0(Y_\xi,y),
  $$
  which contradicts that 
  $$
  \phi_0(Y_\xi,X_0) = Y_\xi < Y_L = \phi_{\mathrm{lb}}\left(X_0\right).
  $$
  When taking $Y_\xi$ sufficiently close to $Y_L$, we have $y_0 \geq Y$.
  Then choose $\xi\geq y_0$ and $\phi_{\mathrm{lb}, \xi}$ satisfies
  $$
  |\phi_{\mathrm{lb},\xi}(y) - \phi_{\mathrm{lb}}(y)| < \varepsilon, \quad \forall y\in[0,Y].
  $$

  We now prove the rest of this proposition. 
  Since $\phi_{\mathrm{lb},\xi}$ is not a feasible solution to Eq.~\eqref{eq:SufNec_phi}, there exists $y_1>0$ such that $\phi_{\mathrm{lb},\xi}(y_1) = y_1$. By Lemma~\ref{Lem:LowerBoundforFeasibleSol}, we have 
  $$
  y_1 > \chi^{(\tau)}_- \xi\geq y_0\geq Y,
  $$
  and therefore, $\phi_{\mathrm{lb},\xi}$ satisfies Eq.~\eqref{eq:SufNec_phi} on the interval $\left[0,\chi^{(\tau)}_- \xi\right]$. Thus, $\phi_{\mathrm{lb},\xi}$ is $\asigmas$-competitive, and we complete the proof.
\end{proof}

\begin{proof}[Proof of Corollary \ref{Cor:LowerSol_vmax}]

Eq. \eqref{eq:LowerSol} is equivalent to 
\begin{equation} \label{eq:LowerBoundODE}
    f^{\prime}(\phi_{\mathrm{lb},\xi}) - f^{\prime} = \frac{1}{\alpha} \phi_{\mathrm{lb},\xi}\cdot f^{\prime\prime}(\phi_{\mathrm{lb},\xi})\phi_{\mathrm{lb},\xi}^{\prime}.
\end{equation}
Since $\phi_{\mathrm{lb},\xi}$ is not a feasible solution to Eq.~\eqref{eq:SufNec_phi}, there exists $y_0>0$ such that $\phi_{\mathrm{lb},\xi}(y_0) = y_0$. By Lemma~\ref{Lem:LowerBoundforFeasibleSol}, we have $y_0 > \chi^{(\tau)}_- \xi$. 
Integrate Eq.~\eqref{eq:LowerBoundODE} and apply the transformation $ \varPhi_{\mathrm{lb},\xi}(y) = f'(\phi_{\mathrm{lb},\xi}(y)) $, or equivalently, $ \phi_{\mathrm{lb},\xi}(y) = f'^{-1}(\varPhi_{\mathrm{lb},\xi}(y))$:  
$$
\begin{aligned}
       \int_{0}^{y} \varPhi_{\mathrm{lb},\xi}(u)\Id u - f(y) = \frac{1}{\alpha} \int_{0}^{y} \phi_{\mathrm{lb},\xi}(u) f^{\prime\prime}\left(\phi_{\mathrm{lb},\xi}\left(u\right)\right)\phi_{\mathrm{lb},\xi}^{\prime}\left(u\right) \Id u.
\end{aligned}
$$
Note that $\phi_{\mathrm{lb},\xi}$ and $\varPhi_{\mathrm{lb},\xi}$ stop increasing after $y_0$. Therefore, if $v_{\rm max} > \varPhi_{\mathrm{lb},\xi}\left(\chi^{(\tau)}_- \xi\right)$, we can estimate the integral on the interval $\left[0, \chi^{(\tau)}_- \xi\right]$, i.e.,
$$
\begin{aligned}
       \int_{0}^{y} \varPhi_{\mathrm{lb},\xi}(u)\Id u - f(y) 
       &\geq \frac{1}{\alpha} \int_{0}^{\chi^{(\tau)}_- \xi} f^{\prime -1}\left(\varPhi_{\mathrm{lb},\xi}\left(u\right)\right)\varPhi_{\mathrm{lb},\xi}^{\prime}\left(u\right) \Id u \\
       & = \frac{1}{\alpha} \int_{0}^{\varPhi_{\mathrm{lb},\xi}\left(\chi^{(\tau)}_- \xi\right)} f^{* \prime}\left(t\right) \Id t \\
       & =  \frac{1}{\alpha} f^{*}\left(f'\left(\chi^{(\tau)}_- \xi\right)\right). 
\end{aligned}
$$
On the other hand, if $v_{\rm max} \leq \varPhi_{\mathrm{lb},\xi}\left(\chi^{(\tau)}_- \xi\right)$, we have 
$$
\begin{aligned}
       \int_{0}^{y} \varPhi_{\mathrm{lb},\xi}(u)\Id u - f(y) &= \frac{1}{\alpha} \int_{0}^{y} \phi_{\mathrm{lb},\xi}(u) f^{\prime\prime}\left(\phi_{\mathrm{lb},\xi}\left(u\right)\right)\phi_{\mathrm{lb},\xi}^{\prime}\left(u\right) \Id u \\
       &= \frac{1}{\alpha} \int_{0}^{y} f^{\prime -1}\left(\varPhi_{\mathrm{lb},\xi}\left(u\right)\right)\varPhi_{\mathrm{lb},\xi}^{\prime}\left(u\right) \Id u \\
       & =  \frac{1}{\alpha} f^{*}\left(f'\left(y\right)\right). 
\end{aligned}
$$
Thus, we complete the proof. 
\end{proof}

\begin{proof}[Proof of Proposition \ref{Prop:AprUpSol}]
  We prove that for sufficiently small $\eta$, $\phi_{\mathrm{ub},\eta}$ can be arbitrarily close to $\phi_{\mathrm{ub}}$ on $[0,Y]$. Let $X_0 = \varepsilon/\chi^{\tau}_+$. Divide $\mathbb{R}_{\geq 0} = [0, X_0)\cup [X_0, +\infty)$. 
  For the first interval, define $\phi_\varepsilon(y)$ as the solution to 
  $$
  \begin{aligned}
  \begin{cases} 
  \phi'(y) = \alpha \cdot F(\phi, y), \quad \text{for all } y \geq 0, \\
  \phi(X_0) =  \varepsilon.
  \end{cases}
  \end{aligned}
  $$
  Similarly as the proof of Theorem \ref{Thm:UpperBoundAch}, by Lemma \ref{Lem:ElasticComp} and the ODE comparison theorem (see Comparison Theorem \cite{Walter1998}, p. 90), we have $\phi_\varepsilon(y) > \chi^{\tau}_+ y$ for $y\in (0,X_0)$. Thus, $\phi_\varepsilon(y)$ is increasing on the interval $[0,X_0)$. 
  If $\eta \leq \phi_\varepsilon(0)$, we have $\phi_{\mathrm{ub},\eta}(y)\leq \phi_\varepsilon(y)$ for . This implies
  $$
  |\phi_{\mathrm{ub},\eta}(y) - \phi_{\mathrm{ub}}(y)| \leq \phi_\varepsilon(y) <  \phi_\varepsilon(X_0) = \varepsilon, \quad \forall y\in[0,X_0). 
  $$
  We then prove that the inequality holds for the interval $[X_0, +\infty)$.
  
  Let $Y_L = \phi_{\mathrm{ub}}\left(X_0\right)$. For any $Y_0\in [Y_L, \varepsilon]$, define $\phi_0(Y_0,y)$ as the solution to 
  $$
  \begin{aligned}
  \begin{cases} 
  \phi'(y) = \alpha \cdot F(\phi, y), \quad \text{for all } y \geq 0, \\
  \phi(X_0) =  Y_0.
  \end{cases}
  \end{aligned}
  $$
  By the continuous dependence of solutions on initial values (see Theorem on Continuous Dependence \cite{Walter1998}, p. 145), there exists $Y_\eta\in (Y_L, \varepsilon]$ such that 
  $$
  \phi_0(Y_\eta,y) - \phi_{\mathrm{ub}}\left(y\right) < \varepsilon, \quad \forall y\geq X_0.
  $$
  Since $Y_\eta > Y_L$, we claim that $\phi_0(Y_\eta,0) > 0$. Otherwise, we have $\phi_0(Y_\eta,0) \leq 0$. Obviously, $\phi_0(Y_\eta,0)$ can not be negative; otherwise, it must intersect with $\phi_{\mathrm{ub}}\left(y\right)$ on the interval $[0,X_0)$. On the other hand, if $\phi_0(Y_\eta,0) = 0$, we have
  $$
  \frac{\mathrm{d} \phi_0}{\mathrm{d}y}(Y_\eta,0) = \chi^{(\tau)}_{+}, 
  $$
  i.e., the initial slope of $\phi_0(Y_\eta,y)$ is $\chi^{(\tau)}_{+}$.
  Otherwise, the initial slope of $\phi_0(Y_\eta,y)$ is $\chi^{(\tau)}_{-}$ (see the local phase portrait in Figure \ref{fig:blowup}), and thus, it must intersect with $\phi_{\mathrm{ub}}\left(y\right)$ on the interval $[0,X_0)$.
  Then, for any $y>0$, we have 
  $$
  \sup \Big \{\phi(y) \mid \phi \text{ is a sol. to Eq.  \eqref{eq:SufNec_phi}} \text{ and } \phi^{\prime}(0)=\chi^{(\tau)}_{+} \Big\} \geq \phi_0(Y_\xi,y) > \phi_{\mathrm{lb}}\left(y\right), 
  $$ 
  which contradicts the fact that $\phi_{\mathrm{ub}}$ is the upper-bound solution. Choose $\eta\leq \phi_0(Y_\eta,0)$ and we complete the fist part of the proof. We now prove $\phi_{\mathrm{ub},\eta}$ is $\asigmas$-competitive with a constant $\beta$. 
  
    Eq. \eqref{eq:UpperSol} is equivalent to 
    $$
    f^{\prime}(\phi_{\mathrm{ub},\eta}) - f^{\prime} = \frac{1}{\alpha} \phi_{\mathrm{ub},\eta}\cdot f^{\prime\prime}(\phi_{\mathrm{ub},\eta})\phi_{\mathrm{ub},\eta}^{\prime}.
    $$
    Integrate both sides on $[0,y]$ and apply the transformation $ \varPhi_{\mathrm{ub},\eta}(y) = f'(\phi_{\mathrm{ub},\eta}(y)) $, or equivalently, $ \phi_{\mathrm{ub},\eta}(y) = f'^{-1}(\varPhi_{\mathrm{ub},\eta}(y))$:
    $$
    \begin{aligned}
       \int_{0}^{y} \varPhi_{\mathrm{ub},\eta}(u)\Id u - f(y) &= \frac{1}{\alpha} \int_{0}^{y} \phi_{\mathrm{ub},\eta}(u) f^{\prime\prime}\left(\phi_{\mathrm{ub},\eta}\left(u\right)\right)\phi_{\mathrm{ub},\eta}^{\prime}\left(u\right) \Id u \\
       &= \frac{1}{\alpha} \int_{0}^{y} f^{\prime -1}\left(\varPhi_{\mathrm{ub},\eta}\left(u\right)\right)\varPhi_{\mathrm{ub},\eta}^{\prime}\left(u\right) \Id u \\
       & = \frac{1}{\alpha} \int_{f^{\prime}(\delta)}^{\varPhi_{\mathrm{ub},\eta}(y)} f^{* \prime}\left(t\right) \Id t \\
       & =  \frac{1}{\alpha} f^{*}\left(\varPhi_{\mathrm{ub},\eta}\left(y\right)\right) - \frac{f^{*}\left(f^{\prime}(\delta)\right)}{\alpha}, 
    \end{aligned}
    $$
    which satisfies the characterization equation Eq. \eqref{ineq:PricingFunctionCharac} with $\beta = f^{*}\left(f^{\prime}(\delta)\right)/\alpha$. Moreover, we have 
    $$
    \begin{aligned}
        f^{*}\left(f^{\prime}(\delta)\right) 
        &= \max_{y\geq 0} f^{\prime}(\delta)y - f(y) \\
        &= f^{\prime}(\delta)\delta - f(\delta). 
    \end{aligned}
    $$
    Thus, we complete the proof. 
\end{proof}

\section{Proofs Related to Universal Design Space} \label{Apx:UniDSProofs}

\subsection{$\mathcal{F}$ is a Finite Set of $(\tau,\sigma)$-Elastic Functions}

\begin{proof}[Proof of Proposition \ref{Prop:f_hat_finite}]
Since $Ef_k''=yf_k'''/f_k''$ is well defined on $(0,+\infty)$ for each $k$, we have $f_k''>0$ there;
together with $f_k'(0)=0$ this gives $f_k'>0$ on $(0,+\infty)$. As a maximum of finitely many
continuous functions, $E^{*}$ is continuous on $(0,+\infty)$.

Let $E^{*}(y) = \max_{k\in[K]} Ef_k''(y)$. We claim that at every switching point $y_s\in\mathcal{Y}$,
\begin{equation}\label{eq:ElastMatch}
    Ef_{k_{s-1}}''(y_s) \;=\; Ef_{k_s}''(y_s) \;=\; E^{*}(y_s).
\end{equation}
Indeed, for $y\in(y_{s-1},y_s)$ we have $Ef_{k_{s-1}}''(y)=E^{*}(y)\geq Ef_{k_s}''(y)$; letting
$y\to y_s^-$ and using continuity gives $Ef_{k_{s-1}}''(y_s)\geq Ef_{k_s}''(y_s)$. The symmetric
argument on $(y_s,y_{s+1})$ gives the reverse inequality, proving \eqref{eq:ElastMatch}. 

\textbf{Construction.}
We define $\hat{f}$ recursively on the intervals $I_s=(y_s,y_{s+1}]$, $s\geq 0$, together with
initialization $\hat f=f_{k_0}$ on $I_0\cup\{0\}$. More precisely,

\emph{Base.} On $[0,y_1]$ set $C_0=1$ and $L\equiv 0$, so that $\hat{f}=f_{k_0}$.

\emph{Induction.} Suppose $\hat{f}$ has been defined on $[0,y_s]$ and is $C^2$ there with
$\hat{f}''(y_s)>0$. Set
\begin{equation}\label{eq:Crec}
    C_s = C_{s-1}\cdot\frac{f_{k_{s-1}}''(y_s)}{f_{k_s}''(y_s)}
         \;=\; \frac{\hat{f}''(y_s)}{f_{k_s}''(y_s)} \;>\;0 ,
\end{equation}
\begin{equation}\label{eq:Lrec}
    L(y) = a_s\,(y-y_s)+b_s \ \ \text{on } I_s, \qquad
    a_s = \hat{f}'(y_s)-C_sf_{k_s}'(y_s), \qquad
    b_s = \hat{f}(y_s)-C_sf_{k_s}(y_s),
\end{equation}
and $\hat{f}=C_sf_{k_s}+L$ on $I_s$. Both $C_s$ and $L$ are well defined because
$f_{k_s}''(y_s)>0$, and $C_s>0$ by induction. This defines $\hat{f}$ on $[0,+\infty)$, of the form \eqref{eq:f_hat_finite} with $C$ piecewise constant and $L$ piecewise linear on the intervals $\{I_s\}$.

\textbf{Regularity and boundary conditions.}
On each $I_s$ the function $\hat{f}$ is a positive multiple of $f_{k_s}$ plus an affine function, hence
smooth, with
\begin{equation}\label{eq:hatfpp}
    \hat{f}''=C_sf_{k_s}''>0 \quad\text{on } I_s .
\end{equation}
At $y_s$, the choices \eqref{eq:Crec}--\eqref{eq:Lrec} give
$$
\hat{f}(y_s^+)=\hat{f}(y_s),\qquad
\hat{f}'(y_s^+)=\hat{f}'(y_s),\qquad
\hat{f}''(y_s^+)=C_sf_{k_s}''(y_s)=\hat{f}''(y_s),
$$
the last equality by \eqref{eq:Crec}. Hence $\hat{f}\in C^{2}([0,+\infty))$ and $\hat{f}''>0$ on
$(0,+\infty)$. For the third derivative, using $g'''=g''\cdot Eg''/y$ and \eqref{eq:ElastMatch},
$$
\hat{f}'''(y_s^-)=C_{s-1}f_{k_{s-1}}'''(y_s)
=\hat{f}''(y_s)\frac{Ef_{k_{s-1}}''(y_s)}{y_s}
=\hat{f}''(y_s)\frac{E^{*}(y_s)}{y_s}
=\hat{f}'''(y_s^+),
$$
so $\hat{f}\in C^{3}([0,+\infty))$. Finally $\hat{f}=f_{k_0}$ on $[0,y_1]$ yields
$\hat{f}(0)=\hat{f}'(0)=0$, and $\hat{f}''>0$ together with $\hat{f}'(0)=0$ yields $\hat{f}'>0$ on
$(0,+\infty)$. 

\textbf{Elasticity.}
By \eqref{eq:hatfpp}, on the interior of each $I_s$,
$$
E\hat{f}''(y)=\frac{y\,\hat{f}'''(y)}{\hat{f}''(y)}
=\frac{y\,C_sf_{k_s}'''(y)}{C_sf_{k_s}''(y)}
=Ef_{k_s}''(y)=E^{*}(y),
$$
the multiplicative constant cancelling. At $y=y_s$ the regularity and boundary conditions give $E\hat{f}''(y_s)=E^{*}(y_s)$. Hence $E\hat{f}''=E^{*}$ on $(0,+\infty)$. Since$\tau-2\leq Ef_k''\leq\sigma-2$ for every $k$, the same bounds hold for $E^{*}$ and therefore for $E\hat{f}''$. Consequently, $\hat{f}$ is strongly $(\tau,\sigma)$-elastic, which proves (i).

\textbf{Safe envelope.}
Fix $f\in\mathcal{F}$ and set $h=\hat{f}''/f''$ on $(0,+\infty)$. By (i) and Step 0, $h$ is positive
and, since $\hat{f}\in C^{3}$ and $f$ is smooth, $h\in C^{1}((0,+\infty))$ with
$$
\bigl(\ln h\bigr)'(y)
=\frac{\hat{f}'''(y)}{\hat{f}''(y)}-\frac{f'''(y)}{f''(y)}
=\frac{E\hat{f}''(y)-Ef''(y)}{y}
=\frac{E^{*}(y)-Ef''(y)}{y}\ \geq\ 0 ,
$$
the last inequality by definition of $E^{*}$. Hence $\ln h$, and therefore $h$, is non-decreasing on
$(0,+\infty)$.

Now fix $\phi\geq y>0$ and put $v=y/\phi\in(0,1]$. For every $u\in[v,1]$ we have $u\phi\leq\phi$, so
$h(u\phi)\leq h(\phi)$, i.e.
$$
\frac{\hat{f}''(u\phi)}{\hat{f}''(\phi)}\ \leq\ \frac{f''(u\phi)}{f''(\phi)} .
$$
Integrating over $u\in[v,1]$ and using
$\int_v^1 g''(u\phi)\Id u=\bigl(g'(\phi)-g'(v\phi)\bigr)/\phi$ for $g\in\{\hat{f},f\}$ gives
$$
\frac{\hat{f}'(\phi)-\hat{f}'(y)}{\phi\,\hat{f}''(\phi)}
\ \leq\
\frac{f'(\phi)-f'(y)}{\phi\,f''(\phi)},
$$
which is exactly $F_{\hat{f}}(\phi,y)\leq F_f(\phi,y)$. This proves (ii).

\textbf{Uniqueness.}
Let $\tilde{C}$ be piecewise constant and $\tilde{L}$ piecewise linear with switching points in
$\mathcal{Y}$, and suppose $\tilde{f}:=\tilde{C}f_{k^{*}}+\tilde{L}$ is twice continuously
differentiable. On $I_s$ we have $\tilde{f}''=\tilde{C}_sf_{k_s}''$, so continuity of $\tilde{f}''$ at
$y_s$ forces $\tilde{C}_sf_{k_s}''(y_s)=\tilde{C}_{s-1}f_{k_{s-1}}''(y_s)$, i.e. the recursion
\eqref{eq:Crec}; continuity of $\tilde{f}'$ and of $\tilde{f}$ at $y_s$ then force $a_s$ and $b_s$ as
in \eqref{eq:Lrec}. Hence $(\tilde{C},\tilde{L})$ is determined by its values on $[0,y_1]$. If
moreover $\tilde{f}(0)=\tilde{f}'(0)=0$, then writing $\tilde{L}(y)=cy+d$ on $[0,y_1]$ and using
$f_{k_0}(0)=f_{k_0}'(0)=0$ gives $d=0$ and $c=0$; thus $\tilde{f}=\lambda\hat{f}$ with
$\lambda=\tilde{C}_0>0$, and $\tilde{f}=\hat{f}$ under the normalisation $\tilde{C}_0=1$. 
\end{proof}

\begin{proof}[Proof of Theorem~\ref{Thm:UniDSElastic}]
    By the construction of $\hat{f}$ in Proposition \ref{Prop:f_hat_finite}, we have for any $f\in \mathcal{F}$, 
    $$
    \frac{\hat{f}'(\phi) - \hat{f}'(y)}{\phi \hat{f}''(\phi)} \leq \frac{f'(\phi) - f'(y)}{\phi f''(\phi)}.
    $$
    Therefore, for any $\hat{\phi}$ satisfying the sufficient condition of Eq.\eqref{eq:SufNec_phi} corresponding to $\hat{f}$, we have 
    \begin{align*}
        \hat{\phi}'(y) = \asigmas\cdot \frac{\hat{f}'(\hat{\phi}) - \hat{f}'(y)}{\hat{\phi} \hat{f}''(\hat{\phi})} \leq \asigmas\cdot \frac{f'(\hat{\phi}) - f'(y)}{\hat{\phi} f''(\hat{\phi})}.
    \end{align*}
    
    Additionally, the lower-bound solution to Eq.\eqref{eq:SufNec_phi} corresponding to $\hat{f}$, denoted by $\hat{\phi}_{\mathrm{lb}}$, must be lower bounded by the lower-bound solution to Eq.\eqref{eq:SufNec_phi} corresponding to $f$, denoted by $\phi_{\mathrm{lb}}$. Otherwise, by the ODE comparison theorem (see Comparison Theorem  \cite{Walter1998}, p.~90) and Lemma~\ref{Lem:LowerBoundforFeasibleSol}, we have $\hat{\phi}_{\mathrm{lb}}(y) \leq y$ for sufficiently large $y>0$, which contradicts that $\hat{\phi}_{\mathrm{lb}}$ is a feasible solution to Eq.\eqref{eq:SufNec_phi}. Meanwhile, the upper-bound solution to Eq.\eqref{eq:SufNec_phi} corresponding to $\hat{f}$, denoted by $\hat{\phi}_{\mathrm{ub}}$, must be upper bounded by the upper-bound solution to Eq.\eqref{eq:SufNec_phi} corresponding to $f$, denoted by $\phi_{\mathrm{ub}}$. This is because, again by the ODE comparison theorem and the definition of $\phi_{\mathrm{ub}}$ (see Definition~\ref{Def:UpperLowerBoundSol}), $\hat{\phi}_{\mathrm{ub}}$ is upper bounded by any solution to Eq.\eqref{eq:mainODE} corresponding to $f$ that upper bounds $\phi_{\mathrm{ub}}$, and thus $\hat{\phi}_{\mathrm{ub}} \leq \phi_{\mathrm{ub}}$. Finally, by Proposition~\ref{Prop:SolSpace}, we complete the proof.
    
\end{proof}

\subsection{$\mathcal{F}$ is an Infinite Set of $(\tau,\sigma)$-Polynomial Functions}

\subsubsection{Preliminaries} \label{Apx:Non-decreasing Elasticity}

We first show that $(\tau,\sigma)$-polynomials have non-decreasing elasticity; a direct computation of their elasticity bounds then establishes the strong $(\tau,\sigma)$-elasticity.

\begin{proposition} \label{Prop:ts_poly_nondecreasing}
    The $(\tau,\sigma)$-polynomials defined in Definition \ref{def:tau_sigma_convex} have non-decreasing elasticity. 
\end{proposition}
\begin{proof}
We are given $f(y) = \sum_{k=\tau}^\sigma c_k y^k$ where $c_k \geq 0$ for all $k$. We want to show that the elasticity function $Ef'(y) = \frac{y f'(y)}{f(y)}$ is non-decreasing for $y > 0$, meaning its derivative 
$$
\frac{\mathrm{d} Ef'}{\mathrm{d}y}(y) = \frac{(f'(y) + y f''(y))f(y) - y (f'(y))^2}{f(y)^2} \geq 0. 
$$
By multiplying the numerator by $y$, we only need to prove that
\begin{equation} \label{eq:y_times_num}
    \left(y f'(y) + y^2 f''(y)\right)f(y) - \left(y f'(y)\right)^2 \geq 0.
\end{equation} 
by substitute the polynomial expansions: $f(y) = \sum_{k=\tau}^{\sigma} c_k y^k$, $y f'(y) = \sum_{k=\tau}^{\sigma} k c_k y^k$, and $y^2 f''(y) = \sum_{k=\tau}^{\sigma} k(k-1) c_k y^k$, Eq. \eqref{eq:y_times_num} can be written as
$$
\left(\sum_{k=\tau}^{\sigma} k^2 c_k y^k\right)\left(\sum_{k=\tau}^{\sigma} a_k y^k\right) - \left(\sum_{k=\tau}^{\sigma} k c_k y^k\right)^2 \geq 0,
$$
which holds by the Cauchy-Schwarz Inequality (setting $u_k = k\sqrt{c_k y^k}$ and $v_k = \sqrt{c_k y^k}$). Therefore, $Ef'(y)$ is non-decreasing for $y > 0$.
\end{proof}

For the rest of this section we write $f(y)=\sum_{k=\tau}^{\sigma}c_ky^{k}$ with
$c_k\ge0$ and $c_\tau,c_\sigma>0$, so that the $c_k$ are exactly the coefficients whose
supports $[\ubar{c}_k,\bar{c}_k]$ appear in Proposition~\ref{Prop:tausigmaPolyUniDS}
and Theorem~\ref{Thm:tausigmaPolyUniDS}. For such an $f$ and any $\phi>y\ge0$,
\begin{equation}\label{eq:F_poly}
F(\phi,y)=\frac{f'(\phi)-f'(y)}{\phi f''(\phi)}
=\frac{\sum_{k=\tau}^{\sigma}kc_k\bigl(\phi^{k-1}-y^{k-1}\bigr)}
       {\sum_{k=\tau}^{\sigma}k(k-1)c_k\phi^{k-1}} .
\end{equation}
We abbreviate, for integers $k\ge1$ and $\tau\leq d\leq\sigma$,
\begin{equation}\label{eq:h_Delta}
h_k(\phi,y)=\frac{\phi^{k}-y^{k}}{k\phi^{k}},
\quad \text{ and } \quad
\Delta_d(f;\phi,y)=F(\phi,y)-h_{d-1}(\phi,y).
\end{equation}
Because several are compared below, we carry the cost function as an argument, and write $F_g$ for the function \eqref{eq:F_poly} associated with a cost function $g$. 

\subsubsection{Technical Lemmas}

\begin{lemma} \label{Lem:TermMono}
Let $\phi>y\ge0$ and let $h_k$ be as in Eq.~\eqref{eq:h_Delta}.
\begin{enumerate}
\item[(i)] $h_k(\phi,y)$ is strictly decreasing with respect to $k$ for $k\ge1$.
\item[(ii)] For any $(\tau,\sigma)$-polynomial $g(y)=\sum_{k=\tau}^{\sigma}c_ky^{k}$ and any integer $m\ge1$, let $\Delta_{m+1}(g;\phi,y) = F_g(\phi,y)-h_m(\phi,y)$. Then we have
\begin{equation}\label{eq:ConvexRep}
\Delta_{m+1}(g;\phi,y) 
=\frac{\sum_{k=\tau}^{\sigma}k(k-1)c_k\phi^{k-1}
        \bigl(h_{k-1}(\phi,y)-h_m(\phi,y)\bigr)}
      {\sum_{k=\tau}^{\sigma}k(k-1)c_k\phi^{k-1}} .
\end{equation}
In particular, if $\tau<\sigma$ and $c_\tau>0$, then
$$
\Delta_\tau(g;\phi,y)\le0
\qquad\text{and}\qquad
\Delta_\sigma(g;\phi,y)>0 .
$$
\end{enumerate}
\end{lemma}

\begin{proof}
(i) Writing $t=y/\phi\in[0,1)$ we have
$$
h_k(\phi,y)=\frac{1-t^{k}}{k}=\int_{t}^{1}u^{k-1}\,\mathrm{d}u .
$$
For every $u\in(t,1)$, $u^{k-1}$ is strictly decreasing in
$k$; hence
$h_k(\phi,y)$ is strictly decreasing in $k$.

(ii) Since $kc_k(\phi^{k-1}-y^{k-1})=k(k-1)c_k\phi^{k-1}h_{k-1}(\phi,y)$, the numerator
of Eq.~\eqref{eq:F_poly} equals $\sum_{k}k(k-1)c_k\phi^{k-1}h_{k-1}(\phi,y)$. Thus
$F_g(\phi,y)$ is the weighted average of $h_{\tau-1},\dots,h_{\sigma-1}$ with the
nonnegative weights $k(k-1)c_k\phi^{k-1}$, whose sum is positive because
$c_\sigma>0$; subtracting $h_m$ from both sides gives Eq.~\eqref{eq:ConvexRep}.

When $m=\tau-1$, by the monotonicity $h_{k-1}-h_{\tau-1}$ in
Eq.~\eqref{eq:ConvexRep} is non-positive, so $\Delta_\tau(g)\le0$. When $m=\sigma-1$, again by the monotonicity $h_{k-1}-h_{\sigma-1}$ is non-negative, and the one for $k=\tau$ is strictly positive, so
$\Delta_\sigma(g)>0$.
\end{proof}

\begin{lemma} \label{Lem:OneCoeff}
Fix $\phi>y\ge0$ and an integer $d$ with $\tau\leq d\leq\sigma$. Let $g$ and $\hat g$ be $(\tau,\sigma)$-polynomials differing only in the coefficient of $y^{d}$, say $\hat c_d=c_d-\delta$. Then
\begin{enumerate}
\item[(i)] $F_{\hat g}(\phi,y)\leq F_{g}(\phi,y)
\iff \delta\cdot\Delta_d(g;\phi,y)\le0$;
\item[(ii)] $\Delta_d(\hat g;\phi,y)$ and $\Delta_d(g;\phi,y)$ have the same sign, and
vanish together.
\end{enumerate}
\end{lemma}

\begin{proof}
Put $A=\sum_{k}kc_k(\phi^{k-1}-y^{k-1})$, $B=\sum_{k}k(k-1)c_k\phi^{k-1}$, so
$F_g=A/B$ with $B>0$, and $a=d(\phi^{d-1}-y^{d-1})\ge0$, $b=d(d-1)\phi^{d-1}>0$, so
that $a/b=h_{d-1}(\phi,y)$. Then $F_{\hat g}=(A-\delta a)/(B-\delta b)$, and
$B-\delta b>0$ because $\hat g$ also has nonnegative coefficients with
$\hat c_\sigma>0$.
Consequently, $F_{\hat g}\leq F_g$ is equivalent to $\delta\,(bA-aB)\le0$, and dividing by $bB>0$ to
$\delta\bigl(A/B-a/b\bigr)=\delta\cdot\Delta_d(g;\phi,y)\le0$.

Recall the mediant inequality: for $q,s>0$, the fraction $\frac{p+r}{q+s}$ lies in
the closed interval with endpoints $\frac pq$ and $\frac rs$, and equals an endpoint
only if $\frac pq=\frac rs$. If $\delta\le0$, then
$F_{\hat g}=\frac{A+|\delta|a}{B+|\delta|b}$ is the mediant of $F_g$ and $h_{d-1}$,
hence lies between them, so $F_{\hat g}-h_{d-1}$ has the same sign as $F_g-h_{d-1}$.
If $\delta>0$, then
$F_{g}=\frac{(A-\delta a)+\delta a}{(B-\delta b)+\delta b}$ is the mediant of
$F_{\hat g}$ and $h_{d-1}$, so $F_g$ lies between $F_{\hat g}$ and $h_{d-1}$; if
$F_g>h_{d-1}$ this forces $F_{\hat g}\geq F_g>h_{d-1}$, if $F_g<h_{d-1}$ it forces
$F_{\hat g}\leq F_g<h_{d-1}$, and if $F_g=h_{d-1}$ the equality case of the mediant
inequality forces $F_{\hat g}=h_{d-1}$.
\end{proof}

\subsubsection{Main Proofs}

\begin{proof}[Proof of Proposition~\ref{Prop:tausigmaPolyUniDS}]
As a special case, if $\tau=\sigma$ the cost function reduces to a power cost function
and $F(\phi,y)$ is independent of the coefficient, so any estimate of $c_\tau$ in
$[\ubar{c}_\tau,\bar{c}_\tau]$ is safe. We therefore assume $\tau<\sigma$. Throughout,
``safe'' is used in the pointwise sense $\hat F(\phi,y)\leq F(\phi,y)$ for all
$\phi>y\ge0$.

Fix an integer $d$ with $\tau\leq d\leq\sigma$, let $c_\tau,\dots,c_\sigma$ be the
realized coefficients, let $\hat c_d=c_d-\delta$ be the estimate of $c_d$ with the
remaining coefficients estimated exactly, and let $\hat f$ be the resulting cost
function; thus $c_d-\bar{c}_d\leq\delta\leq c_d-\ubar{c}_d$. By
Lemma~\ref{Lem:OneCoeff}(i), $\hat c_d$ is safe if and only if
\begin{equation} \label{Ineq:EquivIneq}
\delta\cdot\Delta_d(f;\phi,y)\le0,\qquad\forall\,\phi>y\ge0 .
\end{equation}

\emph{Case $d=\tau$.} By Lemma~\ref{Lem:TermMono}(ii), $\Delta_\tau(f;\phi,y)\le0$ for
all $\phi>y\ge0$, so Eq.~\eqref{Ineq:EquivIneq} holds whenever $\delta\ge0$, i.e.\
whenever $\hat c_\tau\leq c_\tau$. Since $c_\tau$ is known only to lie in
$[\ubar{c}_\tau,\bar{c}_\tau]$, the estimate $\hat c_\tau=\ubar{c}_\tau$ is safe.

\emph{Case $d=\sigma$.} By Lemma~\ref{Lem:TermMono}(ii), $\Delta_\sigma(f;\phi,y)>0$
for all $\phi>y\ge0$, so Eq.~\eqref{Ineq:EquivIneq} holds if and only if $\delta\le0$,
i.e.\ $\hat c_\sigma\geq c_\sigma$, and the estimate $\hat c_\sigma=\bar{c}_\sigma$ is
safe.

Moreover, the two estimates are jointly safe. That is, let $g$ be obtained from $f$ by replacing $c_\tau$ with $\ubar{c}_\tau$, and $g'$ from $g$ by replacing $c_\sigma$ with $\bar{c}_\sigma$. Combining the previous two cases shows $F_{g'}\leq F_f$ pointwise.

\emph{Case $\tau<d<\sigma$.} Here $\Delta_d(f;\cdot,\cdot)$ takes both signs. Fix
$t\in(0,1)$ and put $y=t\phi$; then $h_k(\phi,t\phi)=(1-t^{k})/k$ does not depend on
$\phi$, whereas in Eq.~\eqref{eq:ConvexRep} the weights $k(k-1)c_k\phi^{k-1}$
concentrate on $k=\tau$ as $\phi\to0^{+}$ and on $k=\sigma$ as $\phi\to+\infty$.
Therefore, by Lemma~\ref{Lem:TermMono}(i),
$$
\lim_{\phi\to0^{+}}\Delta_d(f;\phi,t\phi)=h_{\tau-1}-h_{d-1}>0,
\qquad
\lim_{\phi\to+\infty}\Delta_d(f;\phi,t\phi)=h_{\sigma-1}-h_{d-1}<0.
$$
Consequently Eq.~\eqref{Ineq:EquivIneq} forces $\delta=0$, i.e.\ $\hat c_d=c_d$. If the support $[\ubar{c}_d,\bar{c}_d]$ is non-degenerate this cannot be guaranteed and no safe estimate of $c_d$ exists; if it is degenerate then $\hat c_d=c_d$ is forced and safe.

Combining the cases completes the proof.
\end{proof}

\begin{proof}[Proof of Theorem~\ref{Thm:tausigmaPolyUniDS}]
If $\tau=\sigma$ the family $\mathcal{F}$ consists of positive multiples of $y^{\sigma}$,
for which $F$ is the same function, and the claim is trivial; assume $\tau<\sigma$.
Since $\mathcal{G}$ is a finite family of $(\tau,\sigma)$-polynomials, hence strongly $(\tau,\sigma)$-elastic costs, following the same analysis as in the proof of Theorem~\ref{Thm:UniDSElastic}, it suffices to prove that for every $f\in\mathcal{F}$
\begin{equation}\label{Ineq:GoalIneq}
\min_{\tau\leq d\leq\sigma-1}F_{f^{(d)}}(\phi,y)\;\leq\;F_{f}(\phi,y),
\qquad\forall\,\phi\geq y\ge0 .
\end{equation}
Granting Eq.~\eqref{Ineq:GoalIneq}, Proposition~\ref{Prop:f_hat_finite} applied
to $\mathcal{G}$ shows that $\hat f$ is well defined and strongly $(\tau,\sigma)$-elastic with
$E\hat f''\geq Ef^{(d)\prime\prime}$ for each $d$. Thus,
Lemma~\ref{Lem:ElasticComp} gives
$\hat F\leq\min_{d}F_{f^{(d)}}\leq F_f$ for every $f\in\mathcal{F}$. Then Theorem~\ref{Thm:UniDSElastic} yields $\mathcal{S}(\alpha,\hat f)\subseteq\mathcal{S}(\alpha,f)$ for all $f\in\mathcal{F}$.

It remains to prove Eq.~\eqref{Ineq:GoalIneq}. For $\phi=y$ both sides vanish, so fix
$\phi>y\ge0$; all quantities below are evaluated at this $(\phi,y)$ and we suppress it
from the notation. Let $c=(c_\tau,\dots,c_\sigma)$ be the coefficient vector of $f$.
We construct a chain of cost functions, each obtained from the previous one by
modifying a single coefficient, along which $F$ does not increase, and its terminus nevertheless always lies in $\mathcal{G}$.

\textbf{The chain.} Set $g_{\tau-1}=f$ and, for $d=\tau,\tau+1,\dots,\sigma$,
let $g_{d}$ be obtained from $g_{d-1}$ by replacing the coefficient of $y^{d}$ by
$$
\hat c_d=
\begin{cases}
\ubar{c}_d, & \text{if }\Delta_d(g_{d-1})\le0
 \qquad(\text{a \emph{low} move}),\\[3pt]
\bar{c}_d,  & \text{if }\Delta_d(g_{d-1})>0
 \qquad(\text{a \emph{high} move}).
\end{cases}
$$
Every $g_d$ is a $(\tau,\sigma)$-polynomial, and $g_d$ differs
from $g_{d-1}$ in exactly one coefficient; so Lemmas~\ref{Lem:TermMono}(ii) and
\ref{Lem:OneCoeff} apply at every step. Note also that in $g_{d-1}$ the coefficients
$c_d,\dots,c_\sigma$ are still those of $f$, so that
$\delta:=c_d-\hat c_d$ satisfies $\delta=c_d-\ubar{c}_d\ge0$ for a low move and
$\delta=c_d-\bar{c}_d\le0$ for a high move.

In either case $\delta\cdot\Delta_d(g_{d-1})\le0$, so Lemma~\ref{Lem:OneCoeff}(i) gives $F_{g_d}\leq F_{g_{d-1}}$. Chaining over $d=\tau,\dots,\sigma$,
\begin{equation}\label{eq:chain_desc}
F_{g_\sigma}\;\leq\;F_{g_{\sigma-1}}\;\leq\;\cdots\;\leq\;F_{g_{\tau-1}}=F_{f}.
\end{equation}
In particular, by Lemma~\ref{Lem:TermMono}(ii), the move at $d=\tau$ is low, and the move at $d=\sigma$ is high. 

Suppose the move at some index $d$ with $\tau\leq d\leq\sigma-1$ is high, i.e.\ $\Delta_d(g_{d-1})>0$. Since $g_d$
differs from $g_{d-1}$ only in the coefficient of $y^{d}$,
Lemma~\ref{Lem:OneCoeff}(ii) gives $\Delta_d(g_{d})>0$, that is,
$F_{g_d}>h_{d-1}$. By Lemma~\ref{Lem:TermMono}(i), $h_{d-1}>h_{d}$, we have
$$
\Delta_{d+1}(g_{d})=F_{g_d}-h_{d}>0 ,
$$
so the move at $d+1$ is high as well. By induction on $d$, once a high move occurs all
subsequent moves are high.

Consequently, there is an index $d^{\ast}$ with
$\tau\leq d^{\ast}\leq\sigma-1$ such that the moves at $\tau,\dots,d^{\ast}$ are low and
the moves at $d^{\ast}+1,\dots,\sigma$ are high. Hence the coefficients of $g_\sigma$
are $\ubar{c}_k$ for $k\leq d^{\ast}$ and $\bar{c}_k$ for $k>d^{\ast}$, i.e.\
$g_\sigma=f^{(d^{\ast})}$. Combining with Eq.~\eqref{eq:chain_desc},
$$
\min_{\tau\leq d\leq\sigma-1}F_{f^{(d)}}
\;\leq\;F_{f^{(d^{\ast})}}
\;=\;F_{g_\sigma}
\;\leq\;F_{f},
$$
which is Eq.~\eqref{Ineq:GoalIneq}. This completes the proof.
\end{proof}

\section{Proof of Theorem \ref{Thm:AverageCRforMK}} \label{Apx:AverageCRforMK}

We assume w.l.o.g. that $\sigma_j$ is increasing for $j\in[m]$.
By Algorithm~\ref{alg:PRM_MKwC}, we have 
\begin{equation} \label{Eq:EqualLevel}
    f_j'\left(\phi_j(y_j)\right) = f_{j+1}'\left(\phi_{j+1}(y_{j+1})\right).
\end{equation}
For convenience, we simplify $\phi_j(y_j)$ by $\phi_j$ for each $j\in [m]$.
By Eq.~\eqref{Eq:EqualLevel}, we have, for each $j$,
\begin{align*}
    & F_{j}\left(\phi_{j},y_j\right) \geq F_{j+1}\left(\phi_{j+1},y_{j+1}\right) \\
    \quad\Leftrightarrow & \frac{Ef_{j+1}'\left(\phi_{j+1}\right)}{Ef_j'\left(\phi_{j}\right)} \geq \frac{f_{j+1}'\left(\phi_{j+1}\right) - f_{j+1}'\left(y_{j+1}\right)}{f_j'\left(\phi_j\right) - f_j'\left(y_j\right)} \\
    \quad\Leftrightarrow & \frac{Ef_{j+1}'\left(\phi_{j+1}\right)}{Ef_j'\left(\phi_{j}\right)} \geq \frac{1 - f_{j+1}'\left(y_{j+1}\right)/f_{j+1}'\left(\phi_{j+1}\right)}{1 - f_j'\left(y_j\right)/f_j'\left(\phi_j\right)}.
\end{align*}
For each $j\in [m]$, since $f_j$ is strongly $(\tau,\sigma)$-elastic, we have (see Appendix~\ref{Apx:Right Asymptotic Analysis})
$$
\lim_{y_j\rightarrow\infty} \frac{f_j'(y_j)}{f_{j}'\left(\phi_j\right)} = \frac{1}{\sigma_j}.
$$
Therefore, to prove $F_{j}\left(\phi_{j},y_j\right) \geq F_{j+1}\left(\phi_{j+1},y_{j+1}\right)$, we only need to show that for sufficiently large $y>0$
$$
\frac{Ef_{j+1}'\left(\phi_{j+1}\right)}{Ef_j'\left(\phi_{j}\right)} \geq \frac{\sigma_j}{\sigma_{j+1}} \cdot \frac{\sigma_{j+1}-1}{\sigma_{j}-1}.
$$
On the other hand, recall that $Ef_j'(y)$ is bounded by $[\tau-1,\sigma-1]$ and is non-decreasing. So we have
\begin{align*}
    \lim_{y\rightarrow\infty}Ef_j'(y) = \sigma_{j}-1 \quad \forall j\in [m].
\end{align*}
Since $f_j'$ and $\phi_j$ are increasing functions, by Eq.~\eqref{Eq:EqualLevel}, one can define $y_{j}$ as a function of $y_{j+1}$ (which is well-defined almost everywhere). This implies 
\begin{align*}
    \lim_{y_{j+1}\rightarrow\infty}Ef_j'\left(\phi_j(y_j)\right) = \sigma_{j}-1 \quad \forall j\in [m-1].
\end{align*}
Therefore, we have
\begin{align*}
    \lim_{y_{j+1}\rightarrow\infty} \frac{Ef_{j+1}'\left(\phi_{j+1}\right)}{Ef_j'\left(\phi_{j}\right)} 
    &= \frac{\sigma_{j+1}-1}{\sigma_j-1} \\
    &\geq \frac{\sigma_j}{\sigma_{j+1}} \cdot \frac{\sigma_{j+1}-1}{\sigma_{j}-1},
\end{align*}
which leads to $F_{j}\left(\phi_{j},y_j\right) \geq F_{j+1}\left(\phi_{j+1},y_{j+1}\right)$.
That is, $F_j(\phi_j,y_j)$ is decreasing with respect to $j$.
Meanwhile, competitive ratio $\alpha_{\star}^{(\sigma_j)}$ is increasing as $\sigma_j$ is increasing for $j\in [m]$.
Recall Chebyshev's sum inequality \cite{Hardy1991}, 
which states that if $a_1 \leq a_2 \leq \cdots \leq a_n$ and $b_1 \geq b_2 \geq \cdots \geq b_n$, then  
\begin{align*}
\frac{1}{n} \sum_{k=1}^n a_k b_k \leq \left(\frac{1}{n} \sum_{k=1}^n a_k\right) \left(\frac{1}{n} \sum_{k=1}^n b_k\right).
\end{align*}
We then have 
$$
\begin{aligned}
    \sum_{j=1}^{m}\phi_j'(y_j) 
    &= \sum_{j=1}^{m} \alpha_{\star}^{(\sigma_j)}\cdot F_j(\phi_j,y_j) \\
    &\leq \frac{1}{m} \sum_{j\in [m]} \alpha_{\star}^{(\sigma_j)}\cdot \sum_{j=1}^{m}F_j(\phi_j,y_j), 
\end{aligned}
$$
where $y_j\geq 0$ and $j\in [m]$. 
Thus, by Theorem \ref{sufficiency_unlimited_supply_fullycon}, we complete the proof. 

\section{Proofs Related to \OACC with (Hard) Supply Constraint}

Denote Eq. \eqref{eq:eq:SufNec_phi_limited_ODE} as $\phi^{\prime}(y) = \alpha\cdot F_{\textsf{LS}}(\phi,y)$ and $ F_{\textsf{LS}} $ is given by  
$$
F_{\textsf{LS}}(\phi,y) =  \frac{f'(\phi(y))  - f'(y) }{\min\{\phi(y),1\}\cdot f''(\phi(y))},
$$
where the subscript ``\textsf{LS}" stands for ``Limited Supply".

\subsection{Solving Eq. \eqref{eq:SufNec_phi_limited} by Extending Solutions to Eq. \eqref{eq:SufNec_phi}} \label{Apx:ExtendSol}

Let $\alpha = \asigmas$ and we want to find for what $\rho$ Eq. \eqref{eq:SufNec_phi_limited} has at least one feasible solution. For convenience, in this subsection, we relax the two-point boundary condition Eq. \eqref{eq:SufNec_phi_limited_BoundaryCond} into an initial condition $\phi(0)=0$, which recovers the constraints in Eq. \eqref{eq:SufNec_phi}. By Theorem \ref{Thm:UpperBoundAch}, there are also infinitely many solutions to Eq. \eqref{eq:SufNec_phi_limited}, which can be constructed as follows.
Let $\varphi$ be a solution  to Eq. \eqref{eq:SufNec_phi} and $\varphi(y^*) = 1$. Then $\varphi$ can be extended to a solution of Eq. \eqref{eq:SufNec_phi_limited}, denoted as $\tilde{\phi}$. More precisely, let $\phi$ be the solution to the following ODE
$$
\left\{
\begin{aligned}
    &\phi^{\prime}(y) = \alpha \cdot F_{\textsf{LS}}(\phi,y), \quad \forall y\geq y^*, \\
    &\phi(y^*) = 1,
\end{aligned}
\right.
$$
and the extended solution $\tilde{\phi}$ is defined as 
\begin{equation} \label{eq:ExtendedSol}
    \tilde{\phi}(y) =\left\{
    \begin{aligned}
        &\varphi(y), \quad && 0\leq y\leq y^*, \\
        &\phi(y), \quad && y> y^*.
    \end{aligned}
    \right.
\end{equation}
As the supply is limited by 1, $\tilde{\phi}$ is $\alpha$-competitive on the interval $[0,\tilde{\phi}^{-1}(1)]$, and therefore $\tilde{\phi}$ is $\asigmas$-competitive only if the utilization level $\rho\leq \tilde{\phi}(1)$. We say $\tilde{\phi}(1)$ is the largest utilization level for $\tilde{\phi}$, i.e., for any utilization level $\rho > \tilde{\phi}(1)$, the reserve function cannot be $\alpha$-competitive. 

Furthermore, let $\varphi_1$ and $\varphi_2$ be two solutions to Eq.~\eqref{eq:SufNec_phi} with $\alpha = \asigmas$, such that $\varphi_1(y) > \varphi_2(y)$ for all $y > 0$. Denote their corresponding extended solutions as $\tilde{\phi}_1$ and $\tilde{\phi}_2$. Clearly, we have $\tilde{\phi}_1(y) > \tilde{\phi}_2(y)$ for all $y > 0$. Consequently, the largest utilization level associated with $\tilde{\phi}_1$ exceeds that of $\tilde{\phi}_2$, i.e., $\tilde{\phi}_2(1) < \tilde{\phi}_1(1)$.  This observation implies that the upper-bound solution $\phi_{\mathrm{ub}}$ to Eq.~\eqref{eq:SufNec_phi} determines the maximal utilization level achievable by any feasible solution to Eq.~\eqref{eq:SufNec_phi_limited}. We define this \textit{maximal utilization level} for a given $\alpha = \asigmas$ as $P(\alpha)$, formally stated in Definition \ref{Def:UpperBoundInput}. 

\medskip
\noindent \textbf{Illustration of extending solutions.} By our results in Section~\ref{sec_characterizing_solution_space}, for \OACC with $f(y) = y^3 + y^2$ and unlimited supply, the optimal competitive ratio is $\asigmas = 3\sqrt{3}$. Figure~\ref{fig:23PolynomialCost} presents the solutions to Eq.~\eqref{eq:SufNec_phi}, where the red curve represents the upper-bound solution and the blue curves correspond to solutions obtained via Proposition~\ref{Prop:LowerSol}.  We also present the solutions to Eq.~\eqref{eq:SufNec_phi_limited} in Figure~\ref{fig:23PolynomialCost_Limited2} for the case with a competitive ratio of $\asigmas$ and $\rho = 1.2$ (i.e., $v_{\rm max} = 4.32B^2 + 2.4B$). In this figure, the gray lines correspond to solutions for \OACC without supply constraints, while the colored curves represent solutions under supply constraints.  

From Figure~\ref{fig:23PolynomialCost} to Figure~\ref{fig:23PolynomialCost_Limited2}, one can clearly observe how a solution is extended. In general, the transition point occurs when a solution reaches $\phi(y)=1$, beyond which it grows more rapidly compared to its original trajectory.  Another key observation is that moving from the unlimited-supply \OACC setting to the supply-constrained setting broadens the optimal design space. For instance, the lowest feasible solution in Figure~\ref{fig:23PolynomialCost_Limited2} would be infeasible in the original \OACC setting due to a violation of the superlinearity condition. However, because of the change in the differential equation when the solution hits $\phi(y)=1$, this trajectory accelerates and ultimately overlaps with $\phi(y)=y$ at $y=1$, thereby becoming \textit{locally} superlinear on $ y \in [0,1]$.  Finally, we note that our analysis only concerns solutions over the local interval $[0,1]$; the extended ranges shown in the figures are provided solely for illustrative purposes. 

\begin{figure*}
  \centering
  \begin{subfigure}{0.45\textwidth}
    \centering
    \includegraphics[width=0.9\linewidth]{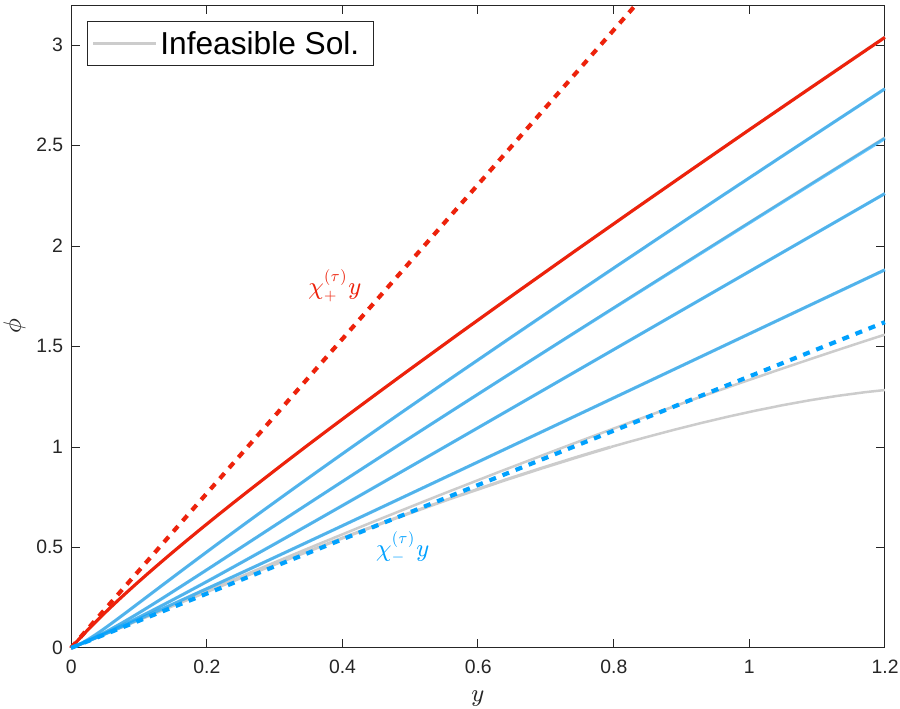}
    \caption{Unlimited Supply Case} \label{fig:23PolynomialCost}
  \end{subfigure}
  \begin{subfigure}{0.45\textwidth}
    \centering
    \includegraphics[width=0.9\linewidth]{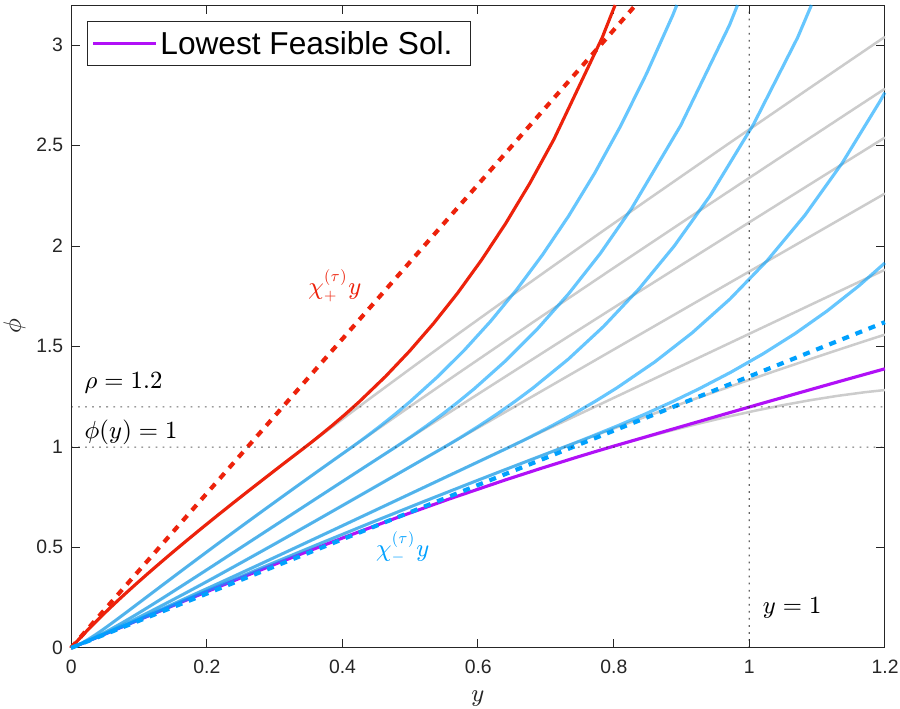}
    \caption{Limited Supply Case} \label{fig:23PolynomialCost_Limited2}
  \end{subfigure}
  \caption{Feasible reserve functions for \OACC with cost function $f(y) = y^3+y^2$ under two different settings: (a) without supply constraints (unlimited supply), and (b) with supply constraints ( limited supply).}
\end{figure*}

\subsection{Proof of Theorem \ref{Thm:ExistenceSol_Limited}} \label{Apx:ExistenceSol_Limited}

    For any $\alpha \geq \ataus$, denote the upper-bound solution to Eq. \eqref{eq:SufNec_phi} as $\phi_{\mathrm{ub}}(y)$ and the corresponding upper-bound solution to Eq. \eqref{eq:SufNec_phi_limited} as $\tilde{\phi}_{\mathrm{ub}}(y)$. For simplicity, here we drop the $\alpha$ in the notations of the upper-bound solutions. 
    By Definition \ref{Def:UpperBoundInput}, we have 
    $$
    P(\alpha) = \sup_{y\geq 0} \tilde{\phi}_{\mathrm{ub}}(y).
    $$

    If $\alpha$ is so small that $\phi_{\mathrm{ub}}(y)\leq 1$ for $y\geq 0$, then we have $\phi_{\mathrm{ub}} = \tilde{\phi}_{\mathrm{ub}}$, and thus, 
    $$
    P(\alpha) = \sup_{y\geq 0} \phi_{\mathrm{ub}}(y).
    $$
    In this case, for any $\rho>0$, the existence of solutions to Eq. \eqref{eq:SufNec_phi_limited} is equivalent to Eq. \eqref{eq:SufNec_phi} with the superlinear constraint on $\left[0,\phi_{\mathrm{ub}}^{-1}(\rho)\right]$. By the definition of the upper-bound solution, there exist
    \begin{itemize}
        \item Case-I: $\rho > P(\alpha)$. There exists no solution to Eq. \eqref{eq:SufNec_phi}.
        \item Case-II: $\rho = P(\alpha)$. There exists a unique solution to Eq. \eqref{eq:SufNec_phi}.
        \item Case-III: $\rho < P(\alpha)$. There exist infinitely many solutions to Eq. \eqref{eq:SufNec_phi}.
    \end{itemize}
    
    Especially, in Case-III, we can construct the solutions by the following approach. Denote $y_0 = \phi_{\mathrm{ub}}(y_0)$, i.e., $\phi_{\mathrm{ub}}$ intersect with $\phi(y)=y$ at $y_0$ (this must happen as $\phi_{\mathrm{ub}}$ is not superlinear for $\alpha< \asigmas$). Then for any sufficiently small $\epsilon>0$, we can construct $\phi_\epsilon$ satisfying $\phi_\epsilon(1)\geq \rho$, where $\phi_\epsilon$ is constructed as the solution to 
    $$
    \left\{
    \begin{aligned}
        &\phi_{\epsilon}^{\prime}(y) = \alpha \cdot F(\phi_\epsilon,y), \quad \forall y\geq 0 \\
        &\phi_\epsilon(y_0-\epsilon) = y_0-\epsilon.
    \end{aligned}
    \right.
    $$
    By Lemma \ref{Lem:LowerBoundforFeasibleSol}, we have $\phi_\epsilon(0)=0$ and $\phi_\epsilon$ is superlinear for $0\leq y\leq y_0 - \epsilon$. Since $\phi_{\mathrm{ub}}(y) > \phi_\epsilon(y)$ on $y\in\left(0,\phi_\epsilon^{-1}(\rho)\right)$, we have 
    $$
    y_0 - \epsilon \leq \phi_\epsilon^{-1}(\rho) < \phi_{\mathrm{ub}}^{-1}(\rho).
    $$
    Therefore, $\phi_\epsilon$ is a solution to Eq. \eqref{eq:SufNec_phi_limited}.
    
    On the other hand, if $\phi_{\mathrm{ub}}(y)> 1$ for some $y\geq 0$, we let $y^* = \phi_{\mathrm{ub}}^{-1}(1)$, i.e., $\phi_{\mathrm{ub}}(y^*) = 1$. Then we have $\tilde{\phi}_{\mathrm{ub}}$ extends $\phi_{\mathrm{ub}}$ from $y^*$ such that $\tilde{\phi}_{\mathrm{ub}}(y^*) =\phi_{\mathrm{ub}}(y^*) = 1$. Then we have $\tilde{\phi}_{\mathrm{ub}}(1) = P(\alpha)$. For any $\rho > P(\alpha)$, if there exist a solution to Eq. \eqref{eq:SufNec_phi_limited}, denoted as $\phi_0$, then we have $\phi_0(1) = \rho > \tilde{\phi}_{\mathrm{ub}}(1)$, which contradicts the definition of $\tilde{\phi}_{\mathrm{ub}}$. If $\rho=P(\alpha)$, $\tilde{\phi}_{\mathrm{ub}}$ is the unique solution. If $\rho<P(\alpha)$, for any sufficiently small $\varepsilon>0$, we can construct $\phi_\varepsilon$ satisfying $\phi_\varepsilon(y^*) = 1 - \varepsilon$, where $\phi_\varepsilon$ is constructed as the solution to
    $$
    \left\{
    \begin{aligned}
        &\phi_{\epsilon}^{\prime}(y) = \alpha \cdot F(\phi_\epsilon,y), \quad \forall y\geq 0 \\
        &\phi_\epsilon(y^*) = 1-\varepsilon.
    \end{aligned}
    \right.
    $$
    By the continuous dependence of solutions on initial values (see Theorem on Continuous Dependence \cite{Walter1998}, p. 145), we can extend $\phi_\varepsilon$ from $y^*$ and obtain a solution to Eq. \eqref{eq:SufNec_phi_limited}, denoted as $\tilde{\phi_\varepsilon}$, such that $\rho < \tilde{\phi_\varepsilon}(1) \leq P(\alpha)$. Thus, we complete the proof of Theorem~\ref{Thm:ExistenceSol_Limited}.

\subsection{Calculate the Optimal Competitive Ratio with Additional Information of $\rho$} \label{Apx: Impact of Addition Information}

In this section, we show how the optimal competitive ratio can be computed when the input $\rho$ is known. Following similar insights from the \OACC\ setting with supply constraints, we leverage the upper-bound solution to explicitly calculate the optimal competitive ratio by exploiting the additional information provided by $\rho$.

Consider Eq.~\eqref{eq:SufNec_phi} with $\alpha < \asigmas$. Since $\asigmas$ has been established as a lower bound, any solution $\phi(y)$ to Eq.~\eqref{eq:SufNec_phi} with $ \phi(0) = 0 $ must eventually violate superlinearity. Specifically, there exists some $ z > 0$ such that $ \phi(z) = z$. From a different perspective, if the reserve function $\phi(y)$ is only required over a limited interval $[0, y_0]$ with $y_0 \leq z$, then the solution $\phi(y)$ remains feasible on this restricted domain. In particular, for any input $\rho \leq \phi(y_0)$, the solution $\phi$ is still $\alpha$-competitive. Hence, for any $\alpha \geq \ataus$, the upper-bound input $P(\alpha)$ (see Definition~\ref{Def:UpperBoundInput}) in the limited-supply setting also applies in this context. 
\begin{definition} \label{Def:UpperBoundInput_unlimited}
  The cost function $f$ is a $(\tau,\sigma)$-polynomial. For any $\alpha\geq 1$, we define the upper bound input $P:[\ataus,+\infty)\rightarrow (0,+\infty]$ as follows: for $\alpha\geq \ataus$,
  $$
  P(\alpha) = \sup_{y\geq 0} \phi_{\mathrm{ub}}(y,\alpha),
  $$
  where $\phi_{\mathrm{ub}}(y,\alpha)$ is the upper-bound solution to Eq. \eqref{eq:SufNec_phi} with parameter $\alpha$. 
\end{definition}
Note that Definition~\ref{Def:UpperBoundInput_unlimited} is slightly different from Definition~\ref{Def:UpperBoundInput}: they rely on the upper-bound solutions to Eq.~\eqref{eq:SufNec_phi} and Eq.~\eqref{eq:SufNec_phi_limited}, respectively.
Similarly, we have $P(\alpha)$ is increasing with respect to $\alpha$, and by Theorem \ref{Thm:UpperBoundAch}, Eq. \eqref{eq:SufNec_phi} is solvable for $\alpha \geq \asigmas$, and thus $P(\alpha) = +\infty$ for $\alpha\geq \alpha^{(\sigma)}_*$.  
Recall that the upper-bound solution always exists for $\alpha\geq \ataus$, and it is a solution to Eq. \eqref{eq:SufNec_phi} where the constraints are satisfied on some bounded interval. 
We can obtain the best competitive ratio for each input $\rho>0$ by Eq. \eqref{eq:BestCRforrho}.
See Figure \ref{fig:23rhoalpha} for an example where the cost function is given by $f(y) = y^3 + y^2$. Since $P(\alpha) = +\infty$ for all $\alpha \geq \asigmas$, the function $\alpha_{\star}(\rho)$ is bounded. This highlights that the unlimited supply case is significantly easier than the limited supply case. Intuitively, the challenge in the limited supply setting arises from the restricted opportunity to make allocation mistakes when competing against the offline optimum, due to the constraint on available supply.

\section{Further Discussions}

\subsection{Dynamical Systems and First Order ODE} \label{Apx:Dynamical Systems and First Order ODE}

In this section, we review some definitions and notations about dynamical systems and breifly explain the relationship between planar dynamical systems and first order ODE. 

We first introduce vector fields and terminologies related to dynamical systems. See Chapter 1 in \cite{Dumortier2006} or Sections 1.1-1.4 in \cite{Chicone2006} for more details.

Let $\Delta$ be an open subset of the Euclidean plane $\mathbb{R}^2$. We define a smooth vector field on $\Delta$ as a $C^{\infty}$ map $X: \Delta \rightarrow \mathbb{R}^2$ where $X(x)$ is meant to represent the free part of a vector attached at the point $x = \left(x^1,x^2\right) \in \Delta$. The graphical representation of a vector field on the plane consists in drawing a number of well chosen vectors $(x, X(x))$. Integrating a vector field means that we look for curves $x(t)$, with $t$ belonging to some interval in $\mathbb{R}$, that are solutions of the differential equation
\begin{equation} \label{eq:IntroDynSys}
  \frac{\mathrm{d}x}{\mathrm{d}t}=X(x),
\end{equation}
where $x \in \Delta$. In the context of dynamical systems, the variable $t$ denotes time.\footnote{We note that the symbol $t$ is also used to index requests in \OACC; this slight abuse of notation should cause no confusion, as its meaning is clear from the context. } Since $X=X(x)$ does not depend on $t$, we say that the differential equation Eq. \eqref{eq:IntroDynSys} is autonomous.

Solutions of Eq. \eqref{eq:IntroDynSys} are differentiable maps $\varphi: I \rightarrow \Delta$ ( $I$ being an interval on which the solution is defined) such that
$$
\frac{d \varphi}{d t}(t)=X(\varphi(t)),
$$
for every $t \in I$. That is, for a solution defining an integral curve, the tangent vector $\frac{d \varphi}{d t}(t)$ at $\varphi(t)$ coincides with the value of the vector field $X$ at the point $\varphi(t)$. 

A point $x \in \Delta$ such that $X(x)=0$ (respectively $\neq 0$ ) is called a singular point (respectively regular point) of $X$. Let $x$ be a singular point of $X$. Then $\varphi(t)=x$, with $-\infty<t<\infty$, is a solution of Eq. \eqref{eq:IntroDynSys}, i.e., $0=\varphi^{\prime}(t)=X(\varphi(t))=X(x)$.

Let $x_0 \in \Delta$ and $\varphi: I \rightarrow \Delta$ be a solution of Eq. \eqref{eq:IntroDynSys} such that $\varphi(0)=x_0$. The solution $\varphi: I \rightarrow \Delta$ is called maximal if for every solution $\psi: J \rightarrow \Delta$ such that $I \subset J$ and $\varphi=\left.\psi\right|_I$ then $I=J$ and, consequently $\varphi=\psi$. In this case we write $I=I_{x_0}$ and call it the maximal interval.

Let $\varphi: I_{x_0} \rightarrow \Delta$ be a maximal solution; it can be regular or constant. Its image $\gamma_{\varphi}=\left\{\varphi(t): t \in I_{x_0}\right\} \subset \Delta$ endowed with the orientation induced by $\varphi$, in case $\varphi$ is regular, is called the \emph{trajectory}, \emph{orbit} or \emph{(maximal) integral curve} associated to the maximal solution $\varphi$. 

We show some basic results about trajectories of a dynamical system:
\begin{itemize}
  \item Given two orbits of $X$ either they coincide or they are disjoint.
  \item A trajectory starting from a non-singular point can never reach a singular point in finite time; it can only approach it asymptotically as $t\rightarrow \pm \infty$.
  \item Any point that a trajectory tends toward as $t\rightarrow \pm \infty$ must be a singular point.
\end{itemize}

By a \emph{phase portrait of the vector field} $X: \Delta \rightarrow \mathbb{R}^2$ we mean the set of (oriented) orbits of $X$. It consists of singularities and regular orbits, oriented according to the maximal solutions describing them, hence in the sense of increasing $t$. In general, the phase portrait is represented by drawing a number of significant orbits, representing the orientation (in case of regular orbits) by arrows. 

Let $p$ be a singular point of a planar smooth vector field $X=(G, H)$. We say that
$$
D X(p)=\begin{bmatrix}
\frac{\partial G}{\partial x}(p) & \frac{\partial G}{\partial y}(p) \\
\frac{\partial H}{\partial x}(p) & \frac{\partial H}{\partial y}(p)
\end{bmatrix}
$$
is the linear part of the vector field $X$ at the singular point $p$.
The singular point $p$ is called \emph{non-degenerate} if 0 is not an eigenvalue.
The singular point $p$ is called \emph{hyperbolic} if the two eigenvalues of $D X(p)$ have real part different from 0 .
The singular point $p$ is called \emph{semi-hyperbolic} if exactly one eigenvalue of $D X(p)$ is equal to 0 . Hyperbolic and semi-hyperbolic singularities are also said to be \emph{elementary singular points}.

We now discuss the relationship between dynamical systems and first order ODE. We consider the following first order ODE
\begin{equation} \label{eq:IntroDynSys_ODE}
  \frac{\mathrm{d}x^2}{\mathrm{d}x^1}=\frac{H(x^1,x^2)}{G(x^1,x^2)}. 
\end{equation}
For $x\in \Delta$ such that $G(x)\neq 0$, ODE Eq. \eqref{eq:IntroDynSys_ODE} only captures the slope of vector arrows in vector field $X$; 
for $x\in \Delta$ such that $G(x)= 0$, ODE Eq. \eqref{eq:IntroDynSys_ODE} is undefined. Specifically, in this paper, we focus on the case where the origin is an isolated zero of $G$ and $H$. 
Thus, the origin is a singularity of the planar system Eq. \eqref{eq:IntroDynSys} and ODE Eq. \eqref{eq:IntroDynSys_ODE} is undefined at the origin. This implies the ODE captures less information than the planar system. In fact, the ODE captures an equivalence class of 2-D systems up to a time-reparametrization, which share the same phase portrait. More precisely, the equivalence class of Eq. \eqref{eq:IntroDynSys} is defined by $\tilde{X}_h(x) = h(x)X(x)$, where $h:\Delta\rightarrow \mathbb{R}_+$. In general, for a domain where $G \neq 0$, a solution to Eq. \eqref{eq:IntroDynSys_ODE} is a projection of a trajectory of the planar system onto the $(x^1,x^2)$-plane. For any nontrivial trajectory of the planar system, we may recover a solution to the slope ODE by a piecewise time-reparametrization. 

\subsection{Illustrating that $\phi_{\mathrm{ub}}$ and $\phi_{\mathrm{lb}}$ are Solutions to Eq.~\eqref{eq:SufNec_phi}} \label{Apx:Discussion on UpperLowerBoundSol}

In this section, we want to briefly show why the upper- and lower-bound solutions $\phi_{\mathrm{ub}}$ and $\phi_{\mathrm{lb}}$ are also solutions to Eq. \eqref{eq:SufNec_phi} for $\alpha \geq \asigmas$. First of all, since solutions to Eq. \eqref{eq:SufNec_phi} are upper and lower bounded by some linear functions, the limits ($\sup$ and $\inf$) used to define $\phi_{\mathrm{ub}}$ and $\phi_{\mathrm{lb}}$ output positive real numbers for any $y\geq 0$. Thus, $\phi_{\mathrm{ub}}$ and $\phi_{\mathrm{lb}}$ are well defined. We then wish that the solutions to Eq. \eqref{eq:SufNec_phi}, used in the limits, converge uniformly to $\phi_{\mathrm{ub}}$ and $\phi_{\mathrm{lb}}$. 

We first introduce why we need the concept of uniform convergence. The fact that $\phi_{\mathrm{ub}}$ and $\phi_{\mathrm{lb}}$ are well defined only tells us those solutions point-wisely converge to $\phi_{\mathrm{ub}}$ and $\phi_{\mathrm{lb}}$, but we are not sure the limit functions $\phi_{\mathrm{ub}}$ and $\phi_{\mathrm{lb}}$ themselves are differentiable (we do not even know whether they are continuous). The uniform convergence, on the other hand, guarantees us better properties of the limit functions. And most importantly, it allows us to interchange the operations of limit process and differentiation. In our case, for example, if $\sup \phi$ converge uniformly to $\phi_{\mathrm{ub}}$, we then have
$$
\begin{aligned}
    \left(\phi_{\mathrm{ub}}\right)^{\prime} &= \left(\sup \phi\right)^{\prime} = \sup \phi^{\prime} \\
    &= \sup \cdot \left(\asigmas F(\phi,y)\right) \\
    &= \asigmas \cdot F(\sup\phi,y) \\
    &= \asigmas \cdot F(\phi_{\mathrm{ub}},y).
\end{aligned}
$$
For details on the concept of uniform convergence, we recommend Chapter 7 in \cite{Rudin2008}.

To prove $\sup \phi$ converge uniformly to $\phi_{\mathrm{ub}}$ on any bounded interval, one can follow a similar manner in the proof of Maximal and Minimal Solutions Theorem in \cite{Walter1998}, p. 93. One may notice that this proof requires $F(\phi,y)$ in Eq. \eqref{eq:SufNec_phi} to be continuous in a domain containing 0, but in our problem, $F(\phi,y)$ is discontinuous at 0. 
In fact, this does not affect us. The continuity in the proof given by \cite{Walter1998} is used to guarantee the prerequisite of the Peano Existence Theorem (see \cite{Walter1998}, p.~73), which helps us find a solution goes through 0. However, we have found such solutions in the proof of Theorem \ref{Thm:UpperBoundAch} (see Appendix \ref{Apx:UpperBound}). One can use the solutions and follow the techniques given by \cite{Walter1998} and prove the uniform convergence.

\subsection{Illustration for $(2,\sigma)$-polynomial Cost Functions} \label{Linearized System}
For $(2,\sigma)$-polynomial cost function, qualitative analysis for the dynamical system Eq. \eqref{eq:MainDynSystem} will be much easier by utilizing its linearized system for which we can obtain the explicit form of the solutions. In this section, we show that for cost functions with $\tau = 2$, there exists a unique solution to Eq. \eqref{eq:SufNec_phi} with an initial slope of $\chi^{(\tau)}_+$, and therefore, it becomes the upper-bound solution. We also illustrate why for competitive ratio $\alpha < \ataus$ one cannot find a feasible solution. 

We let 
$$
f(y) = \sum_{k=2}^{\sigma} c_ky^k, \quad c_2,c_{\sigma}>0 \text{ and } c_k\geq 0, \forall k>2, 
$$
and the linearized system of Eq. \eqref{eq:MainDynSystem} is given by
\begin{equation}
  \begin{bmatrix}
    \frac{\mathrm{d}\phi}{\mathrm{d}t} \\
    \frac{\mathrm{d}y}{\mathrm{d}t}
  \end{bmatrix}
  = 
  \begin{bmatrix}
    \alpha\cdot 2c_2 & -\alpha\cdot 2c_2 \\
    2c_2 & 0 
  \end{bmatrix}
  \begin{bmatrix}
    \phi \\
    y
  \end{bmatrix},\label{eq:LinDynSystem}
\end{equation}
simply denoted as $\dot{X} = AX$. For $\alpha > \ataus = 4$, matrix $A$ has two real eigenvalues: 
$$
\lambda_1 = c_2\left(\alpha + \sqrt{\alpha^2-4\alpha}\right)>0 \quad \text{and} \quad \lambda_2 = c_2\left(\alpha - \sqrt{\alpha^2-4\alpha}\right)>0.
$$
and $Q^{-1}AQ = \operatorname{diag}\{\lambda_1,\lambda_2\}$, where 
$$
Q = \begin{bmatrix}
  1 & 1 \\
  \frac{\alpha - \sqrt{\alpha^2-4\alpha}}{2\alpha} & \frac{\alpha + \sqrt{\alpha^2-4\alpha}}{2\alpha}
\end{bmatrix}. 
$$
Then we can obtain the solutions to the linear system Eq. \eqref{eq:LinDynSystem}:
\begin{equation}
  \begin{bmatrix}
    \phi \\
    y
  \end{bmatrix}
  =
  \begin{bmatrix}
    C_1e^{\lambda_1 t} + C_2e^{\lambda_2 t} \\
    \frac{1}{2\alpha c_2}\left(C_1 \lambda_2 e^{\lambda_1 t} + C_2 \lambda_1 e^{\lambda_2 t}\right)
  \end{bmatrix},
  \quad t\in \mathbb{R}, \label{eq:LinDynSystem_Sol}
\end{equation}
where $C_1,C_2$ are arbitrary constants.\footnote{If system Eq. \eqref{eq:MainDynSystem} is indeed linear, i.e., $\tau = \sigma = 2$, we can obtain that the feasible solutions are given by Eq. \eqref{eq:LinDynSystem_Sol} and $C_1,C_2\geq 0$. Constraints $C_1,C_2\geq 0$ are obtained from Eq. \eqref{eq:SufNec_phi_superlin}: solving $\phi(t) \geq y(t)$ for $t\in \mathbb{R}$, we get $C_1,C_2\geq 0$, which guarantees $\phi$ to be superlinear and thus feasible.} 

Here we discuss what we can know about the solutions to the nonlinear system Eq. \eqref{eq:MainDynSystem} from the solutions Eq. \eqref{eq:MainDynSystem}. Consider the initial slope $\frac{\mathrm{d}\phi}{\mathrm{dy}}\left(0\right)$ of a solution $\phi\geq 0$ to Eq. \eqref{eq:LinDynSystem_Sol}. There is a unique solution with an initial slope of $\chi_+$ when $C_2 = 0$, and infinite many solutions with an initial slope of $\chi_-$ when $C_2>0$. We now show that this is also true for the solutions to the nonlinear system Eq. \eqref{eq:MainDynSystem}. 

\begin{figure}
  \centering
  \begin{subfigure}{0.45\textwidth}
    \includegraphics[width=1\linewidth]{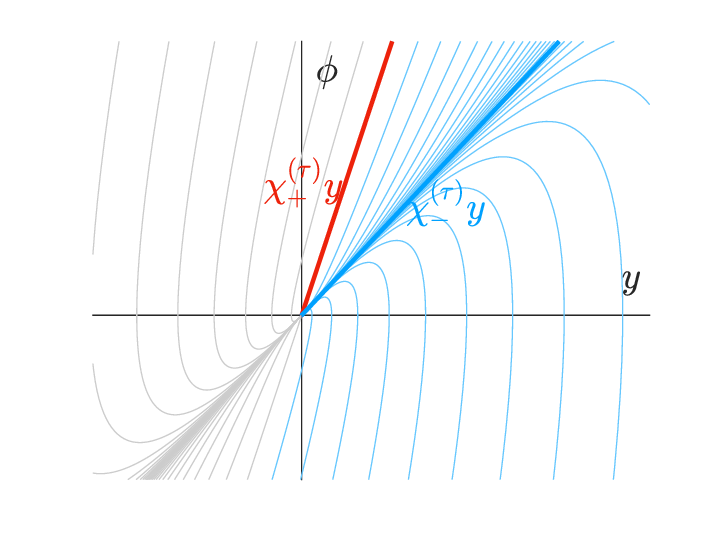}
    \caption{$\alpha > \ataus$} \label{fig:PhDiaLinSys}
  \end{subfigure}
  ~
  \begin{subfigure}{0.45\textwidth}
    \includegraphics[width=1\linewidth]{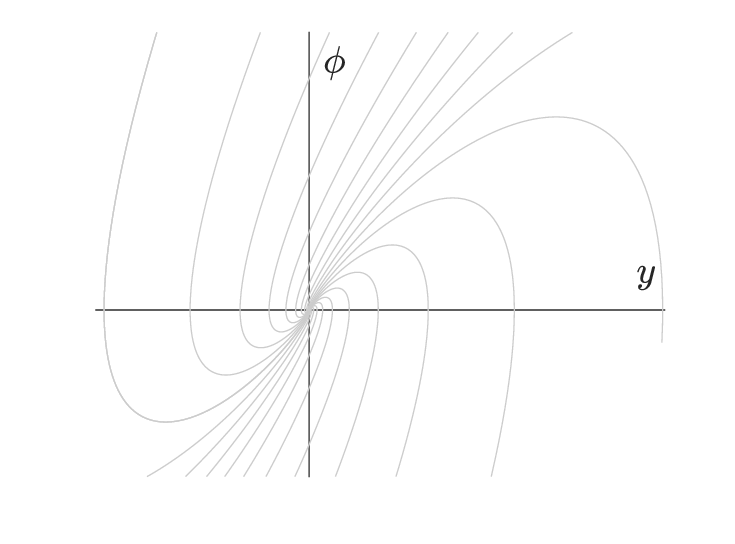}
    \caption{$\alpha < \ataus$} \label{fig:PhDiaLinSys_Smallalpha}
  \end{subfigure}
  \caption{Illustration that $\alpha\geq \ataus$ is necessary. Colored trajectories are tangent to corresponding linear functions at the origin.}
\end{figure}

Since $\lambda_1$ and $\lambda_2$ have non-zero real parts, the origin is the hyperbolic equilibrium point of the nonlinear system Eq. \eqref{eq:MainDynSystem}. In this case, the nonlinear system is locally topologically conjugate to the linear system Eq. \eqref{eq:LinDynSystem} near the equilibrium point by Hartman-Grobman Theorem (see Theorem 7.14 \citep{Chicone2006}), i.e., the phase portrait of the nonlinear system Eq. \eqref{eq:MainDynSystem} near the origin is the same as the phase portrait of the linear system Eq. \eqref{eq:LinDynSystem}, up to a continuous change of coordinates (i.e., a homeomorphism). Furthermore, \citet{Guysinsky2003} has shown that the homeomorphism is differentiable at the equilibrium point, and thus, a solution to the linear system Eq. \eqref{eq:LinDynSystem} and the corresponding solution (under the homeomorphism) to the nonlinear system Eq. \eqref{eq:MainDynSystem} share the same slope at the origin. Therefore, there exist a unique solution with an initial slope $\chi_+$ to the nonlinear system Eq. \eqref{eq:MainDynSystem}. 
\begin{remark}
  Composing $f^{\prime}$ with this unique reserve function, we get the inverse function of the utilization function we constructed when proving Theorem \ref{unlimited_supply_SufNec} (see Proof of Theorem \ref{thm:necessity_IVP}):
  $$
  \Psi(p) = \inf\{\psi(p)\mid \psi \text{ satisfies Eq. } \eqref{ineq:SimNecspsi}\}. 
  $$
\end{remark}

To illustrate why no feasible solution exists when the competitive ratio $\alpha < \ataus$, recall that all solutions to Eq.~\eqref{eq:SufNec_phi} originate from the initial slopes $\chi^{(\tau)}_{\pm}$ at the origin. To obtain a solution to Eq.~\eqref{eq:SufNec_phi}, there must exist corresponding solutions to its linearized system that also begin with an initial slope of $\chi^{(\tau)}_{\pm}$. This holds for $\alpha \geq \ataus$, as illustrated in Figure~\ref{fig:PhDiaLinSys}. In contrast, for $\alpha < \ataus$, Figure~\ref{fig:PhDiaLinSys_Smallalpha} shows that no such solution exists, since all trajectories spiral into the origin without settling on a specific slope. Furthermore, as we discuss in Section~\ref{(Hard) Supply Constraints}, for $\alpha \geq \ataus$ it is possible to find solutions to Eq.~\eqref{eq:SufNec_phi} under a relaxed superlinear constraint, where Eq.~\eqref{eq:SufNec_phi_superlin} is required to hold only over a bounded interval. 

\subsection{Re-frame \OMCR \cite{Devanur2012} under \OACC Setting} \label{Apx:CompConcaveReturn}

For ease of exposition, we assume that the return functions of all offline nodes in online matching with concave returns (\OMCR) \cite{Devanur2012} are identical. When converting \OMCR into the \OACC formulation as described in Section~\ref{Sec:Problem Formulation}, the induced cost function takes the form $g’(0)y - g(y)$. Since $g$ is concave, this expression is upper bounded by the linear function $g’(0)y$. In contrast, the \OACC framework considers more stringent cost structures that grow at least linearly, consistent with the assumption $\tau \geq 1$ for $(\tau,\sigma)$-elastic functions. Therefore, the class of cost functions arising from \OMCR is fundamentally different from the class considered in \OACC. We now compare the characterizations of $\alpha$-competitive designs in \OMCR\ and \OACC.

Beyond the distinction in the cost function class, we also compare the two ODE characterizations.
The characterization of $\alpha$-competitive design for \OMCR is given by
\begin{align} \label{eq:omcr_ode}
    \alpha g'(u) = Y'(v)\frac{\mathrm{d}v}{\mathrm{d}u} + g'(v),
\end{align}
where $Y(v) = g(v) - vg'(v)$ and 
$$
v = \argmax_{y\geq 0} \left(g(y) - \lambda y\right)
$$
for some $\lambda > 0$ (in fact, $\lambda$ is the dual variable we used in \OACC).

Since \OMCR is a special case of \OACC (in terms of problem formulation), we write the ODE for \OMCR in the \OACC form:
Let $y=u$, $f(y) = g'(0)y - g(y)$, and $\phi(y) = v$. We have $f'(y) = g'(0) - g'(y)$ and $f''(y) = -g''(y)$. Therefore, Eq. \eqref{eq:omcr_ode} can be rewritten as
\begin{align}
    & \alpha g'(u) = Y'(v)\frac{\mathrm{d}v}{\mathrm{d}u} + g'(v) = -vg''(v)\frac{\mathrm{d}v}{\mathrm{d}u} + g'(v) \notag \\
    \Leftrightarrow\quad & \alpha(-f'(u) + g'(0)) = \left(\phi(y)f''(y) \right)\phi'(u) + g'(0) - f'(\phi(u)) \notag \\
    \Leftrightarrow\quad & \phi'(u)= \frac{f'(\phi(y)) - \alpha f'(y) + (\alpha-1)g'(0)}{\phi(y)f''(\phi(y))} \label{eq:omcr_ode_oacc}.
\end{align}

When $\alpha = 1$, Eq. \eqref{eq:omcr_ode_oacc} coincides exactly with our ODE Eq. \eqref{eq:mainODE}. 
This equivalence stems from the fact that both frameworks share a common degenerate limit: with a single offline node and linear costs $f(y) = cy$ ($c > 0$), \OACC reduces to reward maximization with margin-adjusted bids $(v_t - c \cdot w_t)_+$. Symmetrically, setting $g(y) = y$ in \OMCR with these adjusted bids induces an equivalent instance. In this regime, both $\PUM\phi$ and the design of \cite{Devanur2012} collapse to the greedy allocation policy, each attaining the ideal competitive ratio of $1$.

If we consider the sufficient conditions for an online algorithm being $\alpha$-competitive in the \OACC and \OMCR settings, we can replace the equalities in Eq. \eqref{eq:mainODE} and Eq. \eqref{eq:omcr_ode_oacc} by ``$\leq$''. Therefore, the sufficient condition for satisfying both \OMCR and \OACC is given by
\begin{align}
    \phi'(y) 
    & \leq \min\left\{\frac{f'(\phi(y)) - \alpha f'(y) + (\alpha-1)g'(0)}{\phi(y)f''(\phi(y))},  \alpha\cdot \frac{f'(\phi(y))  - f'(y) }{\phi(y)\cdot f''(\phi(y))}\right\} \notag \\
    & \leq \alpha\cdot \frac{f'(\phi(y))  - f'(y) }{\phi(y)\cdot f''(\phi(y))}, \label{eq:oacc_and_omcr}
\end{align}
where Eq. \eqref{eq:oacc_and_omcr} is obtained by
\begin{align*}
    & f'(\phi) \leq g'(0) \\
    \Leftrightarrow \quad & (\alpha - 1) f'(\phi) \leq (\alpha - 1) g'(0) \\
    \Leftrightarrow \quad & f'(\phi) + (\alpha-1)g'(0) \geq \alpha f'(\phi). 
\end{align*}
Therefore, by a similar argument to Proposition~\ref{Prop:SolSpace}, one can show the design space of \OACC is contained in the design space of \OMCR. In other words, if Eq.~\eqref{eq:SufNec_phi} exists a solution for cost functions in the form of $g'(0)y - g(y)$, then our algorithm also achieves the optimal competitive ratio for \OMCR.

\section{Empirical Evaluation of Optimal Design Space}
\label{Apx:empirical_evaluation}

In this section, we present a case study that applies the proposed \OACC\ model and algorithms to online cloud resource allocation, using Google’s cluster-usage traces~\cite{reiss2012heterogeneity}.

\subsection*{Experiment 1: Empirical Performance Comparison of Different Optimal Designs} \label{sec:Single Server}

In our first experiment, we consider a single server setting (i.e., $ m = 1 $) and compare the empirical performances of the proposed $ \PUM\phi $ using four different reserve functions: \textsf{UBS} uses the upper-bound solution  $ \phi_{\mathrm{ub}} $, \textsf{LBS} uses the lower-bound solution $ \phi_{\mathrm{lb}} $, \textsf{Linear} uses the linear reserve function $\Dstar y$, and \textsf{MIX} uses the mix-dynamic-static construction based on $ \phi_{\mathrm{ub}} $ and $ \phi_{\mathrm{lb}} $ with turning points $p_1, p_2$ (see Section \ref{sec:Solution Space}). We note that 
\textsf{Linear} is the design of prior work \cite{Tan2020a} and \cite{Huang2019}.  
The upper- and lower-bound solutions,
$ \phi_{\mathrm{ub}} $ and $ \phi_{\mathrm{lb}}$, are numerically computed using the methods in Proposition \ref{Prop:AprUpSol} (set $\eta=10^{-9}$) and Proposition \ref{Prop:LowerSol} (set $\xi = 10^9$). Recall that all four reserve functions can ensure $ \PUM\phi $ to achieve the same optimal competitive ratio in theory, while they exhibit different behaviors regarding their different levels of aggressiveness (as discussed in Section \ref{sec:existence_design_space}). 

\begin{figure*}[htb]
  \centering
  \captionsetup[subfigure]{justification=centering}
  \begin{subfigure}{0.9\textwidth}
    \includegraphics[width=0.43\linewidth]{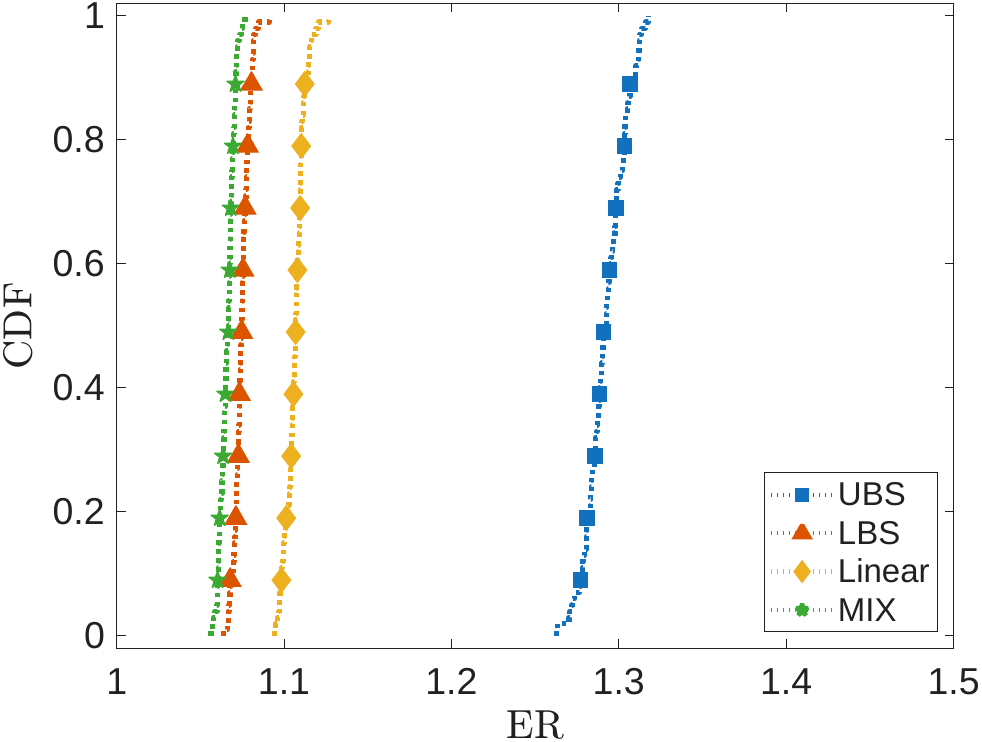}
    \hfill
    \includegraphics[width=0.43\linewidth]{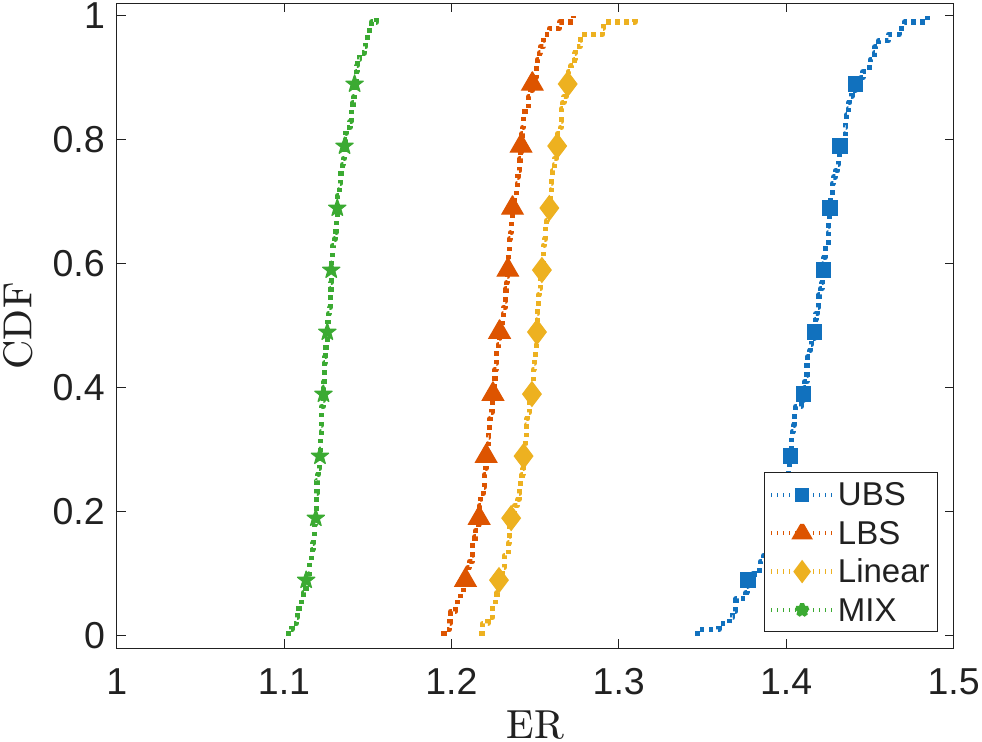}
    \caption{ER under the Two Value Distributions \\ Left: Single-Normal, Right: Mixture-of-Normals}
    \label{fig:Unlim_SingK_CR}
  \end{subfigure}
  
  \vspace{+0.4cm}
  
  \begin{subfigure}{0.9\textwidth}
    \includegraphics[width=0.43\linewidth]{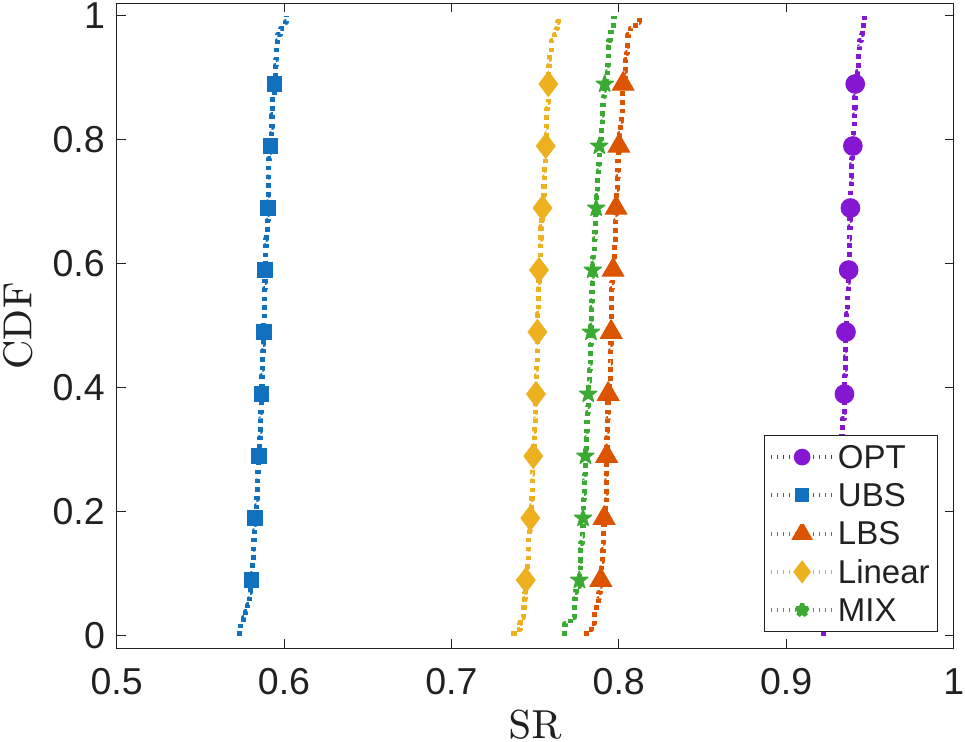}
    \hfill
    \includegraphics[width=0.43\linewidth]{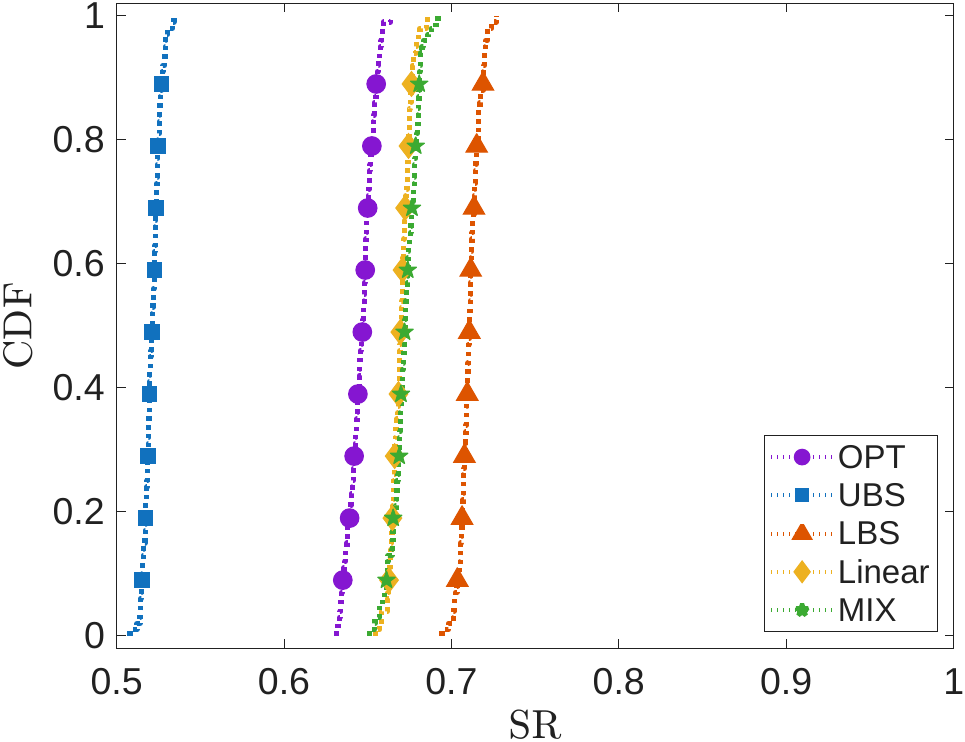}
    \caption{SR under the Two Value Distributions \\ Left: Single-Normal, Right: Mixture-of-Normals}
    \label{fig:Unlim_SingK_SR}
  \end{subfigure}  
  
  \caption{CDF plots of empirical competitive ratios (ER) and served ratios (SR) for the four designs: the upper-bound solution design (\textsf{UBS}), the lower-bound solution design (\textsf{LBS}), the linear design (\textsf{Linear}), and the mixed design of \textsf{UBS} and \textsf{LBS} (\textsf{MIX}). Recall that all these designs share the same competitive ratio of $\asigmas = 3^{3/2}$.}
  
\end{figure*}

We set the cost of allocating $y$ units of CPU as $f(y) = c_1 y^{\tau} + c_2 y^{\sigma}$, which is consistent with the widely adopted \textit{dynamic voltage frequency scaling} mechanism in modeling power consumptions of computing servers \cite{kim2011power,Zhang2015}. We set $c_1 = 3.24,\sigma = 3$, corresponding to the dynamic power consumption, and $c_2 = 10.3, \tau = 2.4$ to account for other costs at the system level, such as cooling. 
Throughout our experiments, we sample $T = 1500$ tasks to simulate online arrivals of demand nodes in \OACC. For each task $t$, we extract the \emph{timestamp}, the scheduling \emph{priority} $p_t$ (which is linearly mapped to $[1,2]$), and the \emph{requested CPU resource} $w_t$. The $T$ tasks are arranged in ascending order of their timestamps. The value of task $t$ is generated as $v_t = r_t p_t \frac{w_t}{\bar{w}}$, where $W = \sum_{t\in[T]} w_t$ denotes the total requested CPU and $\bar{w} =W/T$ denotes the average requested CPU. We randomly generate $r_t$ from $[0,100]$ based on the following two distributions: (i) \textit{Single-Normal}, in which we use the (truncated) normal of $\mathcal{N}(50,10^2)$, and (ii) \textit{Mixture-of-Normals}, where we use a mixture of (truncated) normal distributions as follows: we divide the tasks into four groups per 350 tasks; for the $i$-th group, $r_t$ is drawn from $\mathcal{N}(25i-12.5,10^2)$, i.e., the expected value of the tasks is increasing. Single-Normal and Mixture-of-Normals simulate the easy and hard instances, respectively.

We evaluate the \textit{empirical competitive ratio} (\textbf{ER}) of each algorithm, defined as the ratio between the offline optimum and the algorithm’s performance on a given instance. Each algorithm is tested on 100 independently generated instances, and the cumulative distribution functions (CDFs) of their ERs are shown in Figure~\ref{fig:Unlim_SingK_CR}. Figure~\ref{fig:Unlim_SingK_SR} also reports the \textit{served ratio} (\textbf{SR}) (i.e., the total allocated CPU divided by the total requested CPU), indicating how aggressively an algorithm allocates resources (with higher SR implying more allocation, possibly to lower-value requests).
Under the simple input distribution (i.e., the Single-Normal setting shown in Figures~\ref{fig:Unlim_SingK_CR} and~\ref{fig:Unlim_SingK_SR}), we observe that all designs except \textsf{UBS} perform competitively relative to the offline optimum. In particular, the performance of \textsf{MIX} and \textsf{LBS} are similar, suggesting that \textsf{LBS}, which is the closest to a greedy policy within the optimal design space, is already effective on its own, and the additional tunable parameter in \textsf{MIX} yields limited improvement in this case.  However, in the more challenging setting that simulates harder instances  (i.e., the Mixture-of-Normal input distribution), all baseline designs (\textsf{UBS}, \textsf{LBS}, and \textsf{Linear}) experience a clear performance degradation. In contrast, \textsf{MIX} remains robust and continues to perform well, provided that its turning points are appropriately tuned. This highlights the importance of characterizing the algorithm’s design space (the focus of this paper) and developing adaptive designs to handle input heterogeneity and adversarial structures.

Figure~\ref{fig:Set1} provides deeper insights into how and why \textsf{MIX} outperforms the other designs. The best-fit value of $p_1$ effectively identifies the utilization level at which \textsf{UBS} enters the third phase of the request sequence. Once this threshold is reached, \textsf{MIX} transitions into a static pricing phase (i.e., it stops increasing the reserve function) to maximize acceptance of high-value requests. In the final (fourth) phase, it resumes price growth but does so as gradually as possible, following the \textsf{LBS} structure. The ``Plateau Phase” shown in Figure~\ref{fig:Set1_t1t2} highlights the subset of requests during which \textsf{MIX} employs this static pricing strategy. Figure~\ref{fig:Set1_PhaseSR} further clarifies the differences in allocation behavior. In the first two phases, \textsf{MIX} and \textsf{UBS} exhibit similar served ratios, reflecting a conservative allocation strategy. Beginning in the third phase, however, \textsf{MIX} becomes more aggressive in accepting requests. Crucially, this early conservatism enables \textsf{MIX} to allocate more resources at lower cost during the final phase—thereby outperforming \textsf{LBS} by accepting more valuable requests overall.

\begin{figure*}[htb]
  \centering
  \captionsetup[subfigure]{justification=centering}
  \begin{subfigure}{0.95\textwidth}
    \label{fig:Set1_ValueRatio}
    \includegraphics[width=0.95\textwidth, left]{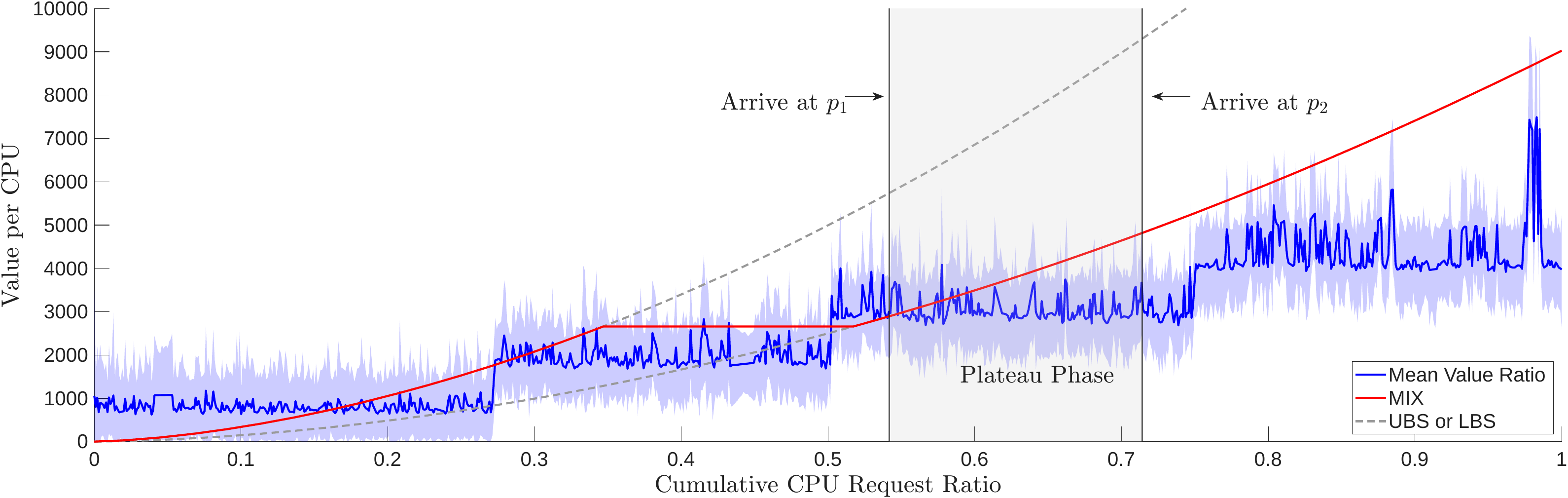}
  \end{subfigure}

  \vspace{+0.2cm}
  
  \begin{subfigure}{0.94\textwidth}
    \hspace{+0.18cm}
    \includegraphics[width=0.936\textwidth]{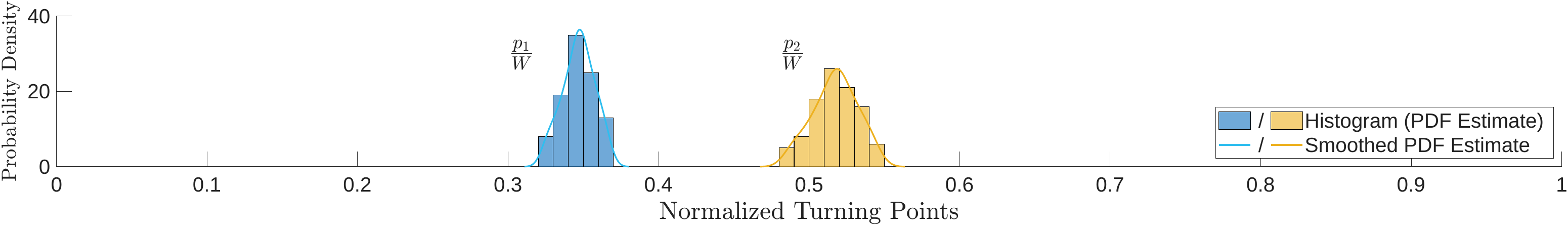}
    \caption{Top: value ratio distribution. Bottom: empirical PDF of $p_1$ and $p_2$.} 
    \label{fig:Set1_t1t2}
  \end{subfigure}

  \vspace{+0.6cm}
  
  \begin{subfigure}{\textwidth}
    \centering
    \includegraphics[width=0.4\linewidth]{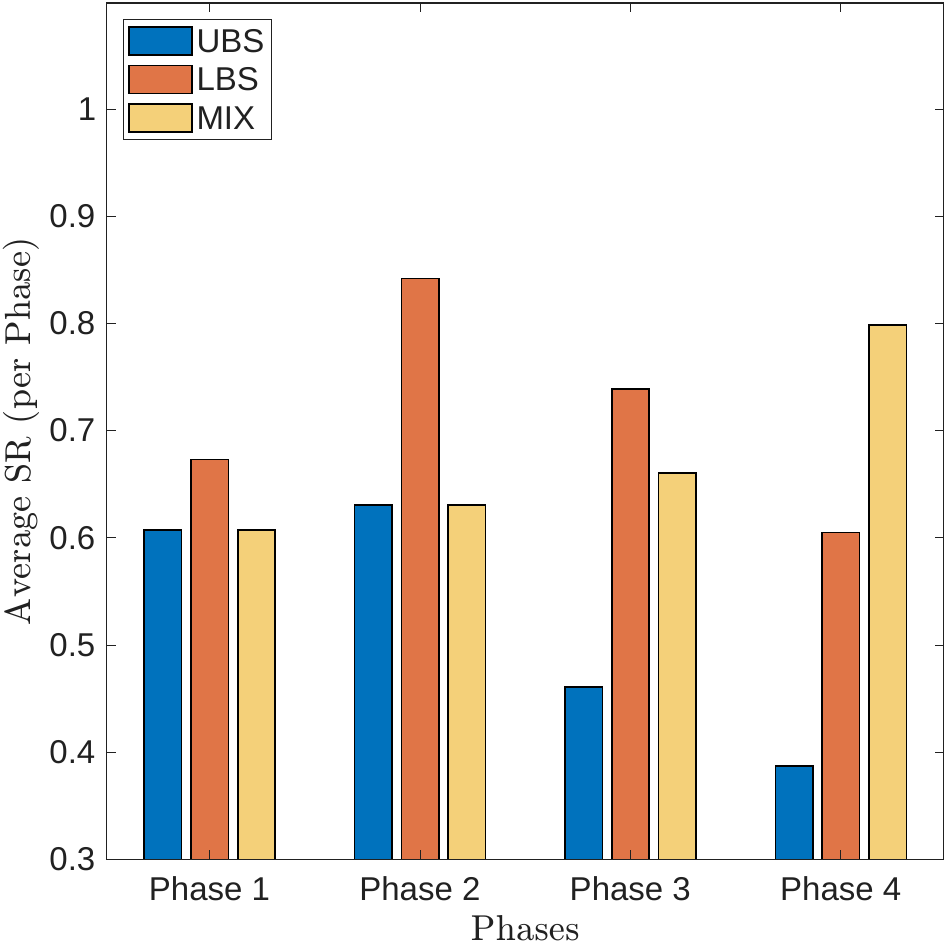}
    \caption{Average SR of UBS, LBS, and MIX in each phase.}
    \label{fig:Set1_PhaseSR}
  \end{subfigure}
  
  \caption{Fig. (a): The top panel compares the pricing function ($f'(\phi)$) of \textsf{MIX} (with $p_1$ set to the mean of best-fit $p_1$ values across 100 instances) against the mean value ratio $\mathbb{E}[v_t / w_t]$. The $x$-axis shows the normalized cumulative CPU ratio $\sum_{i=1}^t w_i / W$, and the blue shaded area marks the 80\% quantile of the empirical value ratio distribution. The bottom panel shows the empirical PDFs of the best-fit turning points $p_1$ and $p_2$ (normalized as $p_1/W$, $p_2/W$) from 100 instances, where $p_1$ is determined via grid search to maximize the empirical competitive ratio (ER). Fig. (b): The bar charts report the mean served ratio (SR) of three reserve function designs, averaged over 100 instances in each phase of the Mixture-of-Normals request sequence.}
  \label{fig:Set1}
  
\end{figure*}

\subsection*{Experiment 2: Performance Evaluation for \OACC with Supply-Oblivious Arrivals}
In our second experiment, we consider a three-server setting ($m = 3$), consisting of Server 1, 2, and 3, each with a distinct cost function. The goal is to compare the empirical performance of the proposed $ \PUM\phi $ algorithm under three different reserve function designs. In the \textsf{Separate} scheme, each server uses its own linear reserve function; in the \textsf{Identical} scheme, all servers adopt the linear reserve function of Server 3; and in the \textsf{Sep-MIX} scheme, Servers 1 and 2 use their respective linear reserve functions, while Server 3 employs the \textsf{MIX} reserve function from Section~\ref{sec:Single Server}, with $p_1 = 5\%\times W$.  
Recall that both \textsf{Separate} and \textsf{Sep-MIX} are guaranteed to achieve the same competitive ratio in theory (see Theorem~\ref{Thm:AverageCRforMK}), whereas \textsf{Identical} is strictly worse in terms of its theoretical competitive ratio.

We consider two settings for modeling the CPU allocation costs of Servers 1, 2, and 3: (i) the \emph{Overlapping Costs} setting, where $f_1(y) = c_1 y^2$, $f_2(y) = f_1(y) + c_2 y^{2.4}$, and $f_3(y) = f_2(y) + c_3 y^3$; and (ii) the \emph{General Costs} setting, where $g_1(y) = c_1 y^2 + 1.5c_1 y^2$, $g_2(y) = 3.5c_2 y^{2.4}$, and $g_3(y) = 1.5c_1 y^2 + 3c_3 y^3$. The cost coefficients are set to $c_1 = 280$, $c_2 = 46$, and $c_3 = 3.9$. For each request $r_t$, the demand is sampled from a truncated normal distribution $\mathcal{N}(50, 10^2)$. To simulate arrival patterns under different system characteristics, we consider two types of server ID distributions: (i) \emph{Supply-Oblivious (SO)}, where each task's server ID is uniformly drawn from $\{1, 2, 3\}$; and (ii) \emph{Decreasingly Supply-Oblivious (DSO)}, where tasks are divided into four equal groups of 350 each. In the $i$-th group, each task is assigned to Server 3 with probability $(i - 1) \times 25\%$ and to $\{1, 2, 3\}$ otherwise—modeling a gradually increasing supply-awareness across task groups.

\begin{figure*}[htb] 
  \centering
  \begin{subfigure}{0.32\textwidth}
    \centering
    \includegraphics[width=\linewidth]{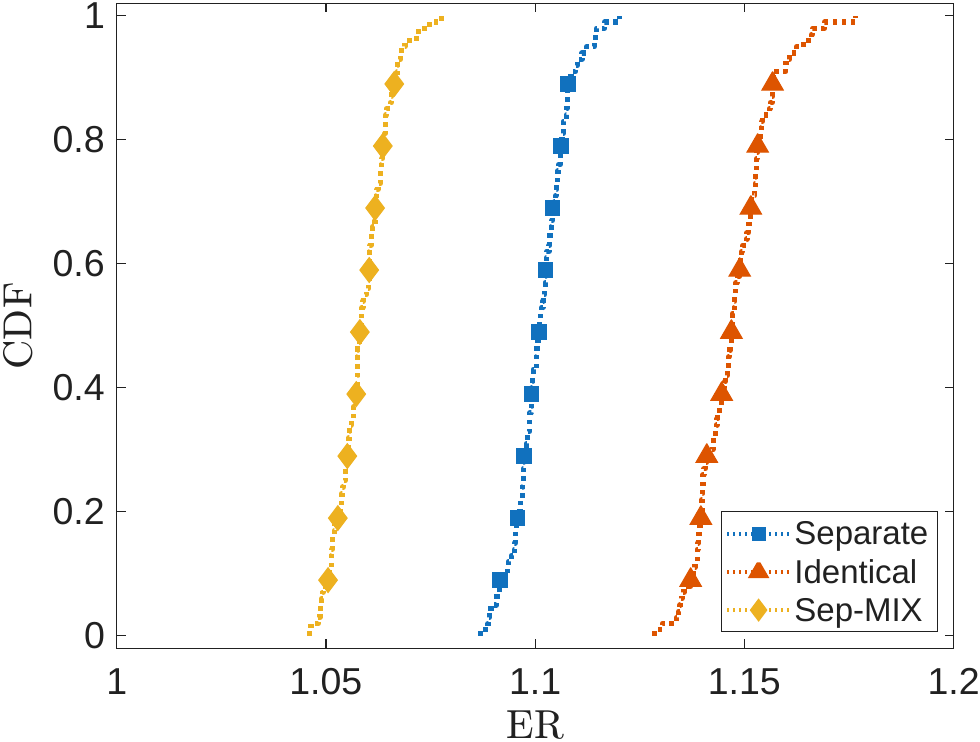}
    \caption{Overlapping Costs with SO}
    \label{fig:Set2_CR1}
  \end{subfigure}
  \begin{subfigure}{0.32\textwidth}
    \centering
    \includegraphics[width=\linewidth]{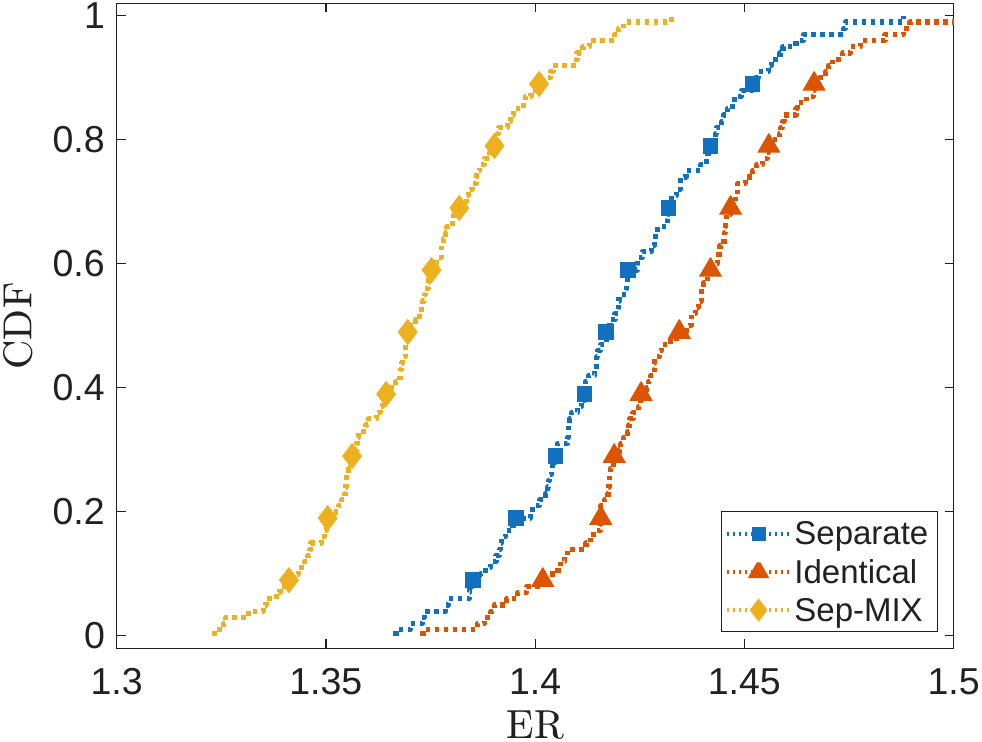}
    \caption{Overlapping Costs with DSO}
    \label{fig:Set2_CR3}
  \end{subfigure}
  \begin{subfigure}{0.32\textwidth}
    \centering
    \includegraphics[width=\linewidth]{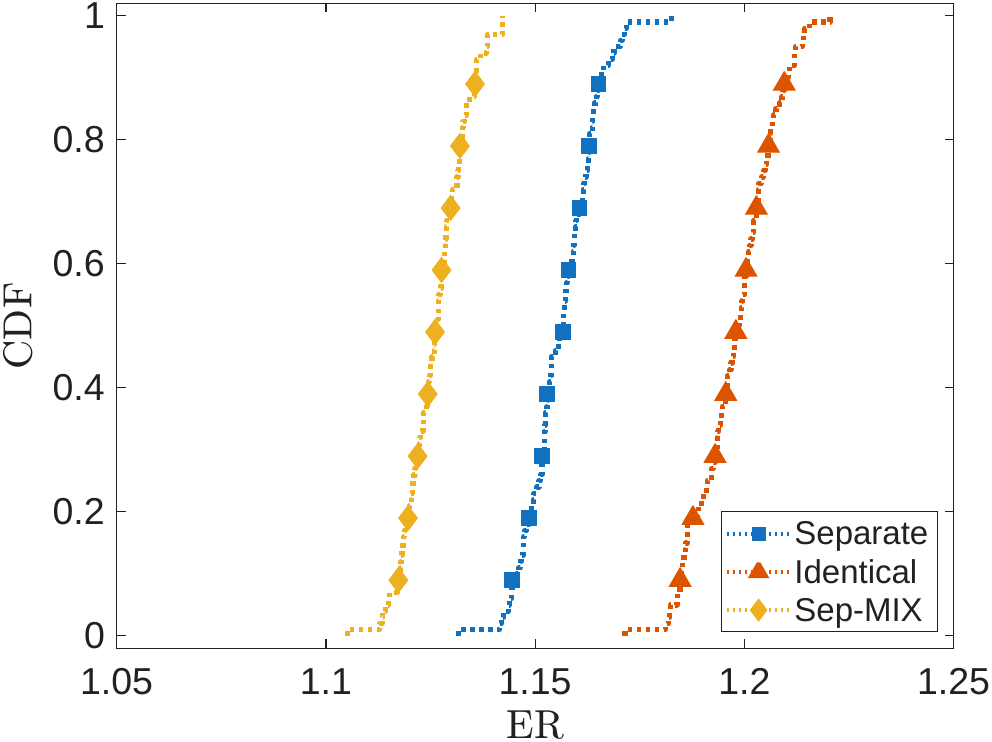}
    \caption{General Costs with SO}
    \label{fig:Set2_CR2}
  \end{subfigure}  
  
  \caption{CDF plots of empirical competitive ratios for the three designs (\textsf{Separate}, \textsf{Identical}, and \textsf{Sep-MIX}) under different cost models (Overlapping Costs and General Costs) and arrival patterns (SO and DSO). }
  
\end{figure*}

We evaluate each algorithm over 100 randomly generated instances and plot the empirical CDFs of their empirical competitive ratios (ERs) in Figures~\ref{fig:Set2_CR1}–\ref{fig:Set2_CR2}. As shown, \textsf{Identical} consistently performs the worst, as it fails to exploit the supply-oblivious structure of arrivals. In contrast, \textsf{Sep-MIX} outperforms \textsf{Separate}, benefiting from the \textsf{LBS}-based reserve function used for Server 3. Figure~\ref{fig:Set2_CR2} further shows that both \textsf{Separate} and \textsf{Sep-MIX} maintain strong performance under general cost structures, even though Theorem~\ref{Thm:AverageCRforMK} only guarantees performance under overlapping costs. Figure~\ref{fig:Set2_CR3} illustrates that, even with unlimited supply, adversarial supply-awareness patterns can significantly degrade performance—but this can be mitigated by the \textsf{MIX} design in \textsf{Sep-MIX}. As seen in Experiment 1, \textsf{MIX}, with suitable turning points, detects phase transitions and adjusts its reserve function accordingly: using \textsf{UBS} early to filter low-value requests and later switching to static or \textsf{LBS}-based pricing to utilize reserved capacity. A similar dynamic unfolds in Figure~\ref{fig:Set2_CR3} for \textsf{Sep-MIX}: in early (less adversarial) phases, \textsf{MIX} behaves conservatively, using \textsf{UBS} to limit resource allocation and preserve supply; in later (more adversarial) phases, it transitions into the “Plateau Phase” or \textsf{LBS}, allowing for more effective allocation due to lower prior utilization.

These results underscore the importance of principled reserve function design (particularly hybrid structures with a mixture of \textsf{UBS} and \textsf{LBS}) in environments with phase-wise dynamics (e.g., batch arrivals), whether induced by shifts in request values (Experiment 1) or supply-obliviousness (Experiment 2). When historical data is available and phase transitions can be estimated via statistical or learning methods, designs like \textsf{MIX} offer a promising path toward robust and high-performing online allocation algorithms.

\end{document}